\documentclass[
onecolumn,
superscriptaddress,
amsmath,amssymb,amsfonts,
aps
]{revtex4-2}

\usepackage{graphicx}
\usepackage{dcolumn}
\usepackage{bm}
\usepackage[colorlinks=true,linkcolor=blue,citecolor=blue,urlcolor=blue]{hyperref}
\usepackage{braket}
\usepackage{comment}
\usepackage{amsthm}
\newtheorem{proposition}{Proposition}
\theoremstyle{definition}

\usepackage{tikz}
\usetikzlibrary{quantikz2}
\usepackage{booktabs}
\usepackage{multirow}
\usepackage{nicematrix}
\usepackage{orcidlink}

\begin{document}

\title{Scalability of quantum error mitigation techniques: from utility to advantage
}

\author{Sergey N. Filippov\,\orcidlink{0000-0001-6414-2137}}

\affiliation{Algorithmiq Ltd, Kanavakatu 3 C, FI-00160 Helsinki, Finland}

\author{Sabrina Maniscalco\,\orcidlink{0000-0001-8559-0828}}

\affiliation{Algorithmiq Ltd, Kanavakatu 3 C, FI-00160 Helsinki, Finland}

\affiliation{QTF Centre of Excellence, Department of Physics, Faculty of Science, University of Helsinki, FI-00014 Helsinki, Finland}

\affiliation{InstituteQ - the Finnish Quantum Institute, University of Helsinki, FI-00014 Helsinki, Finland}

\affiliation{QTF Centre of Excellence, Department of Applied Physics, Aalto University, FI-00076 Aalto, Finland}

\author{Guillermo Garc\'{\i}a-P\'{e}rez\,\orcidlink{0000-0002-9006-060X}}

\affiliation{Algorithmiq Ltd, Kanavakatu 3 C, FI-00160 Helsinki, Finland}

\begin{abstract}
Error mitigation has elevated quantum computing to the scale of hundreds of qubits and tens of layers; however, yet larger scales (deeper circuits) are needed to fully exploit the potential of quantum computing to solve practical problems otherwise intractable. Here we demonstrate three key results that pave the way for the leap from quantum utility to quantum advantage: (1) we present a thorough derivation of random and systematic errors associated to the most advanced error mitigation strategies, including probabilistic error cancellation (PEC), zero noise extrapolation (ZNE) with probabilistic error amplification, and tensor-network error mitigation (TEM); (2) we prove that TEM (i) has the lowest sampling overhead among all three techniques under realistic noise, (ii) is optimal, in the sense that it saturates the universal lower cost bound for error mitigation, and (iii) is therefore the most promising approach to quantum advantage; (3) we propose a concrete notion of practical quantum advantage in terms of the universality of algorithms, stemming from the commercial need for a problem-independent quantum simulation device. We also establish a connection between error mitigation, relying on additional measurements, and error correction, relying on additional qubits, by demonstrating that TEM with a sufficient bond dimension works similarly to an error correcting code of distance 3. We foresee that the interplay and trade-off between the two resources will be the key to a smooth transition between error mitigation and error correction, and hence between near-term and fault-tolerant quantum computers. Meanwhile, we argue that quantum computing with optimal error mitigation, relying on modest classical computer power for tensor network contraction, has the potential to reach larger scales in accurate simulation than classical methods alone.
\end{abstract}

\maketitle

\tableofcontents

\section{\label{section-introduction} Introduction}

Quantum computers nowadays operate with hundreds of qubits and have entered the era of utility~\cite{kim_evidence_2023}, where hardware-enabled estimation is beyond reach for exact brute force state-vector simulations and competes with the approximate classical simulations~\cite{anand2023classical}. The utility era became possible thanks to the hardware improvements and development of noise mitigation techniques that account for errors occurring during the computation. However, noise mitigation is attained at the price of a prolonged wall time for running experiment alterations many times. The figure of merit for the associated computational cost is the sampling overhead $\Gamma$, which is a multiplicative factor in the number of circuit executions needed to get the same estimation accuracy as in the noiseless scenario. The sampling overhead is known to scale exponentially in both the circuit volume and the error rate~\cite{quek-2022,takagi_fundamental_2022,takagi_universal_2023,tsubouchi_universal_2023,filippov_scalable_2023}. This is why some problem-specific state-of-the-art classical simulations still outperform quantum computations~\cite{tindall_efficient_2023,begusic_2023,begusic_fast_2023} and practical quantum advantage in digital quantum simulation (aiming at a superior calculation of some useful system quantities~\cite{hoefler_disentangling_2023,fauseweh_quantum_2024}) is yet to be demonstrated. 

Currently intractable problems of high interest in quantum materials, physics, and chemistry---quantum advantage candidates---are transpiled into circuits with about one hundred qubits and hundreds or even thousands of layers, far beyond the depths up to about 60 probed experimentally so far~\cite{kim_evidence_2023,van_den_berg_probabilistic_2023,obrien_purification-based_2023,bluvstein_logical_2024}. Despite the exponential scaling of the sampling overhead in the circuit volume, error mitigation in circuits of such an enormous scale is not impossible within an affordable wall time---provided the error rate is reasonable. Steady improvements in hardware technology indicate that, at some point, the error rate could be tamed to such an extent that the total number of errors in the circuit $\lesssim 16$ and the associated sampling overhead $\Gamma \sim 10^6$ would enable estimating an observable with $5\%$ error by running $M \sim 4 \times 10^8$ circuit executions, requiring $2$ to $3$ days of quantum computation (provided a single circuit execution with an allowance for the circuit delay takes $\sim 0.5$~ms). 

The key observation in recent studies~\cite{takagi_fundamental_2022,takagi_universal_2023,tsubouchi_universal_2023,filippov_scalable_2023} is that different error mitigation techniques have different factors in their sampling overhead exponents, dramatically impacting the overall computation cost. The goal of this paper is to analyze the sampling overhead, the overall computation time, and the final estimation error in noise-mitigated observables under \textit{realistic} noise for three techniques in use: probabilistic error cancellation (PEC)~\cite{temme-2017,van_den_berg_probabilistic_2023,piveteau-2022}, zero noise extrapolation (ZNE) with probabilistic error amplification~\cite{kim_evidence_2023}, and tensor-network error mitigation (TEM)~\cite{filippov_scalable_2023}. In contrast to noise-agnostic techniques~\cite{rubin-2018,mcardle-2019,koczor-2021,suchsland-2021,czarnik-2021,strikis-2021}, these ones rely on the noise characterization for every essential circuit layer, so that the error mitigation effectively boils down to the inversion of the known noise, suggesting a generally better accuracy~\cite{kim_evidence_2023}. The associated noise-leaning cost does not depend on the error mitigation technique and is shown to be near-constant in the number of qubits for the sparse Pauli-Lindblad (SPL) noise~\cite{van_den_berg_probabilistic_2023,berg_techniques_2023,flammia_2020}. Experiments~\cite{kim_evidence_2023,van_den_berg_probabilistic_2023} confirm that any local noise is converted into the SPL form via Pauli twirling, thus equipping us with a universal and compact noise description with a number of parameters scaling linearly in the number of qubits and bond dimension 4 in the tensor network representation~\cite{filippov_scalable_2023}.

Our study reveals previously unknown random and systematic errors in noise-mitigated observables including ZNE's random error and its associated sampling overhead, which have long remained hidden under the hood of the extrapolation method, as well as the effects of noise stability and noise model accuracy on the systematic bias in the mitigated observable estimation. We also evaluate systematic errors emerging from the limited applicability of exponential extrapolation to a non-exponentially decaying signal in ZNE and the bond dimension compression in TEM. We prove that TEM saturates the recently derived universal lower bound on the sampling overhead~\cite{tsubouchi_universal_2023} for stochastic noise, like SPL, implying the attainability of optimal quantum error mitigation through modest classical computation power, needed for the tensor network contraction. Moreover, we show that TEM with some threshold bond dimension effectively works as a quantum error correcting code of distance $3$ (cf. the recent quantum error correction experiment~\cite{bluvstein_logical_2024}). In this sense, TEM is exemplary in showcasing a smooth transition from near-term to fault-tolerant quantum computation. Since early error-corrected quantum computations cannot afford recursive encoding schemes, a cross-fertilized solution to larger scales could be to mitigate the residual errors of intermediate-scale error correction.

The extensive analysis of errors of different origin enables us to identify candidate circuit sizes for demonstrating quantum advantage with foreseeable quantum computational resources and their quality.  We show that the overall estimation error in noise-mitigated quantum computation can be reasonably small ($\delta \lesssim 10\%$) in dense circuits of size $\sim 100 \times 100$---a challenging scale for uncustomized classical simulations, e.g., with the help of standard tensor networks. Focusing on the \textit{universality} of algorithms (quantum, classical, or hybrid ones), we discuss in Sec.~\ref{section-prospects-advantage} what could be regarded as a useful quantum simulation framework for industrial applications. Computation universality means that different problems from a wide class can be solved in the same algorithmic way. The user just needs to specify parameters of the problem and, once executed, to read the output. Such algorithmic solutions would be of great use in physics, chemistry, and materials science. As an instance of such generic class of problems, we discuss the discrete-time Floquet dynamics that would be classically hard but attainable with near-term hardware. Universal solutions for the discrete-time Floquet dynamics include tensor network simulations as a purely classical method and error-mitigated quantum simulations as a hybrid method requiring quantum power and a modest classical support. A specific problem from the class, the kicked Heisenberg model with an inhomogeneous external field, is instrumental in proving that the generic simulation scheme based on quantum computation and error mitigation can outperform the generic classical methods.

\section{\label{section-results} Main results}

The practical implementation of any noise mitigation technique in large-scale quantum computations is extremely challenging due to the limited resources available. The allocated wall time limits the number $M$ of circuit executions, which in turn translates into the random error $\delta_{\rm random} \propto \frac{1}{\sqrt{M}}$. The relative random error is technique-specific and proportional to the square root of the sampling overhead, $\sqrt{\Gamma}$. In Secs.~\ref{section-pec-random-error}, \ref{section-zne-pea-random-error}, and \ref{section-tem-random-error}, we derive the optimal sampling overhead $\Gamma^{\ast}$ for PEC, ZNE, and TEM, respectively. The results are listed in Table~\ref{table-summary} for convenience (with the notation $N$ for the number of qubits, $L$ for the circuit depth, $\varepsilon$ for the average error rate per qubit per layer, see Sec.~\ref{section-noise-model} for its definition). The exponential dependence of the optimal sampling overhead $\Gamma^{\ast}$ on the total number of errors $\varepsilon N L$ is shown to be different for PEC as compared to ZNE and TEM, becoming a limiting factor in large-circuit quantum computation (cf. $e^{2\varepsilon N L} = 7.9 \times 10^{13}$ and $e^{\varepsilon N L} = 8.9 \times 10^{6}$ for the total number of errors $\varepsilon NL = 16$, see Fig.~\ref{figure-overhead}). We ascribe this to the fact that PEC is observable-agnostic and effectively mitigates errors in the Schr\"{o}dinger picture, whereas ZNE and TEM are tightly bound to the observable measured and can be interpreted as error mitigation in the Heisenberg picture. ZNE and TEM overheads have the same exponents but differ in the prefactor, which scales polynomially in $\varepsilon N L$ for ZNE and is constant for TEM. To attain minimal random errors in ZNE-mitigated observable estimations, the total budget of circuit executions $M$ is to be properly distributed among several experiments with properly chosen noise gain factors $\{G_i^{\ast}\}_i$, which we also list in Table~\ref{table-summary}. In Propositions~\ref{proposition-tem-overhead} and \ref{proposition-tem-overhead-general}, we prove that TEM's sampling overhead is optimal because it asymptotically saturates the universal lower cost bounds~\cite{tsubouchi_universal_2023} for an unbiased estimation of Pauli observables in large-scale quantum circuits both with weak SPL noise and with general weak Pauli noise. 

\begin{table}
    \centering
    \begin{NiceTabular}{|l|c|c|c|}
\midrule
         & PEC & ZNE & TEM \\
\midrule
        Optimal sampling overhead $\Gamma^{\ast}$ & $\exp(2\varepsilon N L)$ & $(1+1.795 \varepsilon N L)^2 \exp(\varepsilon N L)$ & $\exp(\varepsilon N L)$ \\
\midrule
       Optimal noise gain factors & & $G_1^{\ast} = 1$, $G_2^{\ast} = 1 + \dfrac{2.557}{\varepsilon N L}$ & \\
\midrule
       Random error with $M$ shots & $\exp(\varepsilon N L) / \sqrt{M}$ & $(1+1.795 \varepsilon N L) \exp(\varepsilon N L/2) / \sqrt{M}$ & $\exp(\varepsilon N L / 2) / \sqrt{M}$ \\
\midrule
        Systematic error 1: Noise instability & \multirow{2}{*}{$\dfrac{1}{2} \varepsilon N \sqrt{L} \Theta$} & \multirow{2}{*}{$\dfrac{1}{2} \varepsilon N \sqrt{L} \Theta$} & \multirow{2}{*}{$\dfrac{1}{2} \varepsilon N \sqrt{L} \Theta$} \\
         and imprecise noise learning & & & \\
\midrule
         Systematic error 2: Extrapolation & & $\left( 1 + \dfrac{2.557}{\varepsilon N L} \right) \dfrac{NL\varepsilon^2}{36}$ & \\
\midrule
         Systematic error 3: Tensor network & & & \multirow{2}{*}{$\dfrac{NL\varepsilon^2}{30}$} \\
         compression, bond dimension $\chi = \frac{L^2}{2}$ & & &  \\
\midrule
         Systematic error $3'$: Tensor network & & & \multirow{2}{*}{$\sqrt{\frac{ \varepsilon^2 L^2 }{ 72 \ln(4\sqrt{2N}) } \left[ N \big( \frac{1}{32N} \big)^{2\chi / L^2} - \frac{1}{32} \right]}$} \\
        compression, extra term if $\chi < \frac{L^2}{2}$ & & & \\
\midrule
    \end{NiceTabular}
    \caption{Quantitative performance indicators for noise-aware mitigation techniques in typical dense $N \times L$ circuits under realistic (sparse Pauli-Lindblad) noise with the density of errors $\varepsilon$ and standard deviation $\Theta$ in relative temporal fluctuations of the density of errors. Should the intersection ${\cal A}$ of causal cones for the observable and the nontrivial component of the initial state be smaller than the whole circuit area ($|{\cal A}| < NL$), $\varepsilon N L$ is to be replaced by the average number of errors in the causality-covered part of the circuit, $\varepsilon |{\cal A}|$.}
    \label{table-summary}

\end{table}

Other limitations on performance of noise-aware error mitigation techniques originate from a mismatch between the actual noise and its learned model as well as from the method-specific assumptions that can be fulfilled only approximately. All these limitations translate into a collection of the associated systematic errors $\{\delta_{\rm systematic}^{(i)}\}_i$ in estimating the noise-mitigated observable, see Table~\ref{table-summary}. PEC, ZNE, and TEM are all subject to the same systematic error if the actual noise deviates from the SPL form or if the learned SPL parameters differ from the actual ones due to unavoidable temporal fluctuations in the qubit environment, e.g., because of the two-level fluctuators in superconducting quantum computers. Imprecision in single qubit gates and non-Markovian effects~\cite{white_2023} also contribute to the overall noise instability. In Sec.~\ref{section-pec-systematic-error}, we quantify the noise instability in terms of the standard deviation $\Theta$ in relative temporal fluctuations of the density of errors, which can be inferred from the repeated noise characterization experiments. This quantity enables us to evaluate how far the noise can drift from the initially learned parameters. For example, by avoiding the most unstable qubits in a superconducting quantum processor, the noise fluctuations can be reduced down to $\Theta = 1.8 \%$~\cite[Sec. III B in Supplementary Information]{kim_evidence_2023}, which nevertheless leads to a systematic error $\delta_{\rm systematic}^{(1)} = 1.4\%$ ($4.5\%$) in a dense $100 \times 100$ circuit with the density of errors $\varepsilon = 0.16\%$ ($0.5\%$).

\begin{figure}
\includegraphics[width=0.6\textwidth]{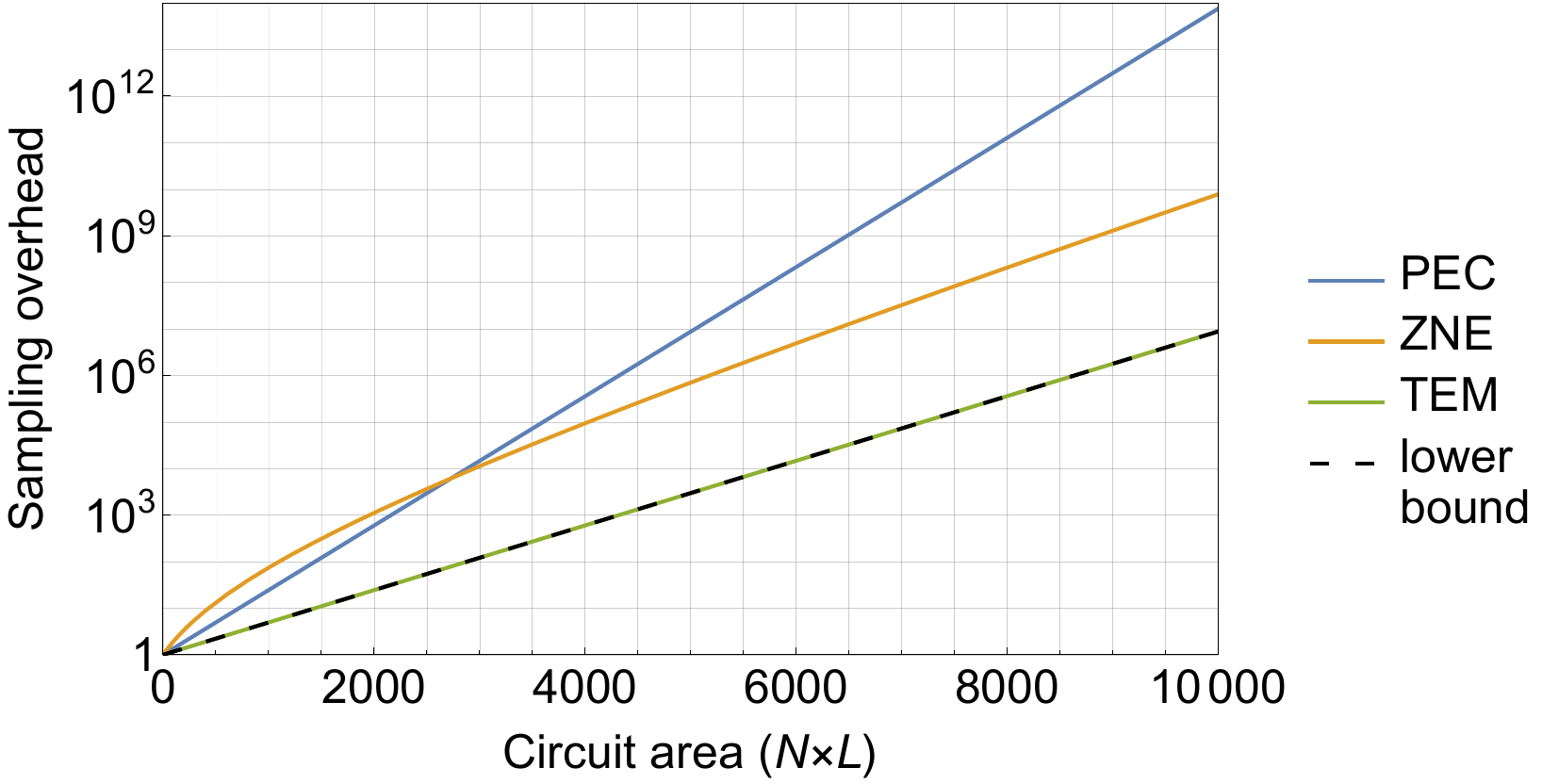}
\caption{\label{figure-overhead} Optimal sampling overhead in estimating high-weight observables in dense $N \times L$ quantum circuits with the foreseeable average error rate per qubit per layer $\varepsilon = 0.16\%$~\cite{mckay_benchmarking_2023}. The lower bound on the sampling overhead is established in Ref.~\cite{tsubouchi_universal_2023}.}
\end{figure}

In large-circuit scenarios, the noisy signal is rather weak, so ZNE resorts to exponential, rather than linear, extrapolation~\cite{kim_evidence_2023}. The exponential extrapolation \textit{per se} induces a systematic error if the signal does not decay exactly exponentially in the noise strength~\cite{cai_multi-exponential_2021}, which is typically the case for non-Clifford circuits, since non-Clifford gates multiply the number of Pauli strings in the Heisenberg-picture evolution of the observable and every Pauli string decays with its own rate. In Sec.~\ref{section-zne-systematic-error}, we evaluate the associated error $\delta_{\rm systematic}^{(2)}$ and specify it for the optimal ZNE performance. In large-scale circuits of high complexity, where the signal decays almost exponentially in the noise strength, $\delta_{\rm systematic}^{(2)}$ grows linearly in the circuit area and quadratically in the error rate. In a dense $100 \times 100$ circuit, $\delta_{\rm systematic}^{(2)} = 0.08\%$ ($0.73\%$) if $\varepsilon = 0.16\%$ ($0.5\%$).

\begin{figure}
\includegraphics[width=0.49\textwidth]{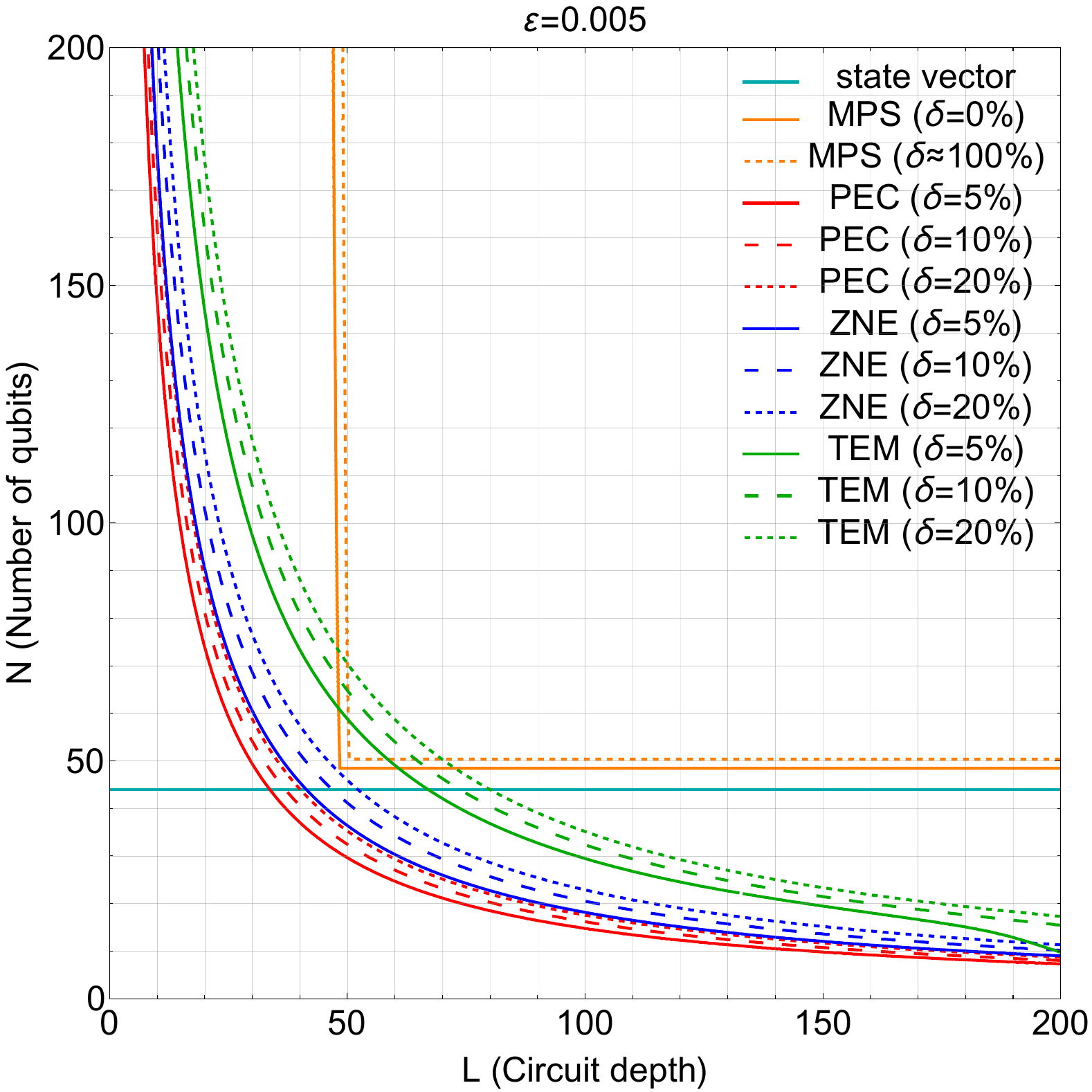} \includegraphics[width=0.49\textwidth]{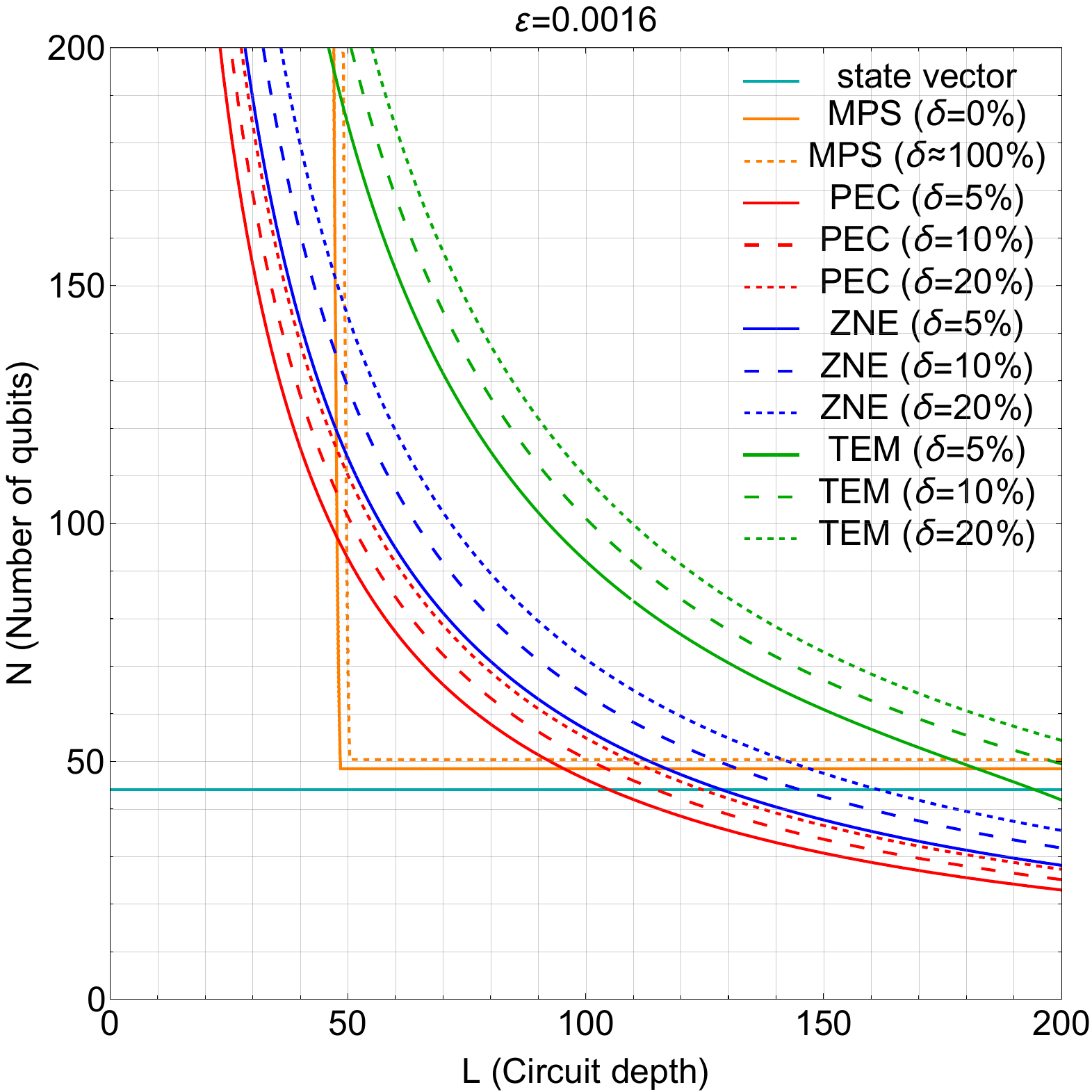}
\caption{\label{figure-summary} Contour diagrams for the typical total error $\delta$ in estimating high-weight observables in \textit{dense} $N \times L$ circuits with average error rate $\varepsilon = 0.5\%$ (left) and $\varepsilon = 0.16\%$ (right). Computation parameters: overall allocated wall time ${\rm T} = 24$~hours, standard deviation in the relative error rate fluctuations $\Theta = 1.8\%$, duration of gates in one layer $\tau_{\rm l} = 0.6$~$\mu$s, measurement duration $\tau_{\rm m} = 0.8$~$\mu$s, and circuit delay time $\tau_{\rm delay} = 0.5$~ms. TEM assumes a classical computational power ${\rm P} = 10^{15}$~FLOPS, the uncustomized matrix product state (MPS) simulation assumes using the most powerful available classical supercomputer (${\rm P} = 1.2 \times 10^{18}$~FLOPS), the state-vector simulation assumes $562 \, 950$~GiB  ($\sim 0.5$~PiB) of memory~\cite{aws_simulating_2022}.}
\end{figure}

TEM effectively modifies the observable by a tensor network so that measuring this modified observable at the output of a noisy quantum computation gives the mitigated estimation of the original observable of interest. Given finite classical computational resources, the tensor network contraction can be implemented within a reasonable wall time only if the bond dimension is kept compressed down to some value $\chi$ after each contraction iteration. As a result of the bond dimension compression, TEM's error mitigating map is found approximately and the mitigated observable estimation contains a systematic error evaluated in Sec.~\ref{section-tem-systematic-error}. As we show in Sec.~\ref{section-tem-systematic-error}, if the bond dimension $\chi$ exceeds some threshold value ($L^2/2$ for Clifford circuits), then the compression-induced systematic error $\delta_{\rm systematic}^{(3)}$ reduces down to the error that would be achieved with quantum error correction with the code distance $3$, i.e., the error rate $\varepsilon$ gets effectively replaced by $\varepsilon^2$. In a dense $100 \times 100$ circuit, $\delta_{\rm systematic}^{(3)} = 0.14\%$ ($1.4\%$) if $\varepsilon = 0.16\%$ ($0.5\%$). If the bond dimension is below the threshold value, then the systematic error is higher equaling $\sqrt{(\delta_{\rm systematic}^{(3)})^2+(\delta_{\rm systematic}^{(3')})^2}$, where the extra term $\delta_{\rm systematic}^{(3')}$ is listed in Table~\ref{table-summary}. 

The total estimation error $\delta$ for a given technique comprises both random and systematic errors and is evaluated as $\delta^2 = \delta_{\rm random}^2 + \sum_{i} (\delta_{\rm systematic}^{(i)})^2$. Equipped with the results of the error analysis above, it is therefore possible to predict what estimation error $\delta$ one should expect in running noisy quantum computations with a particular error mitigation technique under the circumstances of limited time and classical computer support (see Appendices~\ref{appendix-PEC-summary}, \ref{appendix-ZNE-summary}, and \ref{appendix-TEM-summary} for PEC, ZNE, and TEM, respectively). In Fig.~\ref{figure-summary}, we depict contour diagrams for different values of the estimation error $\delta$ to assess the circuit size for which the error-mitigated quantum computation remains reasonably accurate within a reasonable allocated time. As the noise stability and the error rate approach the levels expected from the improved control in current large-scale experiments ($\Theta = 1.8\%$~\cite{kim_evidence_2023} and $\varepsilon = 0.16\%$~\cite{mckay_benchmarking_2023}), a window of opportunity opens for quantum computations to outperform purely classical simulations of $\sim 100 \times 100$ quantum circuits.

As the quantum technology improves and the average error rate $\varepsilon$ decreases, the main limitation on the size $N \times L$ of hardware-simulable circuits comes from the total number of errors $\varepsilon N L$ which cannot exceed $\sim 20$ in order to get a realistic wall time with the sampling overhead $e^{\varepsilon N L} \sim 5 \times 10^8$. This is because the random error depends entirely on the total number of errors $\varepsilon N L$. The simulable circuit size is therefore  $N L \lesssim \frac{20}{\varepsilon}$, with deviations from one mitigation method to another (see Fig.~\ref{figure-summary}). However, the systematic errors $\delta_{\rm systematic}^{(2)}$ and $\delta_{\rm systematic}^{(3)}$ quadratically depend on the error rate and are proportional to $\varepsilon \times (\varepsilon N L) \lesssim 20 \varepsilon$ and thus decrease with the decrease of $\varepsilon$ in the region of simulable scales (mathematically, $\delta_{\rm systematic}^{(2)} \rightarrow 0$ and $\delta_{\rm systematic}^{(3)}  \rightarrow 0$ if $\varepsilon \rightarrow 0$ and $N L \rightarrow \infty$ so that $\varepsilon N L = {\rm conts}$).

As error mitigated quantum computation produces competitive expectation values and continues improving, it is becoming increasingly attractive for solving numerous large-scale simulation problems in an industry-friendly automated way. A digital quantum simulator provides algorithmic solutions to whatever problem from a given a wide class without any need for customization, thus making quantum simulation \textit{universally} applicable within the class. On the other hand, the established but uncustomized classical simulators (e.g., based on conventional tensor network methods like matrix product states) cannot generally provide accurate solutions for all problems within the class as we show in Sec.~\ref{section-advantage-verifiable-model}. Although a problem-specific customization of classical simulation algorithms can lead to better performance indicators in solving a particular problem, such customized solutions cannot be generally applied to other large-scale problems from the same class. Since the recustomization of a classical simulator for each and every problem from the class is costly, quantum simulators become advantageous from a commercial viewpoint, thus giving a new meaning to quantum advantage in industrial applications (see Sec.~\ref{section-prospects-advantage-universality}).

\section{\label{section-analysis} Analysis of error mitigation performance at scale}

\subsection{\label{section-circuits} Overview of techniques and noise model} \label{section-noise-model}

\begin{figure}
\begin{center}

\begin{tabular}{|l|c|}
\hline
~
~PEC~~ & \begin{quantikz}[row sep=0.1cm, column sep=0.1cm]
\gategroup[4,steps=14,style={dashed,rounded
corners,fill=white!20, inner
xsep=15pt,xshift=-14pt,yshift=-5pt},background,label style={label
position=below,anchor=north,yshift=-0.2cm}]{samples from the ideal circuit with the rescaled measurement outcomes} \lstick{$\ket{0}$} & \gate[3,disable auto height,style={fill=orange!20}]{{\cal V}_k^{(1)}} & \gate[3,style={fill=magenta!20}]{{\cal N}_1} & \gate[3,disable auto height,style={fill=green!20}]{{\cal U}_1} & [0.2cm] & \ \ldots \ & [0.2cm] & \gate[3,disable auto height,style={fill=orange!20}]{{\cal V}_k^{(L)}} & \gate[3,style={fill=magenta!20}]{{\cal N}_L} & \gate[3,disable auto height,style={fill=green!20}]{{\cal U}_L} & & \meter{} & \setwiretype{c} & \measure[style={fill=yellow!20}]{\times} & & & \gate[3,disable auto height,style={fill=blue!20}]{O} \\
\lstick[label style={yshift = 3pt, xshift = -5pt}]{$\vdots$} & & & & & \ \ldots \ & & & & & & \meter{} & \setwiretype{c} & \measure[style={fill=yellow!20}]{\times} & & & & \setwiretype{n} \rstick{ $\longrightarrow \bar{O} $} \\
\lstick{$\ket{0}$} & & & & & \ \ldots \ & & & & & & \meter{} & \setwiretype{c} & \measure[style={fill=yellow!20}]{\times} & & & \\
\lstick{$\displaystyle{ \bigcup_{l,k} q_k^{(l)} }$} \setwiretype{c} & \phase[label style = {xshift = -3pt, yshift = -18pt}]{q_k^{(1)}} \wire[u][1]{c} & & & & \ \ldots \ &  & \phase[label style = {xshift = -3pt, yshift = -18pt}]{q_k^{(L)}} \wire[u][1]{c} & & & & \measure[style={fill=yellow!20}]{\gamma \sigma_{{\cal V}}} & & \wire[u][3]{c}
\end{quantikz} \\
\hline 
$\begin{array}{c}
    \text{ZNE} \\
    \text{(PEA)}
\end{array}$ & \begin{quantikz}[row sep=0.1cm, column sep=0.1cm]
\lstick{$\ket{0}$} &  & \gate[3,disable auto height,style={fill=orange!20}]{{\cal V}_k^{(1)}} & \gate[3,style={fill=magenta!20}]{{\cal N}_1} & \gate[3,disable auto height,style={fill=green!20}]{{\cal U}_1} & [0.2cm] & \ \ldots \ & [0.2cm] & \gate[3,disable auto height,style={fill=orange!20}]{{\cal V}_k^{(L)}} & \gate[3,style={fill=magenta!20}]{{\cal N}_L} & \gate[3,disable auto height,style={fill=green!20}]{{\cal U}_L} & & \meter{} & \setwiretype{c} & \gate[3,disable auto height,style={fill=blue!20}]{O} \\
\lstick[label style={yshift = 3pt, xshift = -5pt}]{$\vdots$} & & & & & & \ \ldots \ & & & & & & \meter{} & \setwiretype{c} & & \setwiretype{n} \push{\longrightarrow \bar{O}(G_i)} & \rstick[4]{~~\rotatebox[origin=c]{90}{extrapolation} $\longrightarrow \bar{O}$}\\
\lstick{$\ket{0}$} & & & & & & \ \ldots \ & & & & & & \meter{} & \setwiretype{c} & \\
& \setwiretype{n}  & \phase[label style = {xshift = -4pt, yshift = -18pt}]{p_k^{(1,G_i-1)}\hspace{-0.5cm}} \wire[u][1]{c} \setwiretype{c} & & & &  \ \ldots \ &  & \phase[label style = {xshift = -4pt, yshift = -18pt}]{p_k^{(L,G_i-1)} \hspace{-0.5cm}} \wire[u][1]{c} & \setwiretype{n}  \\
\lstick{$\displaystyle{ \bigcup_{i} G_i }$} \setwiretype{c} & \phase[label style = {xshift = -5pt, yshift = -18pt}]{\cup_{l,k} \ p_k^{(l,G_i-1)} \hspace{-1.5cm}} \wire[u][1]{c} & & & & & &  & & & & & & & & \push{~~~G_i} & \setwiretype{n}
\end{quantikz} \\[-0.4cm]

& \\[0.2cm]

\hline
~~TEM~~ & \begin{quantikz}[row sep=0.1cm, column sep=0.1cm]
\lstick{$\ket{0}$} & \gate[3,style={fill=magenta!20}]{{\cal N}_1}  & \gate[3,disable auto height,style={fill=green!20}]{{\cal U}_1} & [0.2cm] & \ \ldots \ & [0.2cm] & \gate[3,style={fill=magenta!20}]{{\cal N}_L} & \gate[3,disable auto height,style={fill=green!20}]{{\cal U}_L} & & \meter{} & \setwiretype{c} & \gategroup[3,steps=13,style={dashed,rounded
corners,fill=white!20, inner
xsep=-3pt, inner
ysep=-3pt,xshift=0pt,yshift=0pt},background,label style={label
position=below,anchor=north,yshift=-0.2cm}]{tensor network for $\overleftarrow{\bigcirc}_{l=1}^L (\widetilde{\cal N}_l^{-1})^{\dag}$} & \gate[style={fill=cyan!20}]{} \wire[d][2]{c} & & \gate[style={fill=gray!20}]{} \wire[d][2]{c} &  \ \ldots \ & \gate[style={fill=cyan!20}]{} \wire[d][2]{c} & & \gate[style={fill=gray!20}]{} \wire[d][2]{c} & & \gate[style={fill=green!20}]{} \wire[d][2]{c} &  \ \ldots \ & \gate[style={fill=green!20}]{} \wire[d][2]{c} & & \gate[3,disable auto height,style={fill=blue!20}]{O} \\
\lstick[label style={yshift = 3pt, xshift = -5pt}]{$\vdots$} & & & & \ \ldots \ & & & & & \meter{} & \setwiretype{c} & & \gate[style={fill=cyan!20}]{} & & \gate[style={fill=gray!20}]{} &  \ \ldots \ & \gate[style={fill=cyan!20}]{} & & \gate[style={fill=gray!20}]{} & & \gate[style={fill=green!20}]{} &  \ \ldots \ & \gate[style={fill=green!20}]{} & & & \setwiretype{n} \push{\longrightarrow \bar{O}} \\
\lstick{$\ket{0}$} & & & & \ \ldots \ & & & & & \meter{} & \setwiretype{c} & & \gate[style={fill=cyan!20}]{} & & \gate[style={fill=gray!20}]{} &  \ \ldots \ & \gate[style={fill=cyan!20}]{} & & \gate[style={fill=gray!20}]{} & & \gate[style={fill=green!20}]{} &  \ \ldots \ & \gate[style={fill=green!20}]{} & &
\end{quantikz} \\
\hline
\end{tabular}

\end{center}
\caption{\label{figure-techniques} Noise-aware noise mitigation techniques: PEC, ZNE with probabilistic error amplification (PEA), and TEM. PEC operates without regard to the observable (error mitigation in the Schr\"{o}dinger picture). ZNE operates with a specific observable and TEM effectively modifies any observable by the adjacent tensor-network map (mitigation of errors propagated in the Heisenberg picture).}
\end{figure}
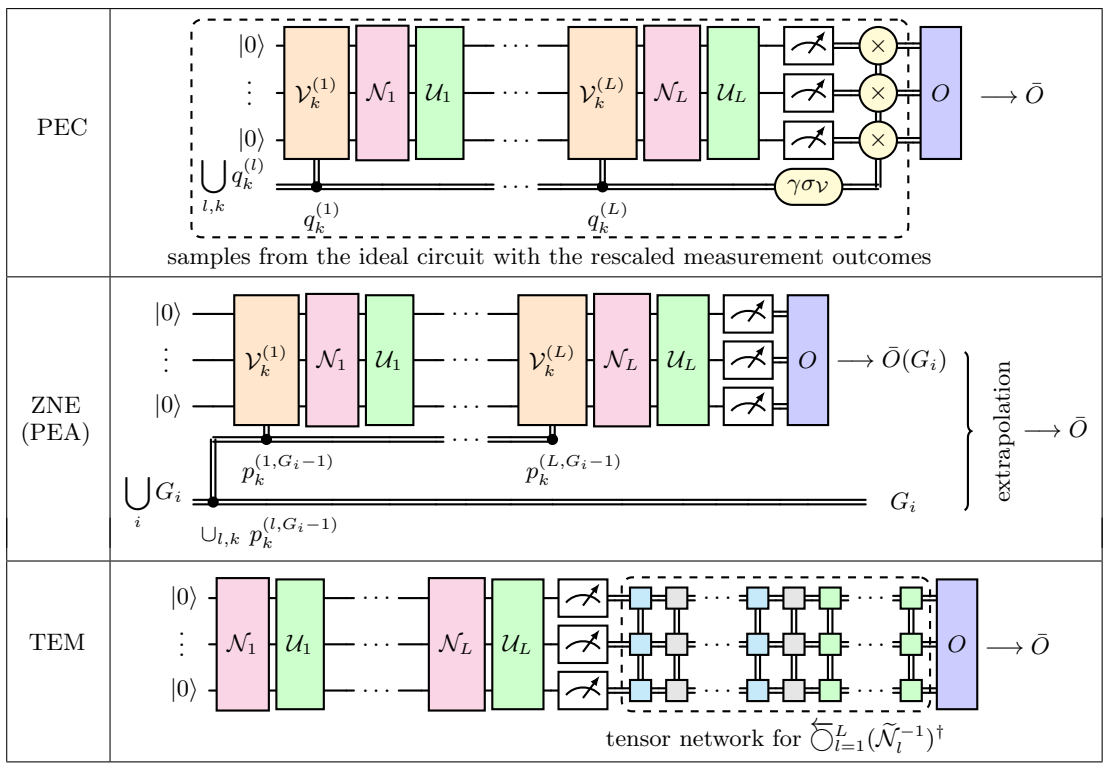

The practical implementation of a digital quantum computation can be described as a sequence of ideal unitary layers ${\cal U}_l$, each preceded by a layer ${\cal N}_l$ of associated unavoidable noise. Noise primarily accompanies entangling gates and typically comprises cross-talk terms involving neighboring qubits, whereas single-qubit unitary operations are almost noiseless (error rate $\lesssim 0.01$\% \cite{kim_evidence_2023}). Assuming a Markovian nature of the noise and its invariance in time, it suffices to characterize and formally invert each noisy layer ${\cal N}_l$ to mitigate the errors. Since the inverse map ${\cal N}_l^{-1}$ is not completely positive, it cannot be directly applied to qubits via control pulses. Different techniques pursue different ways of inverting noise, see Fig.~\ref{figure-techniques}. In PEC, a map ${\cal N}_l^{-1}$ is implemented on average via the application of (presumably noiseless) unitary maps ${\cal V}_{k}^{(l)}$ sampled from the quasiprobability distribution $\{q_{k}^{(l)}\}_k$ in the expansion ${\cal N}_l^{-1} = \sum_k q_{k}^{(l)} {\cal V}_{k}^{(l)}$ (see Sec.~\ref{section-pec}). In ZNE with probabilistic error amplification, the extrapolation to zero-noise value ($G=0$) is made through several noisy estimations ($G=G_1,G_2, \ldots$; $G_i \geq 1$) each obtained on average by sampling (presumably noiseless) unitary gates from the true probability representation for powers ${\cal N}_l^{G_i-1}$ (see Sec.~\ref{section-zne-pea}). In TEM, the observable of interest is concatenated with the tensor-network representation of the noise inversion map in the Heisenberg picture, $\overleftarrow{\bigcirc}_l (\widetilde{\cal N}_l^{-1})^{\dag}$ with $\widetilde{\cal N}_l^{-1} = {\cal U}_{\geq l} \circ {\cal N}_l^{-1} \circ {\cal U}_{\geq l}^{\dag}$ and ${\cal U}_{\geq l} = \overleftarrow{\bigcirc}_{m\geq l} {\cal U}_m$~\footnote{Hereafter, we use symbols $\overleftarrow{\bigcirc}_m$ ($\overrightarrow{\bigcirc}_m$) and $\overleftarrow{\prod}_m$ ($\overrightarrow{\prod}_m$) to unambiguously indicate the increasing order for indices in map concatenations and matrix products:  $\overleftarrow{\bigcirc}_m {\cal A}_m = \cdots \circ {\cal A}_2 \circ {\cal A}_1$ and $\overleftarrow{\prod}_m {A}_m = \cdots {A}_2 \cdot {A}_1$ ($\overrightarrow{\bigcirc}_m {\cal A}_m = {\cal A}_1 \circ {\cal A}_2 \circ \cdots$ and $\overrightarrow{\prod}_m {A}_m = {A}_1 \cdot {A}_2 \cdots $).} (see Sec.~\ref{section-tem}). Given finite quantum and classical computational resources, every technique results in a noise-mitigated observable subject to random errors (due to a sampling overhead) and systematic errors (due to an imprecisely learned noise, fluctuations in noise parameters, errors in single-qubit unitary gates, extrapolation inaccuracy, bond dimension truncation). We evaluate both random and systematic errors for each technique in Secs.~\ref{section-pec}, \ref{section-zne-pea}, \ref{section-tem}.

In experiments with more than 100 qubits~\cite{kim_evidence_2023}, the noise is brought into the sparse Pauli-Lindblad (SPL) form via Pauli twirling and then learned to infer the noise model parameters: $3$ relaxation rates $\lambda_X^{[q]}, \lambda_Y^{[q]}, \lambda_Z^{[q]}$ for each qubit $q$ and $9$ relaxation rates $\lambda_{XX}^{\langle q_1,q_2\rangle}, \lambda_{XY}^{\langle q_1,q_2\rangle}, \lambda_{XZ}^{\langle q_1,q_2\rangle}, \lambda_{YX}^{\langle q_1,q_2\rangle}, \lambda_{YY}^{\langle q_1,q_2\rangle}, \lambda_{YZ}^{\langle q_1,q_2\rangle}, \lambda_{ZX}^{\langle q_1,q_2\rangle}, \lambda_{ZY}^{\langle q_1,q_2\rangle}, \lambda_{ZZ}^{\langle q_1,q_2\rangle}$ for each pair of nearest-neighbor qubits $\langle q_1,q_2\rangle$~\cite{van_den_berg_probabilistic_2023}. The unit of time in defining the rates is the duration of the unitary layer operation. In the presence of cross-talk between nearest qubits only, the number of parameters in the SPL noise model grows linearly in the number of qubits and all these parameters can be estimated via measurements in $9$ different bases only~\cite{van_den_berg_probabilistic_2023}. We therefore stick to the SPL noise model, where ${\cal N}_l = \bigcirc_{q} {\cal N}^{[q]}_{l} \circ \bigcirc_{\langle q_1,q_2\rangle} {\cal N}^{\langle q_1,q_2\rangle}_{l}$ is a concatenation of commuting one- and two-qubit Pauli channels acting nontrivially either at a qubit $q$ or a pair of neighboring qubits $q=\langle q_1,q_2\rangle$. In the Pauli-transfer-matrix representation, ${\cal N}_l$ is given by a diagonal matrix product operator (MPO) of bond dimension $4$ and physical dimension $4$ associated with Pauli matrices $\{I,X,Y,Z\} \equiv \{\sigma_{\alpha}\}_{\alpha = 0}^{3}$~\cite{filippov_scalable_2023}. By ${\rm SPL}({\cal N}_l)$ we denote the set $\{\lambda_{\boldsymbol{\alpha}}^{(l)}\}_{\boldsymbol{\alpha}}$ of relaxation rates for all (one-qubit and two-qubit) jump operators $\sigma_{\boldsymbol{\alpha}} = \bigotimes_{q=0}^{N-1} \sigma_{\alpha_q}$, $\boldsymbol{\alpha} = (\alpha_0,\ldots,\alpha_{N-1})$, so that ${\cal N}_l[\bullet] = \exp[ \sum_{\boldsymbol{\alpha}} \lambda_{\boldsymbol{\alpha}}^{(l)} (\sigma_{\boldsymbol{\alpha}} \bullet \sigma_{\boldsymbol{\alpha}} - \bullet)  ]$. This noise model readily extends to more general scenarios with the jump operators of higher Pauli weight (beyond two-local crosstalk)~\cite{berg_techniques_2023,rouze2023efficient,flammia_2020}, where our results still hold true. For the sake of clarity, though, we consider the two-local SPL noise model as the default one.   

Importantly, any Pauli observable $\sigma_{\boldsymbol{\beta}}$ is an eigenoperator of the SPL noise, ${\cal N}_l[\sigma_{\boldsymbol{\beta}}] = f_{l{\boldsymbol{\beta}}} \sigma_{\boldsymbol{\beta}}$, and the corresponding eigenvalue (fidelity) is sensitive to those of the noise jump operators that anticommute with the observable, 
\begin{eqnarray} 
f_{l{\boldsymbol{\beta}}} &=& \begin{array}[t]{@{}l@{}}
     \displaystyle \prod \\[-0.5ex] \scriptstyle \lambda_{\boldsymbol{\alpha}} \in {\rm SPL}({\cal N}_l): \ \{ \sigma_{\boldsymbol{\alpha}}, \sigma_{\boldsymbol{\beta}} \} = 0
    \end{array} \exp(-2\lambda_{\boldsymbol{\alpha}}) \nonumber\\
&=& \prod_{q} \begin{array}[t]{@{}l@{}}
     \displaystyle \prod \\[-0.5ex] \scriptstyle \alpha_q: \ \{\sigma_{\alpha_q},\sigma_{\beta_q}\} = 0
    \end{array} \exp(-2 \lambda_{\alpha_q}^{[q](l)}) \prod_{\langle q_1,q_2\rangle} \begin{array}[t]{@{}l@{}}
     \displaystyle \prod \\[-0.5ex] \scriptstyle \alpha_{q_1},\alpha_{q_2}: \ \{\sigma_{\alpha_{q_1}} \otimes \sigma_{\alpha_{q_2}}, \sigma_{\beta_{q_1}} \otimes \sigma_{\beta_{q_2}} \} = 0
    \end{array} \exp(-2 \lambda_{\alpha_{q_1} \alpha_{q_2}}^{\langle q_1,q_2\rangle (l)}). \label{Pauli-fidelity}
\end{eqnarray}

\noindent Different Pauli strings gather different $\lambda$-components from the noise model, leading to an exponential range of fidelities: low-weight Pauli strings are affected the least whereas high-weight Pauli strings are affected the most. In a general $N$-qubit quantum circuit, the dynamical observable (in the Heisenberg picture) represents a sum of a large number of typical Pauli strings, in which $4$ elementary Pauli matrices $I,X,Y,Z$ appear with approximately the same frequencies, so that their Pauli weight $\sim 3N/4$. Any jump operator anticommutes with exactly half of all Pauli strings $\{\sigma_{\boldsymbol{\beta}}\}_{\boldsymbol{\beta}}$ (see Appendix~\ref{appendix-anticommutativity}), so the geometric mean of fidelities reads
\begin{equation} \label{geometric-mean-of-fidelities}
\overline{f_{l\boldsymbol{\beta}}}^{\rm geom} = \prod_{q} \prod_{\alpha_q} \exp(- \lambda_{\alpha_q}^{[q](l)}) \prod_{\langle q_1,q_2\rangle} \prod_{\alpha_{q_1},\alpha_{q_2}} \exp(- \lambda_{\alpha_{q_1} \alpha_{q_2}}^{\langle q_1,q_2\rangle (l)}) \equiv \gamma_l^{-1/2},
\end{equation}

\noindent where $\gamma_l \equiv \exp(2\sum_{\lambda_{\boldsymbol{\alpha}} \in {\rm SPL}({\cal N}_l)} \lambda_{\boldsymbol{\alpha}})$ is a conventional quantifier for total noise accumulated in the whole SPL layer ${\cal N}_l$ widely used in PEC theory (PEC sampling overhead for the $l$th layer $\Gamma_l \geq \gamma_l^2$, see Sec.~\ref{section-pec}). For a randomly chosen Pauli string $\sigma_{\boldsymbol{\beta}} \neq I^{\otimes N}$ ($\boldsymbol{\beta} \neq {\bf 0}$), an individual fidelity $f_{l\boldsymbol{\beta}} = \exp(-2\sum_{\lambda_{\boldsymbol{\alpha}} \in {\rm SPL}({\cal N}_l)} \lambda_{\boldsymbol{\alpha}} \xi_{\boldsymbol{\alpha} \boldsymbol{\beta}})$ can be considered as a function of the uniformly distributed random variables $\xi_{\boldsymbol{\alpha} \boldsymbol{\beta}} \in \{0,1\}$ that either activate or deactivate damping~\footnote{Since a given non-trivial Pauli string anticommutes with exactly half of all Pauli strings (see Appendix~\ref{appendix-anticommutativity}), there is an additional requirement on random variables $\xi_{\boldsymbol{\alpha} \boldsymbol{\beta}} \in \{0,1\}$ that $\sum_{\lambda_{\boldsymbol{\alpha}} \in {\rm SPL}({\cal N}_l)} \xi_{\boldsymbol{\alpha} \boldsymbol{\beta}} = 4^N / 2$. The variables $\xi_{\boldsymbol{\alpha} \boldsymbol{\beta}} \in \{0,1\}$ can approximately be treated as independent and identically distributed (uniformly) if $N \gg 1$.}. The fidelities $\{f_{l\boldsymbol{\beta}}\}_{\boldsymbol{\beta}}$ are scattered around $\gamma_l^{-1/2}$, so 
$\ln f_{l\boldsymbol{\beta}} \approx -\frac{1}{2} \ln \gamma_l \pm \frac{1}{2} \sqrt{ \sum_{\lambda_{\boldsymbol{\alpha}} \in {\rm SPL}({\cal N}_l)} \lambda_{\boldsymbol{\alpha}}^2 }$. Upon propagating through all the circuit layers, different individual fidelities get multiplied and yield approximately the geometric mean $\prod_l \gamma_l^{-1/2} \equiv \gamma^{-1/2}$. The overall noise damps a typical high-weight observable as $\braket{O}_{\rm noisy} \approx \gamma^{-1/2} \braket{O}_{\rm ideal}$. Noise mitigation aims at restoring the noiseless estimation, meaning $\braket{O}_{\rm mitigated} \approx \sqrt{\gamma} \braket{O}_{\rm noisy}$.

The overall noise accumulation factor $\gamma \equiv \prod_l \gamma_l \equiv \exp (2 \sum_{\lambda_{\boldsymbol{\alpha}} \in \cup_l {\rm SPL}({\cal N}_l)} \lambda_{\boldsymbol{\alpha}})$ grows exponentially in the circuit volume. For instance, in a dense $N \times L$ circuit with $\frac{1}{2} NL$ {\sc cnot} gates each accompanied by a two-qubit depolarizing noise $\Lambda[\varrho] = (1-p)\varrho + p \, {\rm tr}[\varrho] \frac{1}{4} I^{\otimes 2}$ with $p \ll 1$, the overall noise accumulation factor $\gamma \approx (1 + \frac{15}{8} p)^{NL/2}$~\cite{temme-2017}, i.e., on average $\sqrt[NL]{\gamma} \approx 1 + \frac{15}{16} p$ per qubit per layer. Therefore, a convenient indicator of the error rate $\varepsilon$ per qubit per layer in a general circuit with the SPL noise is $\varepsilon \equiv \sqrt[NL]{\gamma} - 1$, which we also refer to as the density of errors. In dense circuits, $\varepsilon$ essentially reduces to the error per layered gate~\cite{mckay_benchmarking_2023} and $\varepsilon \approx (6NL)^{-1} \sum_{\lambda_{\boldsymbol{\alpha}} \in \cup_l {\rm SPL}({\cal N}_l)} \lambda_{\boldsymbol{\alpha}}$ for a linear chain topology of qubits. This definition enables us to express the overall noise accumulation factor $\gamma \equiv (1 + \varepsilon)^{NL} \approx \exp(\varepsilon NL)$ as the exponent of the average number of errors happening in the circuit, $\varepsilon NL$. If $\varepsilon NL \ll 1$, then the signal is only moderately attenuated and any noise mitigation technique is essentially effortless. However, large scale experiments necessarily involve many errors $\varepsilon NL \gg 1$ (for instance, the average number of errors $\varepsilon NL = 16$ for $N=L=100$ and the lowest gate error $\varepsilon = 0.16\%$ observed in superconducting quantum computers with more than $100$ qubits~\cite{mckay_benchmarking_2023}). This means the original signal of the order of $1$ is typically attenuated down to the values $e^{-\varepsilon NL/2}$ hardly distinguishable from the shot noise in practical scenarios ($e^{-\varepsilon NL/2} \approx 3.4 \times 10^{-4}$ for $\varepsilon NL = 16$). We focus on the challenging yet the most pressing scenario for noise mitigation with $\varepsilon NL \gg 1$, where the sampling overhead $\Gamma$ is extremely sensitive to any factor in the exponential scaling. 

In some experiments, noise is accumulated not from the whole circuit ($N \times L$) but only from some relevant causal area ${\cal A}$ at the intersection of causal cones for the observable and the nontrival component of the initial state~\cite{tran-2023}. In this case, the average number of relevant errors is $\varepsilon |{\cal A}|$ instead of $\varepsilon N L$ and the error mitigation task becomes significantly simpler if $|{\cal A}| \ll NL$ since the signal is attenuated much less, $\braket{O}_{\rm noisy} \approx \exp(-\varepsilon |{\cal A}|/2) \braket{O}_{\rm ideal}$. In the following analysis, we focus on and present results for the most pressing scenario where the causal area is essentially the whole circuit; however, in the derived formulae, one should replace $\varepsilon N L \rightarrow \varepsilon |{\cal A}|$ if applicable because the sampling overhead and the random error in observable estimation are exponentially sensitive to the average number of errors.

\subsection{\label{section-pec} Probabilistic error cancellation (PEC)}

\subsubsection{Method}

In this section, we present the derivation of all errors associated to PEC and summarized in Table~\ref{table-summary}. We begin recalling the method. Suppose one is interested in estimating the average value $\braket{O}$ for a pure state produced by an ideal layered circuit $\overleftarrow{\bigcirc}_l \, {\cal U}_l$. We abstract the measurement implementation from the error mitigation description and merely assume a measurement outcome $m$ is assigned the value $O_m$ in this noiseless scenario. For instance, if $O$ is an $N$-qubit Pauli string $P_0 \otimes \cdots \otimes P_{N-1}$, $P_k \in \{I,X,Y,Z\}$, then the $k$th qubit can be measured in the standard eigenbasis of the local Pauli operator $P_k$ and, depending on the measurement outcome, assigned the value $m_k = 0$ or $m_k = 1$ so that $m = (m_1,\ldots,m_{N-1})$ and $O_m = \prod_{k: \, P_k \neq I} (-1)^{m_k}$. In this noiseless scenario, $\braket{O}$ is estimated as the mean $\bar{O} = \frac{1}{|{\sf S}|} \sum_{m \in {\sf S}} O_m$ over a collection ${\sf S} \equiv {\sf S}_{\rm ideal}$ of $|{\sf S}|$ measurement shots. In the limit of infinitely many shots, $\lim_{|{\sf S}|\rightarrow \infty} \bar{O} = \braket{O}$. For a finite number of shots $|{\sf S}| \gg 1$, the deviation of $\bar{O}$ from the true value $\braket{O}$ is a random error quantified via the standard deviation estimation $\Delta \bar{O} = \frac{1}{|{\sf S}|} \sqrt{\sum_{m \in {\sf S}} (O_m - \bar{O})^2 }$.

In noisy quantum computations, every unitary layer ${\cal U}_l$ is preceded by a noisy layer ${\cal N}_l$ so that $\frac{1}{|{\sf S}|} \sum_{m \in {\sf S}} O_m$ is no longer a valid estimate for $\braket{O}$ because the distribution of outcomes after the noisy evolution (in the Schr\"{o}dinger picture) is completely different from that for the noiseless evolution (${\sf S} \equiv {\sf S}_{\rm noisy}$). In PEC~\cite{temme-2017,van_den_berg_probabilistic_2023,piveteau-2022}, the noise-inversion map for layer $l$ is represented as ${\cal N}_l^{-1} = \sum_k q_{k}^{(l)} {\cal V}_{k}^{(l)}$, where $\{{\cal V}_{k}^{(l)}\}_k$ are presumably noiseless unitary maps and $\{q_{k}^{(l)}\}_k$ is a quasiprobability. Since some $q_{k}^{(l)}$ are negative and $\sum_k q_{k}^{(l)} = 1$, the associated noise accumulation factor per layer $\gamma_{l} = \sum_k |q_{k}^{(l)}| > 1$. The factor for the whole circuit is $\gamma = \prod_l \gamma_l$. Every sample in PEC represents a new quantum circuit $\overleftarrow{\bigcirc}_l \, {\cal U}_l \circ {\cal V}_{k_l}^{(l)}$, where a unitary operation ${\cal V}_{k_l}^{(l)}$ for layer $l$ is drawn randomly from the true probability distribution $\{ \frac{1}{\gamma_l} |q_{k}^{(l)}| \}_k$. Every sample PEC circuit ${\cal V} \in {\sf Q}$ is run on the noisy hardware and measured. A measurement outcome $m$ is now assigned the value $\gamma \sigma_{\cal V} O_m$, where $\sigma_{\cal V} \in \{\pm 1\}$ is the sample-dependent sign of $\prod_l q_{k_l}^{(l)}$. The sought value $\braket{O}$ is estimated as the mean $\bar{\bar{O}}$ over circuit samples ${\sf Q}$ and measurement outcomes ${\sf S}_{\cal V}$ per each sampled circuit ${\cal V} \in {\sf Q}$, i.e., $\bar{\bar{O}} = \frac{1}{|{\sf Q}|} \sum_{{\cal V} \in {\sf Q}} \frac{1}{|{\sf S}_{\cal V}|} \sum_{m \in {\sf S}_{\cal V}} \gamma \sigma_{\cal V} O_m$. In the limit of infinitely many PEC circuits, $\lim_{|{\sf Q}| \rightarrow \infty} \bar{\bar{O}} = \braket{O}$ provided the noise is characterized perfectly and the unitary operations ${\cal V}_{k}^{(l)}$ are implemented perfectly too. For a finite number of PEC circuits $|{\sf Q}| \gg 1$ and the same finite number of shots per every circuit $|{\sf S}_{\cal V}| = |{\sf S}| \gg 1$, the deviation of $\bar{\bar{O}}$ from the true value $\braket{O}$ is a random error quantified via the standard deviation estimation $\Delta \bar{\bar{O}} = \frac{1}{|{\sf Q}|} \sqrt{ \sum_{{\cal V} \in {\sf Q}} \big[ ( \Delta \bar{O}_{\cal V})^2 + (\bar{O}_{\cal V} - \bar{\bar{O}})^2 \big]}$, where $\bar{O}_{\cal V} = \frac{1}{|{\sf S}|} \sum_{m \in {\sf S}_{\cal V}} \gamma \sigma_{\cal V} O_m$ and $\Delta \bar{O}_{\cal V} = \frac{1}{|{\sf S}|} \sqrt{\sum_{m \in {\sf S}_{\cal V}} (\gamma \sigma_{\cal V} O_m - \bar{O}_{\cal V})^2 }$.

From the implementation viewpoint, the repeatedly executed circuits are recompiled every time to include proper gates, whereas the measurement procedure remains exactly the same as in the ideal noiseless scenario. Since interventions in the circuit composition are made prior to the measurement, this can be interpreted as an amendment of the system dynamics in the Schr\"{o}dinger picture: $\frac{1}{|{\sf Q}|} \sum_{{\cal V} \in {\sf Q}} \gamma \sigma_{\cal V} \overleftarrow{\bigcirc}_l {\cal U}_l \circ {\cal N}_l \circ {\cal V}_{k_l}^{(l)} \rightarrow \overleftarrow{\bigcirc}_l \, {\cal U}_l$ if $|{\sf Q}| \rightarrow \infty$. The scaling coefficient $\gamma \sigma_{\cal V}$ is the same for all observables and measurement outcomes. In the absence of systematic errors, PEC results in an unbiased estimator of the pure state at the end of the ideal circuit, i.e., $\mathbb{E}_{{\cal V} \in {\sf Q}} \left(  \gamma \sigma_{\cal V} \overleftarrow{\bigcirc}_l \, {\cal U}_l \circ {\cal N}_l \circ {\cal V}^{(l)} [(\ket{0}\bra{0})^{\otimes N}] \right) = \overleftarrow{\bigcirc}_l \, {\cal U}_l [(\ket{0}\bra{0})^{\otimes N}]$.

Should one be interested in estimating a collection of observables $\{O_i\}_i$, where each observable requires a different measurement setting, then the same error mitigation protocol should be applied repeatedly to every observable on an individual basis. The associated cost could be overcome with the use of classical shadows~\cite{jnane_2024} or generalized multi-outcome measurements (e.g., informationally complete ones~\cite{garcia-perez-2021}) that enable one to estimate a number of different observables from a single collection of measurement outcomes. Generalized measurements allow for an additional optimization to minimize the random error~\cite{garcia-perez-2021,glos2022adaptive,malmi2024enhanced,fischer2024dual,caprotti2024optimising}. 

\subsubsection{Random error} \label{section-pec-random-error}

Let $M$ be the total number of circuit executions (measurement shots) available. If this budget of $M$ shots were used in the ideal quantum computation, then a Pauli observable $O=P_0 \otimes \cdots \otimes P_{N-1}$ would be estimated via measurement outcomes $O_m = \pm 1$ with the standard error $\Delta \bar{O} = \sqrt{\frac{1-\bar{O}^2}{M}}$. Instead, if this total budget of shots is used in a PEC-enhanced noisy quantum computation with one shot per circuit ($|{\sf S}|=1$, $|{\sf Q}| = M$), then $\Delta \bar{O}_{\cal V} = 0$ for any sampled circuit ${\cal V} \in {\sf Q}$ and $\Delta \bar{\bar{O}} = \sqrt{\frac{\gamma^2 - \bar{\bar{O}}^2}{M}}$. Since $\bar{O} = \bar{\bar{O}}$ in the limit $M \rightarrow \infty$, the sampling overhead in circuit runs $\Gamma = \frac{\gamma^2 - \bar{\bar{O}}^2}{1-\bar{O}^2} \geq \gamma^2$. Redistributing the total budget among $|{\sf Q}|$ circuit executions and $|{\sf S}|$ measurement shots per circuit ($M=|{\sf Q}| |{\sf S}|$) can be beneficial from a practical viewpoint to accelerate the experiment; however, this generally increases the estimation error. We therefore consider the best case scenario with the sampling overhead $\Gamma = \gamma^2$ and the random error $\delta_{\rm random} = \frac{\gamma}{\sqrt{M}}$.  

The locality of the sparse Pauli-Lindblad noise enables one to efficiently sample single-qubit gates composing layers ${\cal V}_k^{(l)}$ in PEC. The implementation of the map $e^{- \lambda (P \bullet P - \bullet)}$ with a single Pauli-Lindblad generator $P$ requires sampling from 2 unitary operations $\{I,P\}$ and induces the overhead $e^{2\lambda}$ \cite{van_den_berg_probabilistic_2023}. If all constituent maps for the sparse Pauli-Lindblad noise model are sampled in this way, the noise accumulation factor for a layer is $\gamma_l =  \exp (2 \sum_{\lambda_{\boldsymbol{\alpha}}^{(l)} \in {\rm SPL}({\cal N}_l)} \lambda_{\boldsymbol{\alpha}}^{(l)})$. A more efficient sampling with a slightly smaller overhead is possible if single-qubit maps (defined via 3 relaxation rates $\lambda_X$, $\lambda_Y$, $\lambda_Z$) are implemented via sampling from $4$ unitary operations $\{I,X,Y,Z\}$, and two-qubit maps (defined via $9$ relaxation rates $\lambda_{XX}, \lambda_{XY}, \lambda_{XZ}, \lambda_{YX}, \lambda_{YY}, \lambda_{YZ}, \lambda_{ZX}, \lambda_{ZY}, \lambda_{ZZ}$) are implemented via sampling from $16$ unitary operations $\{I,X,Y,Z\}^{\otimes 2}$ \cite{filippov_scalable_2023}. Nonetheless, this still leads to the same first order expansion with respect to the noise parameters in the exponent for $\gamma_l$. The overall sampling overhead $\gamma^2$ grows exponentially in the circuit volume $NL$ and the error rate $\varepsilon$ because $\gamma = \prod_l \gamma_l = \exp (2 \sum_{\lambda_{\boldsymbol{\alpha}} \in \cup_l {\rm SPL}({\cal N}_l) } \lambda_{\boldsymbol{\alpha}}) \equiv (1 + \varepsilon)^{NL} \approx e^{\varepsilon NL}$.

The standard PEC implementation can be seen as the noise mitigation in the Sch\"{o}dinger picture because the observable is measured the same way in both noisy and noiseless scenarios. To illustrate this, consider the trivial observable $I^{\otimes N}$ the  average value of which $\braket{I^{\otimes N}} = 1$ for any state (mixed or pure). The noisy values $O_m = +1$ for all measurement outcomes $m$, so the PEC sampling gives the values $\gamma \sigma_{\cal V} O_m = \pm \gamma$ that average to $\bar{\bar{O}} \approx 1$ but the error $\Delta \bar{\bar{O}} = \sqrt{\frac{\gamma^2 - 1}{M}}$ is exponential in the circuit volume. On the other hand, the propagation of the trivial observable $I^{\otimes N}$ through a noise-inversion map preserves it, $({\cal N}_{l}^{-1})^{\dag}[I^{\otimes N}] = I^{\otimes N}$, and this fact actually allows one to reliably estimate the observable without resorting to the error mitigation. An adaptation of PEC to the mitigation of only relevant noisy contributions (within the observable causal cone) is made in Ref.~\cite{piveteau-2022}. However, if the observable $O$ acts nontrivially on most of the qubits, the technique essentially involves the whole circuit and the overhead is exponential in the circuit volume (with no regard to how close $O$ could potentially be to $I^{\otimes N}$ or how much simpler the observable evolution in the Heisenberg picture could be as compared to the state evolution in the Schr\"{o}dinger picture \cite{hartmann_density_2009}).

\subsubsection{Systematic error} \label{section-pec-systematic-error}

The accuracy of the learned noise model plays a crucial role in the accuracy of the mitigated values. In practice, the noise is always characterized imprecisely because of the instabilities. For instance, the fluctuating microscopic defects (two-level systems) strongly affect the coherence of superconducting qubits~\cite{carroll_dynamics_2022} and may occasionally result in the increase of the single qubit overhead $\exp[2(\lambda_X+\lambda_Y+\lambda_Z)]$ from $\approx 1.0$ up to $\approx 1.3$~\cite[Sec. III B in Supplementary Information]{kim_evidence_2023}. In a quantum circuit with 127 qubits and 47 to 48 {\sc cnot} gates per layer, the observed temporal variations of $\gamma_l$ within $29$ separate learning attempts are of the order of $\gamma_l \sim 12_{-2}^{+4}$~\cite[Sec. IV B in Supplementary Information]{kim_evidence_2023}. High noise instability makes it necessary to relearn and update the noise model while running the experiment. 

To assess the foreseeable levels of noise stability, we refer to the $65$-qubit experiment with a series of $20$ sequential learning attempts within approximately $17$ hours~\cite[Sec. III B in Supplementary Information]{kim_evidence_2023}. In this experiment, the exclusion of the most unstable $5$ qubits from the $65$-qubit register results in a stable noise model for a single layer with $\gamma \approx 2.54$ and an average change in $\gamma$ from one attempt to another $\Delta\gamma \approx 0.042$~\cite[Sec. III B in Supplementary Information]{kim_evidence_2023}. This translates into a density of errors $\varepsilon = 0.01566 \pm 0.00028$. We use this experiment as a benchmark for the relative accuracy $\Delta \varepsilon / \varepsilon \approx 0.018$ in noise stability within $17$ hours. We anticipate the same relative accuracy to be achievable with lower noise intensities too, i.e., $\gamma_l \pm \Delta\gamma_l = [1+\varepsilon(1 \pm 0.018)]^N \approx \gamma_l (1 \pm 0.018\varepsilon)^N$. Therefore, our noise stability model assumes that, for a given average error rate $\varepsilon$ and an experiment wall time ${\rm T} \lesssim 1$ day, the noise amplification factor $\gamma_l$ for the $l$th layer fluctuates around the learned value as $\gamma_l (1 + \theta_l \varepsilon)^N$, where $\theta_l$ is a normally distributed random variable with mean $0$ and standard deviation $\big(\overline{\theta_l^2}\big)^{1/2} \equiv \Theta = 0.018$. In general, we expect a diffusion-like growth of the mismatch between the learned model and the actual noise in time, $\Theta \propto \sqrt{{\rm T}}$; however, based on the observations available, we stick to a simpler time-independent evaluation for achievable noise stability $\Theta = 1.8 \%$ within ${\rm T} \lesssim 1$ day.

Since $\braket{O}_{\rm mitigated} \approx \prod_l \gamma_l^{1/2} \braket{O}_{\rm noisy}$, the imprecision in the learned noise model and the noise scaling factor $\gamma_l$ translates into an imprecision in the mitigated observable. A layer $l$ results in a factor $\gamma_l^{1/2} (1 + \frac{1}{2} \theta_l \varepsilon)^N$. The errors from different layers partially compensate each other because the observable estimation is sensitive to different noise parameters $\{\lambda_{\boldsymbol{\alpha}}\}_{\boldsymbol{\alpha}}$ in ${\rm SPL}({\cal N}_l)$ (fidelities $f_{l\boldsymbol{\beta}}$) even if the noise layer composition remains the same. Since $\mathbb{E}(\sum_l \theta_l) = 0$ and $\mathbb{E}(\sum_l \theta_l)^2 = L \Theta^2$, the resulting observable estimation reads $\braket{O}_{\rm mitigated} \left(1 + \frac{1}{2} \varepsilon \sqrt{L} \Theta \right)^N \approx \braket{O}_{\rm mitigated} \left(1 + \frac{1}{2} \varepsilon N \sqrt{L} \Theta \right)$ if $\varepsilon N \sqrt{L} \Theta \ll 1$. For a square circuit with $N=L=100$, average error rate $\varepsilon = 0.16\%$ foreseeable in superconducting quantum computers with more than $100$ qubits~\cite{mckay_benchmarking_2023}, and standard deviation $\Theta = 1.8 \%$ in the relative density of errors, we get the systematic error $\delta_{\rm systematic}^{(1)} = 1.44 \%$.

\subsection{\label{section-zne-pea} Zero noise extrapolation (ZNE) with probabilistic error amplification}

\subsubsection{Method}

In this section, we present the derivation of all errors associated to ZNE and summarized in Table~\ref{table-summary}. We begin by recalling the method. The learned noise model ${\cal N}_l$ for a layer $l$ is amplified to the power ${\cal N}_l^G$, $G > 1$ by artificially implementing on hardware the extra noise ${\cal N}_l^{G-1}$ via sampling (presumably noiseless) unitary operations ${\cal V}_k^{(l)}$ according to the true probability distribution $\{p_k^{(l,G-1)}\}_k$ in the expansion ${\cal N}_l^{G-1} = \sum_k p_k^{(l,G-1)} {\cal V}_k^{(l)}$. As in PEC with SPL noise, a unitary operation ${\cal V}_k^{(l,G-1)}$ corresponds to a layer of single-qubit gates. Suppose $O$ is an observable of interest and the qubits are measured to provide the estimation $\bar{O}(G)$ of the average value ${\rm tr}\big\{ O \overleftarrow{\bigcirc}_l {\cal U}_l \circ {\cal N}_l^G [(\ket{0}\bra{0})^{\otimes N}] \big\}$. (If $O$ is a single Pauli string $P_0 \otimes \cdots \otimes P_{N-1}$, then qubit $k$ could be projectively measured in the eigenbasis of $P_k$; otherwise, generalized measurements could be exploited \cite{garcia-perez-2021}.) The procedure is repeated $R$ times for a number of noise amplification factors $G=G_1,G_2, \ldots$ ($1 \leq G_1 < G_2 < \ldots$) to get a collection of various noisy observable estimations $\{\bar{O}(G_i)\}_{i=1}^{R}$. These data are used to infer the extrapolation function $F(G)$ for $\bar{O}(G)$ and obtain the noise-mitigated value $\braket{O}_{\rm mitigated} = F(0)$~\cite{temme-2017,li-2017,kandala_error_2019,kim_scalable_2023}. 

Although ZNE with probabilistic error amplification implies the same active interventions in the circuit composition for any choice of the observable $O$, the final steps in the method still remain observable-specific. This is in drastic contrast to PEC, where sampling and assigning vales $\pm \gamma$ to individual circuits results in a noiseless circuit regardless of the observable of interest ($\frac{1}{|{\sf Q}|} \sum_{{\cal V} \in {\sf Q}} \gamma \sigma_{\cal V} \overleftarrow{\bigcirc} {\cal U}_l \circ {\cal N}_l \circ {\cal V}_{k_l}^{(l)} \rightarrow \overleftarrow{\bigcirc}_l \, {\cal U}_l$ if $|{\sf Q}| \rightarrow \infty$). We therefore argue that ZNE effectively operates with the dynamics of the specific observable $O$ in the noise-propagation Heisenberg picture. In fact, $\bar{O}(G)$ is an estimation for the average value 

\begin{eqnarray}
{\rm tr}\big\{ O \cdot \underbrace{\overleftarrow{\bigcirc}_l {\cal U}_l \circ {\cal N}_l^G [(\ket{0}\bra{0})^{\otimes N}]}_\text{noisy~computation} \big\} &=& {\rm tr}\big\{ \overrightarrow{\bigcirc}_l ({\cal N}_l^G)^{\dag} \circ {\cal U}_l^{\dag} [O] \cdot \underbrace{(\ket{0}\bra{0})^{\otimes N}}_\text{initial~state} \big\} \nonumber\\
&=& {\rm tr}\big\{ \overrightarrow{\bigcirc}_l ({\cal N}_l^{G})^{\dag} \circ {\cal U}_l^{\dag} [O] \cdot \underbrace{\overrightarrow{\bigcirc}_l {\cal U}_l^{\dag} \circ \overleftarrow{\bigcirc}_l {\cal U}_l}_{\rm Id} [(\ket{0}\bra{0})^{\otimes N}] \big\} \nonumber\\
&=& {\rm tr}\big\{ \overleftarrow{\bigcirc}_l {\cal U}_l \circ \overrightarrow{\bigcirc}_l ({\cal N}_l^{G})^{\dag} \circ {\cal U}_l^{\dag} [O] \cdot \underbrace{ \overleftarrow{\bigcirc}_l {\cal U}_l [(\ket{0}\bra{0})^{\otimes N}] }_\text{noiseless~computation} \big\} \nonumber\\
&=& {\rm tr}\big\{ \underbrace{\overrightarrow{\bigcirc}_l {\cal U}_{\geq l} \circ ({\cal N}_l^{G})^{\dag} \circ {\cal U}_{\geq l}^{\dag} }_\text{noise~propagation} [O] \cdot \underbrace{ \overleftarrow{\bigcirc}_l {\cal U}_l [(\ket{0}\bra{0})^{\otimes N}] }_\text{noiseless~computation} \big\} \nonumber\\
&=& \braket{ \overrightarrow{\bigcirc}_l (\widetilde{\cal N}_l^{G})^{\dag} [O] }_{\rm ideal}, \label{Schrodinger-to-Heisenberg-ZNE}
\end{eqnarray}
where ${\cal U}_{\geq l} = \overleftarrow{\bigcirc}_{m \geq l} {\cal U}_m$ represents layers $l,\ldots,L$ of the noiseless circuit  and $(\widetilde{\cal N}_l^{G})^{\dag} \equiv {\cal U}_{\geq l} \circ ({\cal N}_l^{G})^{\dag} \circ {\cal U}_{\geq l}^{\dag}$ is the Heisenberg-picture propagation of the noisy layer $({\cal N}_l^{G})^{\dag}$. The SPL noise is self-dual, ${\cal N}_l^{G\dag} = {\cal N}_l^{G}$, and the corresponding $\lambda$-parameters are the those for ${\cal N}_l$ multiplied by the noise rescaling factor $G$. The propagated noisy map $\widetilde{\cal N}_l^G \equiv (\widetilde{\cal N}_l)^G$ is self-dual too and reduces to the identity map in the limit $G \rightarrow 0$. 

Equation \eqref{Schrodinger-to-Heisenberg-ZNE} clarifies that $\{\bar{O}(G_i)\}_{i=1}^{R}$ are the noiseless estimations for the set of modified observables $\{\overrightarrow{\bigcirc}_l (\widetilde{\cal N}_l^{G_i})^{\dag} [O]\}_{i=1}^{R}$. In a general quantum circuit, ${\cal U}_{\geq l}^{\dag} [\sigma_{\boldsymbol{\beta}}] = \sum_{\boldsymbol{\alpha}} c_{\boldsymbol{\alpha\beta}}^{(l)} \sigma_{\boldsymbol{\alpha}}$ is a sum of Pauli strings $\{\sigma_{\boldsymbol{\alpha}}\}_{\boldsymbol{\alpha}}$ and each of those Pauli strings $\sigma_{\boldsymbol{\alpha}}$ gets scaled by the corresponding fidelity $f_{l\boldsymbol{\alpha}}^G$ defined through Eq.~\eqref{Pauli-fidelity} when affected by $({\cal N}_l^G)^{\dag}$. Therefore, $(\widetilde{\cal N}_l^G)^{\dag}[\sigma_{\boldsymbol{\beta}}] \equiv {\cal U}_{\geq l} ({\cal N}_l^G)^{\dag} {\cal U}_{\geq l}^{\dag} [\sigma_{\boldsymbol{\beta}}]$ represents a complex sum of rescaled Pauli stings. In operator $\overrightarrow{\bigcirc}_l (\widetilde{\cal N}_l^G)^{\dag} [O]$, the summands are rescaled by products of various fidelities, so the average value $\braket{ \overrightarrow{\bigcirc}_l \widetilde{\cal N}_l^{G} [O] }_{\rm ideal}$ is a sum of exponential functions of $G$. If we deal with a Clifford circuit and a Pauli observable $O = \sigma_{\boldsymbol{\beta}}$, then we have only one summand because $(\widetilde{\cal N}_l)^{\dag}[\sigma_{\boldsymbol{\beta}}] = f_{l\boldsymbol{\beta}(l)}^G \sigma_{\boldsymbol{\beta}}$, where $\boldsymbol{\beta}(l)$ is the index of the Pauli string ${\cal U}_{\geq l}^{\dag} [\sigma_{\boldsymbol{\beta}}] \equiv \sigma_{\boldsymbol{\beta}(l)}$. We get the exactly exponential observable rescaling $\overrightarrow{\bigcirc}_l (\widetilde{\cal N}_l^{G})^{\dag} [\sigma_{\boldsymbol{\beta}}] = \prod_l f_{l\boldsymbol{\beta}(l)}^G \sigma_{\boldsymbol{\beta}}$, justifying the use of exponential extrapolation $F(G)$ for the estimations $\{\bar{O}(G_i)\}_{i=1}^{R}$ in the case of Clifford circuits and Pauli observables. Note that the product $\prod_l f_{l\boldsymbol{\beta}(l)}^G \approx \prod_l \gamma_l^{-G/2} \equiv \gamma^{-G/2}$ for a random Clifford circuit in view of the geometric mean for fidelities \eqref{geometric-mean-of-fidelities}. For the same reason, in a general scenario of non-Clifford circuits, each exponential summand in $\overrightarrow{\bigcirc}_l (\widetilde{\cal N}_l^G)^{\dag}[O]$ scales roughly as $\prod_l \gamma_l^{-G/2} \exp\left( \pm \frac{G}{2} \sqrt{\sum_{\lambda_{\boldsymbol{\alpha}} \in {\rm SPL}({\cal N}_l)} \lambda_{\boldsymbol{\alpha}}^2} \right) \approx \gamma^{-G/2} \exp\left( \pm \frac{G}{2} \sqrt{\sum_{\lambda_{\boldsymbol{\alpha}} \in \cup_l {\rm SPL}({\cal N}_l)} \lambda_{\boldsymbol{\alpha}}^2} \right)$. 

The exponential extrapolation $F(G)$ leads to two types of errors in the mitigated value $\braket{O}_{\rm mitigated}$: a random error originating from imprecise estimations $\{\bar{O}(G_i) \pm \Delta\bar{O}(G_i) \}_{i=1}^{R}$ due to a finite budget of measurement shots, and a systematic error originating from the fact that the exponential extrapolation itself cannot precisely reproduce the sum of various exponents in the expansion for $\overrightarrow{\bigcirc}_l (\widetilde{\cal N}_l^{G})^{\dag} [O]$. To distinguish the random error, we consider the case Clifford circuits and Pauli observables, in which the exponential extrapolation is consistent, i.e., provides exactly the sought value $F(0) = \braket{O}_{\rm ideal}$ in the limit of infinitely many measurement shots available. The systematic error is evaluated via a typical spread of exponents in the expansion for $\overrightarrow{\bigcirc}_l (\widetilde{\cal N}_l^{G})^{\dag} [O]$.

\subsubsection{Random error} \label{section-zne-pea-random-error}

Consider a Clifford circuit $\overleftarrow{\bigcirc}_l {\cal U}_l$ and a Pauli observable $O$ stabilizing the circuit output, i.e., $O \ket{\psi} = \ket{\psi}$ for $\ket{\psi} = \overleftarrow{\bigcirc}_l U_l \ket{0}^{\otimes N}$ and $\braket{O}_{\rm ideal} = 1$. The noisy average value $\braket{ \overrightarrow{\bigcirc}_l \widetilde{\cal N}_l^{G_i} [O] }_{\rm ideal} = \prod_l f_{l\boldsymbol{\beta}(l)}^{G_i} \equiv K^{G_i}$, where $K = \prod_l f_{l\boldsymbol{\beta}(l)}$. The estimation $\bar{O}(G_i)$ obtained with $|{\sf S}_i|$ shots fluctuates around $K^{G_i}$ with $\Delta\bar{O}(G_i) = |{\sf S}_i|^{-1/2}$. Our goal is to deduce the minimum random error attainable with the use of \textit{exponential} extrapolation. In practice, this means linear extrapolation for the logarithm of experimental values. Provided $\Delta\bar{O}(G_i) \ll \bar{O}(G_i)$, the natural logarithm $\ln [\bar{O}(G_i) + \xi_i \Delta\bar{O}(G_i)] \approx \ln \bar{O}(G_i) + \xi_i \Delta\bar{O}(G_i) / \bar{O}(G_i)$ fluctuates around $\ln \bar{O}(G_i) \approx G_i \ln K$ with error $\Delta\bar{O}(G_i) / \bar{O}(G_i)$ (here, $\xi_i$ is a normally distributed random variable with mean $0$ and standard deviation $1$). 

The linear extrapolation $y = a x + b$ through $R$ points $\{(x_i, y_i + \xi_i \Delta y_i)\}_{i=1}^R$ via the least squares fitting yields $y(0) = b = \dfrac{\sum_i x_i^2 \sum_j (y_j + \xi_j \Delta y_j) - \sum_i x_i \sum_j x_j (y_j + \xi_j \Delta y_j)}{R\sum_i x_i^2 - (\sum_i x_i)^2}$~\cite{york_least-squares_1966}. For i.i.d. $\xi_i$ from the standard normal distribution we get the standard error $\Delta b = \dfrac{ \{ \sum_j [\sum_i x_i (x_i - x_j)]^2 (\Delta y_j)^2 \}^{1/2}}{R\sum_i x_i^2 - (\sum_i x_i)^2}$. Applying this result to the points $x_i = G_i$, $y_i = G_i \ln K$, $\Delta y_i = \Delta\bar{O}(G_i) / \bar{O}(G_i) = K^{-G_i} |{\sf S}_i|^{-1/2}$, we get $\mathbb{E}b = 0$ and $F(0) = e^{\xi \Delta b} \approx 1 + \xi \Delta b$ fluctuating around the ideal value $1$ with the random extrapolation error 
\begin{eqnarray}
\Delta F(0) = \Delta b &=& \frac{ \bigg\{ \sum\limits_j \big[\sum_i G_i (G_i - G_j) \big]^2 \dfrac{1}{K^{2G_j} |{\sf S}_j|} \bigg\}^{1/2}}{R\sum_i G_i^2 - (\sum_i G_i)^2} \label{ZNE-random-error} \\
& \approx & \frac{ \bigg\{ \sum_j \big[\sum_i G_i (G_i - G_j)\big]^2 \dfrac{\exp(\varepsilon G_j N L)}{|{\sf S}_j|} \bigg\}^{1/2}}{R\sum_i G_i^2 - (\sum_i G_i)^2}, \nonumber
\end{eqnarray}

\noindent where in the second line we have taken into account the typical noisy estimation $K \approx \gamma^{-1/2} = e^{-\varepsilon N L / 2}$. 

Given a finite number of measurement shots $M$, one can find the optimal ZNE setting by solving the constrained optimization problem: minimize $\Delta F(0)$ over sets $\{G_i\}_{i=1}^R$ of $R$ noise scaling factors $G_i \geq 1$ and measurement shots $\{|{\sf S}_i|\}_{i=1}^R$ under the constraint $M = \sum_{i=1}^R |{\sf S}_i|$. For fixed noise scaling factors $\{G_i\}_{i=1}^R$, the optimal arrangement of shots $\{|{\sf S}_i^{\ast}|\}_{i=1}^R$ is straightforward to find: $|{\sf S}_j^{\ast}| \propto \left\vert \sum_i G_i (G_i - G_j) \right\vert K^{-G_j}$ (see details in Appendix~\ref{appendix-ZNE-random}). The minimal relative error in the noise mitigation value reads
\begin{eqnarray}
\min \Delta F(0) & = & \min\limits_{\{G_i\}_{i=1}^R} \frac{ \sum_{j} \left\vert \sum_i G_i (G_i - G_j) \right\vert K^{-G_j} }{R\sum_i G_i^2 - (\sum_i G_i)^2} \cdot \frac{1}{\sqrt{M}} \label{ZNE-random-error-min} \\
& \approx & \min\limits_{\{G_i\}_{i=1}^R} \frac{ \sum_{j} \left\vert \sum_i G_i (G_i - G_j) \right\vert e^{\varepsilon G_j N L / 2} }{R\sum_i G_i^2 - (\sum_i G_i)^2} \cdot \frac{1}{\sqrt{M}}. \nonumber
\end{eqnarray}

\noindent On the one hand, the noise scaling factors cannot be all concentrated near the value $1$ because of the vanishing denominator in $\Delta F(0)$. On the other, the farther the noise scaling factor $G_j$ deviates from $1$, the more demanding the need to increase the allotted number of shots $|{\sf S}_j|$ to compensate for the exponentially growing term $\exp(\varepsilon G_j N L)$ becomes. Numerics suggest that the optimization with several extrapolation points ($R=3,4$) gives the same minimum in~\eqref{ZNE-random-error-min} as that for $R=2$. The case $R=2$ is exactly solvable and the straightforward analysis yields the optimal scaling factors $G_1^{\ast} = 1$ and $G_2^{\ast} = 1 +[1+W(1/e)] / {\ln (1/K)} \approx 1 + 1.278 / {\ln (1/K)} \approx 1 + \frac{2.557}{\varepsilon N L}$, where $W$ is the principal branch of the Lambert W function (see details in Appendix~\ref{appendix-ZNE-random}). The resulting random error reads
\begin{equation} \label{ZNE-minimal-random-error}
    \delta_{\rm random} = \min \Delta F(0) = \frac{1+[W(1/e)]^{-1} \ln(1/K)}{K\sqrt{M}} \approx \frac{(1+1.795 \varepsilon N L)\exp(\varepsilon N L / 2)}{\sqrt{M}}.
\end{equation}

The typical random error in evaluating an observable $O$, $\|O\|\leq 1$, in a noiseless quantum computation is $\frac{1}{\sqrt{M}}$. Comparing this value against the random error \eqref{ZNE-minimal-random-error} in the noise-mitigated estimation, we infer the minimum ZNE sampling overhead $\Gamma^{\ast} = [1 + 3.591 \ln(1/K)]^2 K^{-2} = (1+1.795 \varepsilon N L)^2 \exp(\varepsilon N L)$.

\subsubsection{Systematic error} \label{section-zne-systematic-error}

\textbf{Effect of imprecisely learned noise model and noise instability}. Similarly to PEC, ZNE with probabilistic error amplification is sensitive to the fluctuations in noise parameters and the imprecision in noise learning. Suppose the noise during the data acquisition differs from that in the learned model: $K'$ is the actual average value of a stablizer observable in a noisy Clifford circuit without noise amplification, and $K^{G_i - 1}$ is a multiplier in the observable average value for a selected noise factor $G_i$ due to amplification of the learned noise ($K' \neq K$). Since $\ln ( K'K^{G_i-1} ) = \ln(K'/K) + G_i \ln K$, the exponential extrapolation in ZNE results in the systematic error $\exp(\ln(K'/K)) - 1 = \frac{K'}{K} - 1$. Since $K \approx \gamma^{-1/2}$, this systematic error reduces to the one for PEC, namely, $\frac{1}{2} 
\varepsilon N \sqrt{L} \Theta$ with the standard deviation in the relative error per layer per qubit $\Theta$ (provided $\varepsilon N \sqrt{L} \Theta \ll 1$, see Sec.~\ref{section-pec-systematic-error}).

\textbf{Effect of extrapolation}. In ZNE, there is an additional systematic error arising from the fact that the exponential extrapolation leads to a bias for a general signal that decays non-exponentially in noise strength~\cite{cai_multi-exponential_2021}. To analytically study this source of error, we consider the following ``mirrored'' non-Clifford circuit. The circuit starts with a noiseless layer of Hadamard gates $H^{\otimes N}$ followed by a noiseless layer of $N_{T}$ non-Clifford $T$-gates scattered among the qubits (identity transformation for the remaining $N-N_T$ qubits). Then we have a depth-$\frac{L}{2}$ random Clifford circuit with SPL noise, $\overleftarrow{\bigcirc}_{l=1}^{L/2} {\cal U}_l \circ {\cal N}_l$. This concludes the half of the circuit, and the second half just inverts the unitary part of the first half, i.e., is given by the map $\overleftarrow{\bigcirc}_{l=L/2 + 1}^{L} {\cal U}_{L+1-l}^{\dag} \circ {\cal N}_l$ followed by $N_{T}$ non-Clifford $T^{\dag}$ gates and a layer of Hadamard gates $H^{\otimes N}$. We consider the observable $O = Z^{\otimes N}$ so that in the noiseless case, $\braket{O}_{\rm ideal} = 1$.

Since $THZHT^{\dag} = \frac{1}{\sqrt{2}}(X + Y)$, we have $2^{N_T}$ different Pauli strings propagating through the Clifford part of the circuit in the Heisenberg picture. At the beginning and at the end of the Clifford evolution, these Pauli strings contain $X$ or $Y$ operators at qubits subjected to $T$ gates and $X$ operators elsewhere. Each of the Pauli strings gathers a specific fidelity from the SPL noise. Since $\bra{0}HT^{\dag}XTH\ket{0} = \bra{0}HT^{\dag}YTH\ket{0} = \frac{1}{\sqrt{2}}$, the noise-$G$ average value 
\begin{equation} \label{O-sum-of-exponents}
\braket{O}(G) = \frac{1}{2^{N_T}} \sum_{n=1}^{2^{N_T}} \prod_{l=1}^L f_{l\boldsymbol{\beta}_n(l)}^G,
\end{equation}

\noindent which is a sum of $2^{N_T}$ exponential functions in $G$. One can readily verify by simple examples of a sum of exponents [such as $f(G) = \frac{1}{2}(e^{-G}+ e^{-2G})$] that the exponential extrapolation through points $G_i \geq 0$ cannot exactly reproduce the value $f(0)= 1$ and actually leads to its underestimation. In what follows, we evaluate the extrapolation error for Eq.~\eqref{O-sum-of-exponents}.

As explained in Sec.~\ref{section-noise-model}, each fidelity $f_{l\boldsymbol{\beta}_n(l)}$ in Eq.~\eqref{Pauli-fidelity} can be thought of as a random variable $\exp(-2 \sum_{\boldsymbol{\alpha}} \lambda_{\boldsymbol{\alpha}}^{(l)} \xi_{\boldsymbol{\alpha}}^{(l)})$, where $\xi_{\boldsymbol{\alpha}}^{(l)} \in \{0,1\}$ is the uniform random variable [$p(\xi_{\boldsymbol{\alpha}}^{(l)} = 0) = p(\xi_{\boldsymbol{\alpha}}^{(l)} = 1) = \frac{1}{2}$] reflecting whether the Pauli string in the dynamical observable commutes or anticommutes with the jump operator $\sigma_{\boldsymbol{\alpha}}$ associated with the relaxation rate $\lambda_{\boldsymbol{\alpha}}^{(l)}$. Since the mean $\mathbb{E}\xi_{\boldsymbol{\alpha}}^{(l)} = \frac{1}{2}$ and the variance $\mathbb{D}\xi_{\boldsymbol{\alpha}}^{(l)} = \frac{1}{4}$, the random variable $\Xi \equiv 2 \sum_l \sum_{\boldsymbol{\alpha}} \lambda_{\boldsymbol{\alpha}}^{(l)} \xi_{\boldsymbol{\alpha}}^{(l)}$ has the mean $\mathbb{E}\Xi = \sum_l \sum_{\boldsymbol{\alpha}} \lambda_{\boldsymbol{\alpha}}^{(l)}$ and the variance $\mathbb{D} \Xi = \sum_l \sum_{\boldsymbol{\alpha}} ( \lambda_{\boldsymbol{\alpha}}^{(l)} )^2$. The random variable $\prod_{l=1}^L f_{l\boldsymbol{\beta}_n(l)}^G = \exp[-G\Xi]$, so $\braket{O}(G)$ is a sum of $2^{N_T}$ random variables, namely, $\braket{O}(G) = \frac{1}{2^{N_T}} \sum_{n=1}^{2^{N_T}} \exp(- G \Xi_n)$ with i.i.d. random variables $\Xi_n$, whose distribution is well approximated by a normal one with mean $\mathbb{E}\Xi = \sum_l \sum_{\boldsymbol{\alpha}} \lambda_{\boldsymbol{\alpha}}^{(l)} \equiv \sum_{\lambda_{\boldsymbol{\alpha}} \in \cup_l {\rm SPL}({\cal N}_l)} \lambda_{\boldsymbol{\alpha}} \approx \frac{1}{2} \varepsilon N L $ and variance $\mathbb{D} \Xi = \sum_l \sum_{\boldsymbol{\alpha}} ( \lambda_{\boldsymbol{\alpha}}^{(l)} )^2 \equiv \sum_{\lambda_{\boldsymbol{\alpha}} \in \cup_l {\rm SPL}({\cal N}_l)} \lambda_{\boldsymbol{\alpha}}^2 \approx \frac{1}{18} \varepsilon^2 N L$ if the one-qubit (two-qubit) relaxation rates $\lambda_{\boldsymbol{\alpha}}^{(l)}$ in the SPL model are normally distributed with both mean and standard deviation equal to $\varepsilon/12$ ($\varepsilon/36$) similarly to the relaxation rates in the experiment~\cite{van_den_berg_probabilistic_2023}. Since $\sqrt{\mathbb{D} \Xi} \approx 1.4 \times 10^{-3} \ll 1$ if $N=L=100$ and $\varepsilon = 0.16\%$, we can resort to the series expansion. The logarithm of the noisy value is
\begin{eqnarray}
y_i &=& \ln \Big( \frac{1}{2^{N_T}} \sum_{n=1}^{2^{N_T}} \exp(-G_i \Xi_n) \Big) \nonumber\\
&=& \ln \Big( \exp(-G_i \mathbb{E}\Xi) \frac{1}{2^{N_T}} \sum_{n=1}^{2^{N_T}} \exp[-G_i (\Xi_n - \mathbb{E}\Xi)] \Big) \nonumber\\
&\approx& -G_i \mathbb{E}\Xi + \ln \Big( \frac{1}{2^{N_T}} \sum_{n=1}^{2^{N_T}} \left[ 1 - G_i (\Xi_n - \mathbb{E}\Xi) + \tfrac{1}{2} G_i^2 (\Xi_n - \mathbb{E}\Xi)^2 \right] \Big) \nonumber\\
&\approx& -G_i \mathbb{E}\Xi + \ln \bigg( 1 - \frac{G_i \sqrt{\mathbb{D}\Xi} }{2^{N_T/2}} + \tfrac{1}{2} G_i^2 \mathbb{D}\Xi \bigg) \nonumber\\
&\approx& -G_i \mathbb{E}\Xi - \frac{G_i \sqrt{\mathbb{D}\Xi} }{2^{N_T/2}} + \tfrac{1}{2} G_i^2 \mathbb{D}\Xi - \frac{G_i^2 \mathbb{D}\Xi }{2^{N_T + 1}}. \nonumber
\end{eqnarray}

\noindent The linear terms $\propto G_i$ in $y_i$ do not contribute to the estimation bias $\Delta F(0) \equiv F(0)-1$ because they describe exactly exponential decay in $\braket{O}(G)$. Nonlinear terms $\propto G_i^2$ in $y_i$ lead to the extrapolation error 
\begin{eqnarray}
|\Delta F(0)| = |\Delta b| &=& \frac{\mathbb{D}\Xi}{2} \Big( 1 - \frac{1}{2^{N_T}} \Big) \left\vert \frac{ (\sum_i G_i^2)^2 - (\sum_i G_i) (\sum_i G_i^3) }{R\sum_i G_i^2 - (\sum_i G_i)^2} \right\vert \nonumber \\
& \geq & \frac{\mathbb{D}\Xi}{2} \Big( 1 - \frac{1}{2^{N_T}} \Big) G_1 G_2 \nonumber \\
& = & \frac{G_1 G_2}{2} \Big( 1 - \frac{1}{2^{N_T}} \Big) \sum_{\lambda_{\boldsymbol{\alpha}} \in \cup_l \text{SPL}({\cal N}_l)} \lambda_{\boldsymbol{\alpha}}^2 \label{ZNE-systematic-error},
\end{eqnarray}

\noindent where the inequality in the second line is proven in Appendix~\ref{appendix-ZNE-systematic}. It is also shown in Appendix~\ref{appendix-ZNE-systematic} that $\Delta F(0)$ is negative meaning the ZNE tends to underestimate the observable in this example (since all the Pauli strings have positive weights in the observable). If $N_T \gg 1$, the noise amplification factors $G_1$ and $G_2$ are chosen as to minimize the sampling overhead ($G_1 = G_1^{\ast}$ and $G_2 = G_2^{\ast}$), and the one-qubit (two-qubit) relaxation rates $\lambda_{\boldsymbol{\alpha}}^{(l)}$ in the SPL model are normally distributed with both mean and standard deviation equal to $\varepsilon/12$ ($\varepsilon/36$), similarly to the experiment~\cite{van_den_berg_probabilistic_2023}, then
\begin{equation}
\delta_{\rm systematic}^{(2)} = |\Delta F(0)| \approx \left( 1 + \frac{2.557}{\varepsilon N L} \right) \frac{\varepsilon^2 NL}{36}.
\end{equation}

\subsection{\label{section-tem} Tensor-network error mitigation (TEM)}

\subsubsection{Method}

In this section, we present the derivation of all errors associated to TEM and summarized in Table~\ref{table-summary}. We begin by recalling the method. To get a noiseless estimation $\bar{O}_{\rm mitigated}$ of an observable $O$ without any interventions in the dynamics of a noisy quantum computation, one can actually measure another observable at the output of the noisy circuit:
\begin{eqnarray}
{\rm tr}\big\{ O \cdot \underbrace{\overleftarrow{\bigcirc}_l {\cal U}_l [(\ket{0}\bra{0})^{\otimes N}]}_\text{noiseless~computation} \big\} &=& {\rm tr}\big\{ \overrightarrow{\bigcirc}_l {\cal U}_l^{\dag} [O] \cdot \underbrace{(\ket{0}\bra{0})^{\otimes N}}_\text{initial~state} \big\} \nonumber\\
&=& {\rm tr}\big\{ \overrightarrow{\bigcirc}_l {\cal U}_l^{\dag} [O] \cdot \underbrace{\overrightarrow{\bigcirc}_l {\cal N}_l^{-1} \circ {\cal U}_l^{\dag} \circ \overleftarrow{\bigcirc}_l {\cal U}_l \circ {\cal N}_l}_{\rm Id} [(\ket{0}\bra{0})^{\otimes N}] \big\} \nonumber\\
&=& {\rm tr}\big\{ \overleftarrow{\bigcirc}_l {\cal U}_l \circ ({\cal N}_l^{-1})^{\dag} \circ \overrightarrow{\bigcirc}_l {\cal U}_l^{\dag} [O] \cdot \underbrace{ \overleftarrow{\bigcirc}_l {\cal U}_l \circ {\cal N}_l [(\ket{0}\bra{0})^{\otimes N}] }_\text{noisy~computation} \big\} \nonumber\\
&=& {\rm tr}\big\{ \underbrace{\overleftarrow{\bigcirc}_l {\cal U}_{\geq l} \circ ({\cal N}_l^{-1})^{\dag} \circ {\cal U}_{\geq l}^{\dag} }_\text{inverse-noise~propagation} \!\!\![O] \cdot \underbrace{ \overleftarrow{\bigcirc}_l {\cal U}_l \circ {\cal N}_l [(\ket{0}\bra{0})^{\otimes N}] }_\text{noisy~computation} \big\} \nonumber\\
&=& \braket{ \overleftarrow{\bigcirc}_l (\widetilde{\cal N}_l^{-1})^{\dag} [O] }_{\rm noisy}, \label{Schrodinger-to-Heisenberg-TME}
\end{eqnarray}
where ${\cal U}_{\geq l} = \overleftarrow{\bigcirc}_{m\geq l} {\cal U}_m$ represents layers $l,\ldots,L$ of the noiseless circuit and $\widetilde{\cal N}_l \equiv {\cal U}_{\geq l} \circ {\cal N}_l \circ {\cal U}_{\geq l}^{\dag}$ is the Heisenberg-picture propagation of the noisy layer ${\cal N}_l$. Therefore, the modified observable is $\overleftarrow{\bigcirc}_l (\widetilde{\cal N}_l^{-1})^{\dag} [O]$. 
The point of TEM is to make the identification of the modified observable efficient via tensor-network machinery. 

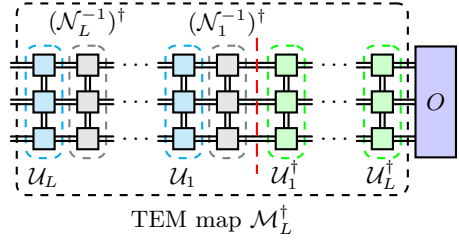
\begin{figure}
\centering
\begin{quantikz}[row sep=0.1cm, column sep=0.1cm]
\setwiretype{c} \gategroup[3,steps=24,style={dashed,rounded
corners,fill=white!20, inner
xsep=-3pt, inner
ysep=12pt,xshift=2.5pt,yshift=1pt},background,label style={label
position=below,anchor=north,yshift=-0.2cm}]{TEM map ${\cal M}_L^{\dag}$} & & & \gate[style={fill=cyan!20}]{} \wire[d][2]{c} \gategroup[3,steps=1,style={cyan,dashed,rounded
corners,fill=white!20, inner
xsep=-3pt, inner
ysep=-3pt,xshift=0pt,yshift=0pt},background,label style={label
position=below,anchor=north,yshift=-0.2cm}]{${\cal U}_L$} & & & \gate[style={fill=gray!20}]{} \wire[d][2]{c} \gategroup[3,steps=1,style={gray,dashed,rounded
corners,fill=white!20, inner
xsep=-3pt, inner
ysep=-3pt,xshift=0pt,yshift=0pt},background,label style={label
position=above,anchor=south,yshift=-0.2cm}]{$({\cal N}_L^{-1})^{\dag}$} & &  \ \ldots \ & & \gate[style={fill=cyan!20}]{} \wire[d][2]{c} \gategroup[3,steps=1,style={cyan,dashed,rounded
corners,fill=white!20, inner
xsep=-3pt, inner
ysep=-3pt,xshift=0pt,yshift=0pt},background,label style={label
position=below,anchor=north,yshift=-0.2cm}]{${\cal U}_1$} & & & \gate[style={fill=gray!20}]{} \wire[d][2]{c} \gategroup[3,steps=1,style={gray,dashed,rounded
corners,fill=white!20, inner
xsep=-3pt, inner
ysep=-3pt,xshift=0pt,yshift=0pt},background,label style={label
position=above,anchor=south,yshift=-0.2cm}]{$({\cal N}_1^{-1})^{\dag}$} & & \slice{} & & & \gate[style={fill=green!20}]{} \wire[d][2]{c} \gategroup[3,steps=1,style={green,dashed,rounded
corners,fill=white!20, inner
xsep=-3pt, inner
ysep=-3pt,xshift=0pt,yshift=0pt},background,label style={label
position=below,anchor=north,yshift=-0.1cm}]{${\cal U}_1^{\dag}$} & & \ \ldots \ & & \gate[style={fill=green!20}]{} \wire[d][2]{c} \gategroup[3,steps=1,style={green,dashed,rounded
corners,fill=white!20, inner
xsep=-3pt, inner
ysep=-3pt,xshift=0pt,yshift=0pt},background,label style={label
position=below,anchor=north,yshift=-0.1cm}]{${\cal U}_L^{\dag}$} & & & \gate[3,disable auto height,style={fill=blue!20}]{O} \\
\setwiretype{c} & & & \gate[style={fill=cyan!20}]{} & & & \gate[style={fill=gray!20}]{} & & \ \ldots \ & & \gate[style={fill=cyan!20}]{} & & & \gate[style={fill=gray!20}]{} & & & & & \gate[style={fill=green!20}]{} & & \ \ldots \ & & \gate[style={fill=green!20}]{} & & & \\
\setwiretype{c} & & & \gate[style={fill=cyan!20}]{} & & & \gate[style={fill=gray!20}]{} & & \ \ldots \ & & \gate[style={fill=cyan!20}]{} & & & \gate[style={fill=gray!20}]{} & & & & & \gate[style={fill=green!20}]{} & & \ \ldots \ & & \gate[style={fill=green!20}]{} & & &
\end{quantikz}
\caption{Modification of observable $O$ by the TEM map ${\cal M}_L^{\dag}$ in the Heisenberg picture (maps act from right to left). Vertical line (red) depicts the starting map ${\cal M}_0^{\dag} = {\rm Id}$ in the iterative construction of TEM map, ${\cal M}_l^{\dag} = {\cal U}_l ({\cal N}_l^{-1})^{\dag} {\cal M}_{l-1}^{\dag} {\cal U}_l^{\dag}$.}
\label{figure-tem-map-heisenberg}
\end{figure}

Reference \cite{filippov_scalable_2023} suggests to represent the map ${\cal M}_L^{\dag} \equiv \overleftarrow{\bigcirc}_l (\widetilde{\cal N}_l^{-1})^{\dag}$ via the recurrent relation ${\cal M}_l^{\dag} = {\cal U}_l \circ ({\cal N}_l^{-1})^{\dag} \circ {\cal M}_{l-1}^{\dag} \circ {\cal U}_l^{\dag}$ with ${\cal M}_0^{\dag} = {\rm Id}$, see Fig.~\ref{figure-tem-map-heisenberg}. Each of the constituent maps ${\cal U}_l$, $({\cal N}_l^{-1})^{\dag}$, and ${\cal U}_l^{\dag}$ adopts a matrix-product-operator representation (MPO) with bond dimension $4$ if ${\cal U}_l$ is composed of two-qubit entangling gates acting on neighbouring qubits and ${\cal N}_l$ is a SPL noise layer~\cite{filippov_scalable_2023}. By sequentially applying the iterative formula, we get a bond dimension for ${\cal M}_l$ that grows exponentially in $l$ if the multiplication of MPOs is treated exactly. To avoid the exponential complexity, TEM provides a collection of noise mitigation maps $\{{\cal M}_l^{\dag (\chi)}\}_{l=1}^L$ for different depths via the iterative procedure ${\cal M}_l^{\dag (\chi)} = {\cal C}_{\chi}\left( {\cal U}_l \circ ({\cal N}_l^{-1})^{\dag} \circ {\cal M}_{l-1}^{\dag (\chi)} \circ {\cal U}_l^{\dag} \right)$ involving the bond dimension compression ${\cal C}_{\chi}$ if the bond dimension of ${\cal U}_l \circ ({\cal N}_l^{-1})^{\dag} \circ {\cal M}_{l-1}^{\dag (\chi)} \circ {\cal U}_l^{\dag}$ exceeds some predefined value $\chi$. [For the sake of a faster classical calculation, the compression can be divided into two steps: ${\cal M}_{l-1}^{\prime \dag (\chi)} = {\cal C}_{\chi} \left( ({\cal N}_l^{-1})^{\dag} \circ {\cal M}_{l-1}^{\dag (\chi)} \right)$ and ${\cal M}_{l}^{\dag (\chi)} = {\cal C}_{\chi} \left( {\cal U}_l \circ {\cal M}_{l-1}^{\prime \dag (\chi)} \circ {\cal U}_l^{\dag} \right)$.] The most common compression techniques in linear tensor networks are based on the singular value decomposition (SVD) and on variational search. The bond dimension compression neglects the smallest contributions to the noise mitigation map thus ensuring its computational efficiency while introducing a bias in the observable estimation $\bar{O}_{\rm mitigated} = \braket{ {\cal M}_L^{\dag (\chi)} [O] }_{\rm noisy}$.

Measurements in TEM are optimized to make the estimation $\braket{ {\cal M}_L^{\dag(\chi)} [O] }_{\rm noisy}$ most efficient. The complexity of the modified observable ${\cal M}_L^{\dag(\chi)} [O]$ can in general be smaller, equal, or greater than that of the original observable $O$. The last situation takes place, for instance, if $O$ is a single Pauli observable and the circuit is non-Clifford. In this case, projective measurements in the eigenbasis of the original observable may still suffice for an accurate estimation of the mitigated observable despite the fact that ${\cal M}_L^{\dag(\chi)} [O]$ is a sum of many (generally non-commuting) Pauli strings (see Secs.~\ref{section-tem-random-error} and \ref{section-tem-systematic-error}). This is because not only the modified observable but also the structure of the noisy density operator at the output of the noisy quantum computation affects the average value $\braket{ {\cal M}_L^{\dag(\chi)} [O] }_{\rm noisy}$. This opens an opportunity for a simplified measurement strategy, in which the individual Pauli components $\sigma_{\boldsymbol{\beta}}$ of the original observable $O=\sum_{\boldsymbol{\beta}} w_{\boldsymbol{\beta}} \sigma_{\boldsymbol{\beta}}$ are rescaled by the diagonal components $\big( {\cal M}_L^{\dag(\chi)} \big)_{\boldsymbol{\beta\beta}}$ of the TEM map. As we show in Sec.~\ref{section-tem-systematic-error}, this simplified measurement strategy induces a minor systematic error $\varepsilon^2 N L / 72$. Otherwise, ${\cal M}_L^{\dag(\chi)} [O]$ can be estimated with the use of generalized measurements (e.g., informationally complete ones that can be adjusted as to minimize the estimation error~\cite{garcia-perez-2021,glos2022adaptive,malmi2024enhanced,fischer2024dual,caprotti2024optimising}).

\subsubsection{Random error} \label{section-tem-random-error}

As in the analysis of ZNE's random error, TEM's random error can readily be evaluated in Clifford circuits. In the case of a Clifford circuit $\overleftarrow{\bigcirc}_l {\cal U}_l$ and a Pauli observable $O = \sigma_{\boldsymbol{\beta}}$, the modified observable adopts an exact representation, namely, ${\cal M}_L^{\dag} [O] = \prod_l f_{l\boldsymbol{\beta}(l)}^{-1} O \equiv K^{-1} O$, where $\boldsymbol{\beta}(l)$ is the index of the Pauli string evolved in the Heisenberg picture, ${\cal U}_{\geq l}^{\dag} [\sigma_{\boldsymbol{\beta}}] \equiv \sigma_{\boldsymbol{\beta}(l)}$. The estimate $\bar{O}_{\rm mitigated}$ is obtained by performing conventional local projective measurements in the eigenbasis of the original observable and multiplying all $M$ conventional outcomes $\pm 1$ by the same number $K^{-1} \equiv \prod_l f_{l\boldsymbol{\beta}(l)}^{-1}$. This incurs the estimation error $\Delta\bar{O}_{\rm mitigated} = \frac{1}{K\sqrt{M}}$ and the sampling overhead $\Gamma = K^{-2}$. In typical large-scale circuits, $K \approx \gamma^{-1/2} = \exp(-\varepsilon N L / 2)$, therefore $\delta_{\rm random} = \Delta\bar{O}_{\rm mitigated} = \exp(\varepsilon N L / 2) / \sqrt{M}$ and $\Gamma^{\ast} = \gamma \approx \exp(\varepsilon N L)$. 

Let us extend the analysis of TEM random error beyond Clifford circuits and Pauli observables. Consider the propagation of an observable $O = \sum_{\boldsymbol{\beta}} w_{\boldsymbol{\beta}} \sigma_{\boldsymbol{\beta}}$ through the TEM map. Every non-Clifford unitary layer ${\cal U}_l$ results in the branching of Pauli strings: ${\cal U}_l^{\dag} [\sigma_{\boldsymbol{\beta}}] = \sum_{\boldsymbol{\alpha}} c_{\boldsymbol{\alpha\beta}}^{(l)} \sigma_{\boldsymbol{\alpha}}$, where $\left( c_{\boldsymbol{\alpha\beta}}^{(l)} \right)_{\boldsymbol{\alpha,\beta}} \equiv C_l$ is a real orthogonal matrix. By sequentially applying the rightmost $L$ layers of the TEM map to $O$, we get the generalized Bloch vector $(\overrightarrow{\prod}_l C_l) {\bf w}$ for the operator ${\cal U}_1^{\dag} \circ \cdots \circ {\cal U}_L^{\dag} [O]$. In the absence of noise, the next layers of the TEM map just undo this transformation; however, in the presence of noise, the Pauli map ${\cal N}_l^{-1}$ rescales every Pauli string $\sigma_{\boldsymbol{\beta}}$ by its inverse fidelity, i.e., it acts with the matrix ${\rm diag}(\{f_{l\boldsymbol{\beta}}^{-1}\}_{\boldsymbol{\beta}}) \equiv D_l$ in the Bloch vector picture. The resulting transformation of the Bloch vector reads ${\bf w} \rightarrow (\overleftarrow{\prod}_l C_l^{-1} D_l) (\overrightarrow{\prod}_l C_l) {\bf w}$. Recall that $\ln f_{l\boldsymbol{\beta}} \approx -\frac{1}{2} \ln \gamma_l \pm \frac{1}{2} \sqrt{ \sum_{\lambda_{\boldsymbol{\alpha}} \in {\rm SPL}({\cal N}_l)} \lambda_{\boldsymbol{\alpha}}^2 }$ (see Sec.~\ref{section-noise-model}) and $\sum_{\lambda_{\boldsymbol{\alpha}} \in {\rm SPL}({\cal N}_l)} \lambda_{\boldsymbol{\alpha}}^2 = \frac{1}{18} \varepsilon^2 N$ if the one-qubit (two-qubit) relaxation rates $\lambda_{\boldsymbol{\alpha}}^{(l)}$ in the SPL model are normally distributed with both mean and standard deviation equal to $\varepsilon/12$ ($\varepsilon/36$) similarly to the experiment~\cite{van_den_berg_probabilistic_2023}. This means $\ln f_{l\boldsymbol{\beta}} \approx - \frac{1}{2} \varepsilon N \pm \frac{\varepsilon}{6} \sqrt{ \frac{N}{2} }$ and $\varepsilon \sqrt{N} \ll 1$ in the relevant scenarios ($\varepsilon \sqrt{N} = 0.016$ if $\varepsilon = 0.16\%$, $N=100$, despite the fact that the average number of errors in a depth-$100$ circuit is $\varepsilon N L = 16 \gg 1$). Since $\varepsilon \sqrt{N} \ll 1$, $D_l \approx \exp(\varepsilon N/2) \left( I + \frac{\varepsilon}{6} \sqrt{ \frac{N}{2} } \Xi_l \right)$, where $\Xi_l$ is a diagonal matrix whose diagonal elements are random variables with mean $0$ and standard deviation $\sim 1$. Note that the exponential factor in $D_l$ is layer-independent, whereas the perturbation is layer dependent. In the first-order expansion with respect to perturbations $\sim\varepsilon \sqrt{N}$, we get $(\overleftarrow{\prod}_l C_l^{-1} D_l) (\overrightarrow{\prod}_l C_l) = \exp(\varepsilon NL/2) \left[ I + \frac{\varepsilon}{6} \sqrt{ \frac{N}{2} } \sum_l C_{\geq l}^{-1} \Xi_l C_{\geq l} \right]$ with $C_{\geq l} = \overrightarrow{\prod}_{m\geq l} C_m$. The total Bloch vector transformation in TEM is therefore ${\bf w} \rightarrow \exp(\varepsilon NL/2) \left[ {\bf w} + \frac{\varepsilon}{6} \sqrt{ \frac{N}{2} } \sum_l {\bf w}_l \right]$, where ${\bf w}_l = C_{\geq l}^{-1} \Xi_l C_{\geq l} {\bf w}$ is a random vector with an average square norm $\mathbb{E}(\|{\bf w}_l\|^2) = \|{\bf w}\|^2$ that is independent of $l$. Assuming high complexity of the quantum circuit, the orthogonal matrices $\{C_{\geq l}\}_{l=1}^L$ form an approximate 2-design~\cite{harrow_random_2009}, so the average scalar product $\mathbb{E}({\bf w}_l \cdot {\bf w}_{l^{\prime}}) = \delta_{ll^{\prime}} \|{\bf w}\|^2$. Hence, $\mathbb{E} ( \| \sum_l {\bf w}_l \|^2 ) = L \|{\bf w}\|^2$ and typically $\| \sum_l {\bf w}_l \| \sim \sqrt{L} \|{\bf w}\|$. Finally, the TEM transformation for Bloch vectors is ${\bf w} \rightarrow \exp(\varepsilon NL/2) \left[ {\bf w} + \frac{\varepsilon}{6} \sqrt{ \frac{NL}{2} } {\bf w}^{\prime} \right]$ with $\|{\bf w}^{\prime}\| \approx \|{\bf w}\|$. In a typical large-scale scenario with $\varepsilon = 0.16\%$ and $N=L=100$, the weight $\frac{\varepsilon}{6} \sqrt{ \frac{NL}{2} } = 0.019 \ll 1$ justifying the use of the perturbation expansion. The TEM-modified observable ${\cal M}_L^{\dag}[O]$ is a sum of a rescaled original observable $\exp(\varepsilon NL/2) O$ and $\frac{\varepsilon}{6} \sqrt{ \frac{NL}{2} } \exp(\varepsilon NL/2) O'$ with $O' = \sum_{\boldsymbol{\beta}} w_{\boldsymbol{\beta}}^{\prime} \sigma_{\boldsymbol{\beta}}$. 

The estimation-theory considerations of noise mitigation cost~\cite{tsubouchi_universal_2023} allow general multi-qubit measurements. In this scenario, the observables $O$ and $O'$ can be measured separately at the output of the noisy quantum processor. By optimally distributing the total budget of $M$ shots between the estimations of $O$ and $O'$ (proportionally to their weights in the modified observable ${\cal M}_L^{\dag}[O]$), we get the minimum random error $\|{\bf w}\| \sqrt{ \Gamma / M}$, with sampling overhead $\Gamma = \left( 1 + \frac{\varepsilon}{6} \sqrt{ \frac{NL}{2} } \right)^2 \exp(\varepsilon NL) \approx \left( 1 + \frac{\varepsilon}{3} \sqrt{ \frac{NL}{2} } \right) \exp(\varepsilon NL)$, which means a deviation of $3.8\%$ from $\exp(\varepsilon NL) = 8.9 \times 10^6$ for $\varepsilon = 0.16\%$ and $N=L=100$. 

However, since general multi-qubit measurements are not available in quantum processors, a more cost-efficient strategy is to measure only the largest contribution to ${\cal M}_L^{\dag}[O]$, i.e., only those Pauli strings $\sigma_{\boldsymbol{\beta}}$ that are present in the original observable $O = \sum_{\boldsymbol{\beta}} w_{\boldsymbol{\beta}} \sigma_{\boldsymbol{\beta}}$. TEM provides a rescaling factor $\big( {\cal M}_L^{\dag} \big)_{\boldsymbol{\beta\beta}} \approx \exp(\varepsilon NL/2)$ for each $\sigma_{\boldsymbol{\beta}}$ by merely fixing the input and output indices of TEM map ${\cal M}_L^{\dag}$ to $\boldsymbol{\beta}$. The factors are scattered around the exponent $\exp(\varepsilon NL/2)$ with the standard deviation $\sim \frac{\varepsilon}{6} \sqrt{ \frac{NL}{2} } \exp(\varepsilon NL/2)$, which also grows exponentially in the circuit size. This is why the tensor-network contraction for ${\cal M}_L^{\dag}$ is needed to find these factors with high accuracy. The resulting sampling overhead in this experimentally friendly scenario is $\Gamma = \exp(\varepsilon NL)$ on average. The replacement of the true modified observable ${\cal M}_L^{\dag}[O] = {\cal M}_L^{\dag}[\sum_{\boldsymbol{\beta}} w_{\boldsymbol{\beta}} \sigma_{\boldsymbol{\beta}}] = \sum_{\boldsymbol{\beta}} w_{\boldsymbol{\beta}} {\cal M}_L^{\dag}[\sigma_{\boldsymbol{\beta}}]$ with its approximation $\sum_{\boldsymbol{\beta}} w_{\boldsymbol{\beta}} \big( {\cal M}_L^{\dag} \big)_{\boldsymbol{\beta\beta}} \sigma_{\boldsymbol{\beta}}$ is well justified since the measurement is performed on the highly noisy state $\varrho_{\rm noisy} \equiv \overleftarrow{\bigcirc}_l {\cal U}_l \circ {\cal N}_l [(\ket{0}\bra{0})^{\otimes N}]$ and the incurred systematic error is $\epsilon^2 N L / 72 = 0.036\%$ for $\varepsilon = 0.16\%$ and $N=L=100$ (see Sec.~\ref{section-tem-systematic-error}). 

\subsubsection{Optimality with respect to the sampling overhead} \label{section-tem-optimality}

It is instructive to compare TEM's sampling overhead with the fundamental cost bound of quantum error mitigation based on quantum estimation theory~\cite{tsubouchi_universal_2023}. The result of Ref.~\cite{tsubouchi_universal_2023} is built around the quantum Cram\'{e}r-Rao inequality for an unbiased estimator of the mitigated observable. Since the result concerns the number of measurement shots only and does not restrict classical postprocessing power in any way, it is applicable to the compression-free TEM that provides an unbiased estimator for the mitigated observable if the noise model is perfectly known. Theorem~1 in Ref.~\cite{tsubouchi_universal_2023} establishes the fact that the minimum number of shots $|{\sf S}|_{\min}$ required for an unbiased estimation of the mitigated observable with standard deviation $\delta$ grows exponentially in the circuit depth. For instance, if all unitary layers in a depth-$L$ quantum circuit are accompanied by global depolarizing noise $\varrho \rightarrow (1-p)\varrho + p \, {\rm tr}[\varrho] \frac{1}{2^N} I^{\otimes N}$, then for any nontrivial Pauli observable $|{\sf S}|_{\min} \geq \frac{1}{\delta^2} \left( 1 - (1-p)^L \right) (1-p)^{-2L} \approx \frac{1}{\delta^2} (1-p)^{-2L}$ if $pL \gg 1$ and $p \ll 1$~\cite{tsubouchi_universal_2023}, which implies the lower bound on the sampling overhead $\Gamma_{\min} \approx (1-p)^{-2L}$. Since the observable damping $K = (1-p)^L$, the modified observable ${\cal M}_L^{\dag} [O] = K^{-1} O = (1-p)^{-L} O$, and TEM's sampling overhead $\Gamma = (1-p)^{-2L}$ saturates the universal lower bound $\Gamma_{\min}$. As a side remark, we note that TEM with bond dimension $\chi=2$ is compression-free for global depolarizing noise~\cite{filippov_scalable_2023}.

In the case of local noise ${\cal N} = \Lambda^{\otimes N}$, Ref.~\cite{tsubouchi_universal_2023} reveals the exponential growth of the minimum sampling overhead in the total circuit area ($NL$) and specifies this finding for some standard noisy maps $\Lambda$ in the asymptotic regime $N,L \gg 1$. Namely, $\Gamma_{\min} \approx (1 + \frac{3}{2}p)^{NL}$ for local depolarizing noise with intensity $p$ ($\Lambda^{\dag}[X] = (1-p) X$, $\Lambda^{\dag}[Y] = (1-p) Y$, $\Lambda^{\dag}[Z] = (1-p) Z$) and $\Gamma_{\min} \approx (1 + p)^{NL}$ for local amplitude damping noise with intensity $p$ ($\Lambda^{\dag}[X] = \sqrt{1-p} X$, $\Lambda^{\dag}[Y] = \sqrt{1-p} Y$, $\Lambda^{\dag}[Z] = (1-p) Z + p I$). The local depolarizing noise is a particular case of SPL noise with $\gamma = (1 - p)^{-3NL/2} \approx (1 + \frac{3}{2}p)^{NL}$~\cite{temme-2017}, so the sampling overhead in TEM $\Gamma = \gamma$ again saturates the lower bound $\Gamma_{\min}$. The amplitude damping map is converted, via Pauli twirling, to the generalized amplitude damping map with the infinite-temperature bath (the unital map $\widetilde{\Lambda}^{\dag}[X] = \sqrt{1-p} X$, $\widetilde{\Lambda}^{\dag}[Y] = \sqrt{1-p} Y$, $\widetilde{\Lambda}^{\dag}[Z] = (1-p) Z$~\cite[section 8.3.5]{nielsen-2010}), which is SPL noise with $\gamma = (1-p)^{-NL} \approx (1+p)^{NL}$. So TEM's sampling overhead $\Gamma = \gamma$ saturates the lower bound $\Gamma_{\min}$ in this case too. These examples just illustrate the following general result.

\begin{proposition} \label{proposition-tem-overhead}
The sampling overhead $\Gamma$ in TEM asymptotically saturates the lower cost bound $\Gamma_{\min} = \mathbb{E}_{{\cal U}_1 \cdots {\cal U}_L} ( \delta^2 |{\sf S}|_{\min})$ for the unbiased estimation of Pauli observables with the standard deviation $\delta$ in large-scale quantum circuits with weak sparse Pauli-Lindblad noise and unitary layers ${\cal U}_1, \ldots, {\cal U}_L$ drawn from a set forming a unitary 2-design.
\end{proposition}
\begin{proof}
Elementary one-qubit maps ${\cal N}^{[q]}_{l}$ and two-qubit maps ${\cal N}^{\langle q_1,q_2\rangle}_{l}$ constituting the SPL noise layer ${\cal N}_l = \bigcirc_{q} {\cal N}^{[q]}_{l} \circ \bigcirc_{\langle q_1,q_2\rangle} {\cal N}^{\langle q_1,q_2\rangle}_{l}$ are unital and invertible. Therefore, the noisy quantum circuit satisfies conditions of Theorem~S1 in Ref.~\cite[Supplemental material]{tsubouchi_universal_2023} yielding $\Gamma_{\min} \approx \prod_l \prod_{q} \nu( {\cal N}^{[q] - 1}_{l} ) \prod_{\langle q_1,q_2\rangle} \nu ( {\cal N}^{\langle q_1,q_2\rangle -1}_{l} )$ in the asymptotic limit $N \gg 1$. Here, $\nu(\Lambda) \equiv {\rm tr}[\Omega_{\Lambda}^2]$ is the purity parameter~\footnote{If $\Lambda$ is completely positive and trace preserving, then $\nu(\Lambda) \leq 1$; however, if $\Lambda$ is the inverse of a completely positive and trace preserving map, then $\nu(\Lambda) \geq 1$.} of the normalized Choi operator $\Omega_{\Lambda} \equiv \Lambda \otimes {\rm Id} [\ket{\psi_+}\bra{\psi_+}]$, $\ket{\psi_+}$ is the maximally entangled state in the Hilbert space ${\cal H}^{\otimes 2}$, ${\rm dim}{\cal H} = 2^m$ for an $m$-qubit map $\Lambda$. If $\Lambda$ is an $m$-qubit Pauli map, then $\nu(\Lambda) = \frac{1}{4^m} \sum_{\boldsymbol{\alpha}=(\alpha_0,\ldots,\alpha_{m-1})} f_{\boldsymbol{\alpha}}^2$ is the regularized sum of squares of all Pauli fidelities. In the case of SPL noise, $\nu ( {\cal N}^{[q] - 1}_{l} ) = \frac{1}{4} \{ 1 + \exp[4(\lambda_Y^{[q]} + \lambda_Z^{[q]})] + \exp[4(\lambda_X^{[q]} + \lambda_Z^{[q]})] + \exp[4(\lambda_X^{[q]} + \lambda_Y^{[q]})]\} \approx 1 + 2(\lambda_X^{[q]} + \lambda_Y^{[q]} + \lambda_Z^{[q]}) \approx \exp[2(\lambda_X^{[q]} + \lambda_Y^{[q]} + \lambda_Z^{[q]})] \equiv \gamma({\cal N}^{[q]}_{l})$, where the approximate equalities hold true in the first order of the decoherence rates (weak noise). Using the explicit expressions for 2-qubit fidelities~\cite[Appendix F 4]{filippov_scalable_2023}, we similarly get $\nu ( {\cal N}^{\langle q_1,q_2\rangle -1}_{l} ) \approx \exp[2 \sum_{\lambda \in {\rm SPL}({\cal N}^{\langle q_1,q_2\rangle}_l)} \lambda] \equiv \gamma({\cal N}^{\langle q_1,q_2\rangle -1}_{l})$. Finally, $\Gamma_{\min} \approx \prod_l \prod_{q} \gamma( {\cal N}^{[q]}_{l} ) \prod_{\langle q_1,q_2\rangle} \gamma ( {\cal N}^{\langle q_1,q_2\rangle}_{l} ) \equiv \gamma$, which is the sampling overhead in TEM.
\end{proof}

In fact, the analysis of the sampling overhead optimality can be extended beyond 2-local SPL noise models, e.g., to a more complex yet learnable SPL model, in which the jump operators have higher-weight terms~\cite{berg_techniques_2023,rouze2023efficient,flammia_2020}. In what follows, we analyze a general Pauli noise that describes the effects beyond nearest-neighbor interactions. 

\begin{proposition} \label{proposition-tem-overhead-general}
The sampling overhead $\Gamma$ in TEM asymptotically saturates the lower cost bound $\Gamma_{\min} = \mathbb{E}_{{\cal U}_1 \cdots {\cal U}_L} ( \delta^2 |{\sf S}|_{\min})$ for an unbiased estimation of Pauli observables with the standard deviation $\delta$ in large-scale quantum circuits with a general weak Pauli noise and unitary layers ${\cal U}_1, \ldots, {\cal U}_L$ drawn from a set forming a unitary 2-design.
\end{proposition}
\begin{proof}
Since TEM's sampling overhead $\gamma=\prod_l \gamma_l$ and the lower bound $\Gamma_{\min}$ from Ref.~\cite{tsubouchi_universal_2023} are both multiplicative with respect to circuit layers, it suffices to establish the equivalence between them for a single noisy layer. A Pauli noise layer ${\cal N}_l$ is fully described by the set of fidelities $\{f_{l\boldsymbol{\beta}}\}_{\boldsymbol{\beta}}$, where $f_{l\boldsymbol{\beta}} = \exp(-2\sum_{\boldsymbol{\alpha}: \{\sigma_{\boldsymbol{\alpha}},\sigma_{\boldsymbol{\beta}}\}=0 } \lambda_{\boldsymbol{\alpha}}) \approx 1 - 2 \sum_{\boldsymbol{\alpha}: \{\sigma_{\boldsymbol{\alpha}},\sigma_{\boldsymbol{\beta}}\}=0 } \lambda_{\boldsymbol{\alpha}}$ if the noise is weak. If $N \gg 1$, then by the same arguments as in the proof of Proposition~\ref{proposition-tem-overhead}, $\Gamma_{l \, \min} \approx \nu({\cal N}_l^{-1}) = 4^{-N} \sum_{\boldsymbol{\beta}} f_{l\boldsymbol{\beta}}^{-2} \approx 1 + 4^{1-N} \sum_{\boldsymbol{\beta}} \sum_{\boldsymbol{\alpha}: \{\sigma_{\boldsymbol{\alpha}},\sigma_{\boldsymbol{\beta}}\}=0 } \lambda_{l\boldsymbol{\alpha}}$. Since any nontrivial Pauli string anticommutes with exactly half of Pauli strings (see Appendix~\ref{appendix-anticommutativity}), $\Gamma_{l \, \min} \approx 1 + 2 \sum_{\boldsymbol{\beta}} \lambda_{l\boldsymbol{\beta}}$. On the other hand, $\gamma_l = \exp(2\sum_{\boldsymbol{\beta}} \lambda_{\boldsymbol{\beta}}) \approx 1 + 2 \sum_{\boldsymbol{\beta}} \lambda_{\boldsymbol{\beta}}$ if $\sum_{\boldsymbol{\beta}} \lambda_{\boldsymbol{\beta}} \ll 1$ (noise weakness). Therefore, $\Gamma_{l \, \min} \approx \gamma_l$, the typical sampling overhead in TEM for the $l$th layer.
\end{proof}

We stress that our results on the optimality of the sampling overhead in TEM are applicable to the realistic SPL noise model justified experimentally~\cite{kim_evidence_2023,van_den_berg_probabilistic_2023}, where $X$-, $Y$-, and $Z$-errors are all present and are roughly of the same order $\sim \varepsilon/12$ on average (similarly, the correlated $XX$-, $XY$-, \ldots, and $ZZ$-errors are all present and are roughly of the same order $\sim \varepsilon/36$ on average). Should the noise be of a different form, for example, a non-realistic qubit-local pure dephasing noise $\Lambda[\bullet] = (1-p)\bullet + p Z \bullet Z$ of intensity $p \ll 1$, then the sampling overheads for PEC and TEM would coincide and be equal to $\exp(2pNL)$. (PEC's sampling overhead for this case is calculated, e.g., in Ref.~\cite{takagi_fundamental_2022}. TEM's sampling overhead results from a typical damping $(1-p)^{NL}$ of the dynamical observable in the Heisenberg picture. Only half of the Pauli string components in the dynamical observable is affected by the qubit-local pure dephasing noise: $\Lambda^{\dag}[I] = I$, $\Lambda^{\dag}[X] = (1-2p) X$, $\Lambda^{\dag}[Y] = (1-2p) Y$, $\Lambda^{\dag}[Z] = Z$.) Therefore, the optimality of PEC's sampling overhead in the case of local dephasing noise~\cite{takagi_fundamental_2022} is actually accompanied by the optimality of TEM's sampling overhead in this case too. However, if we return to the realm of realistic SPL noise~\cite{kim_evidence_2023,van_den_berg_probabilistic_2023}, then TEM provides a quadratically smaller sampling overhead with respect to PEC (see Table~\ref{table-summary}), excluding the optimality of PEC sampling overhead. The global depolarizing noise provides another analytical model showcasing the optimality of the sampling overhead in TEM~\cite{filippov_scalable_2023} and the non-optimality of PEC's~\cite{takagi_fundamental_2022}.

\subsubsection{Systematic error} \label{section-tem-systematic-error}

\textbf{Effect of imprecisely learned noise model and noise instability}. The first source of systematic error in TEM is the same as in PEC and ZNE, namely, noise instability during the experiment and inaccuracies in the learned noise model. The associated relative error is derived in exactly the same way as in PEC (see Sec.~\ref{section-pec-systematic-error}) and reads $\Delta \bar{O} / \bar{O} = \frac{1}{2} \varepsilon N \sqrt{L} \Theta$ provided $\Delta \bar{O} / \bar{O} \ll 1$, where $\Theta$ is the standard deviation in the relative error per layer per qubit.

\textbf{Effect of simplified measurement strategy (if used)}. In non-Clifford circuits, TEM modifies an observable $O = \sum_{\boldsymbol{\beta}} w_{\boldsymbol{\beta}} \sigma_{\boldsymbol{\beta}}$ into $\sum_{\boldsymbol{\beta}} w_{\boldsymbol{\beta}} {\cal M}_L^{\dag}[\sigma_{\boldsymbol{\beta}}]$, which generally represents a sum of exponentially many Pauli strings. If generalized measurements are not available, then one can instead measure a simpler observable $\sum_{\boldsymbol{\beta}} w_{\boldsymbol{\beta}} \big( {\cal M}_L^{\dag} \big)_{\boldsymbol{\beta\beta}} \sigma_{\boldsymbol{\beta}}$ at the price of some systematic error. 

To evaluate the systematic error of this kind, consider the same ``mirrored'' non-Clifford mirror circuit as in the ZNE analysis (see Sec.~\ref{section-zne-systematic-error}), where $N_T$ qubits are subjected to $T$-gates. In the Heisenberg picture along the TEM map ${\cal M}_L^{\dag} = \overleftarrow{\bigcirc}_l {\cal U}_l \circ ({\cal N}_l^{-1})^{\dag} \circ \overrightarrow{\bigcirc}_l {\cal U}_l^{\dag}$, the original observable $O = Z^{\otimes N}$ is first left unperturbed by a noiseless circuit, $\overrightarrow{\bigcirc}_l {\cal U}_l^{\dag} [Z^{\otimes N}] = Z^{\otimes N}$. Qubit-local operators $Z$ in the Pauli string $Z^{\otimes N}$ are then transformed into either $THZHT^{\dag} = \frac{1}{\sqrt{2}}(X + Y)$ (for $N_T$ qubits) or $HZH = X$ (for $N-N_T$ qubits) by the first noiseless layer ${\cal U}_1$. Then, these $2^{N_T}$ different Pauli strings (composed of $X$ and $Y$ operators) propagate through the inverted noisy Clifford part of the circuit in the Heisenberg picture, $\overrightarrow{\bigcirc}_{l=2}^{L-1} {\cal U}_l \circ ({\cal N}_l^{-1})^{\dag}$. The $n$th Pauli string out of $2^{N_T}$ gathers a specific rescaling factor $c_n \equiv \prod_{l} f_{l\boldsymbol{\beta}_n(l)}^{-1}$ from the SPL noise, with the fidelities being given by Eq.~\eqref{Pauli-fidelity}. The final unitary layer of TEM map in the Heisenberg picture performs the transformation $HT^{\dag} \bullet TH$ for $N_T$ qubits and $H \bullet H$ for the rest $N-N_T$ qubits. Since $HT^{\dag}\frac{1}{\sqrt{2}}(c_X X + c_Y Y)TH = \frac{1}{2}(c_X + c_Y)Z + \frac{1}{2}(c_X - c_Y)Y$, the compression-free TEM outputs the operator 
\begin{eqnarray}
{\cal M}_L^{\dag} [O] &=& \overleftarrow{\bigcirc}_l {\cal U}_l \circ ({\cal N}_l^{-1})^{\dag} [Z^{\otimes N}] \nonumber\\
&=& \underbrace{ \frac{1}{2^{N_T}} \left( \sum_{n \in \{X,Y\}^{N_T}} c_n \right) }_{\sim \exp(\varepsilon N L /2)} Z^{\otimes N} + \sum_{k \in \{Z,Y\}^{N_T}} \underbrace{ \frac{1}{2^{N_T}} \left( \sum_{n \in \{X,Y\}^{N_T}}  (-1)^{\langle k, n \rangle_Y} c_n \right) }_{\sim \pm \frac{1}{2^{N_T/2}} \varepsilon \sqrt{NL} \exp(\varepsilon N L /2)} \sigma_k \otimes Z^{\otimes (N-N_T)}, \label{non-Clifford-modified-observable-last-line}
\end{eqnarray}

\noindent where $\langle k, n \rangle_Y$ is the number of positions $i\in\{1,\ldots,N_T\}$ in the strings $k$ and $n$, where $k_i = n_i = Y$ (e.g., if $N_T = 2$, $\langle ZY, XY \rangle_Y = 1$ and $\langle ZY, YX \rangle_Y = 0$). Recalling the averaging effect in the product of fidelities (see Sec.~\ref{section-noise-model}), the first underbrace term $\frac{1}{2^{N_T}} \sum_{n \in \{X,Y\}^{N_T}} c_n \approx \left( 1 \pm \frac{\varepsilon}{6} \sqrt{\frac{NL}{2^{N_T+1}}} + \frac{\varepsilon^2 NL}{72} \right) \exp(\varepsilon N L / 2)$ and the second underbrace term $\frac{1}{2^{N_T}} \sum_{n \in \{X,Y\}^{N_T}} (-1)^{\langle k, n \rangle_Y} c_n \approx \pm \frac{\varepsilon}{6} \sqrt{\frac{NL}{2^{N_T+1}}} \exp(\varepsilon N L / 2)$. A simpler observable to measure is the TEM-rescaled version of the original one, $({\cal M}_L^{\dag})_{Z^{\otimes N},Z^{\otimes N}} Z^{\otimes N}$, with
\begin{equation} \label{tem-rescaled-original-O}
    ({\cal M}_L^{\dag})_{Z^{\otimes N},Z^{\otimes N}} = \frac{1}{2^{N_T}} \left( \sum_{n \in \{X,Y\}^{N_T}} c_n \right) \approx \left( 1 \pm \frac{\varepsilon}{6} \sqrt{\frac{NL}{2^{N_T+1}}} + \frac{\varepsilon^2 NL}{72} \right) \exp(\varepsilon N L / 2).
\end{equation}
As we will see in the following analysis, measuring $({\cal M}_L^{\dag})_{Z^{\otimes N},Z^{\otimes N}} Z^{\otimes N}$ instead of ${\cal M}_L^{\dag} [O]$ at the output of the noisy quantum computation leads to a small systematic error $\propto \varepsilon^2 N L$.

The initial state $(\ket{0}\bra{0})^{\otimes N} = \frac{1}{2^N} \sum_{k \in \{I,Z\}^N} \sigma_k$ is a sum of $2^N$ Pauli strings. The noisy quantum computation $\varrho_{\rm noisy} \equiv \overleftarrow{\bigcirc}_l {\cal U}_l \circ {\cal N}_l [(\ket{0}\bra{0})^{\otimes N}]$ preserves qubit-local identity operators $I$ in $2^N-1$ Pauli strings $\{\sigma_k\}_{k \in \{I,Z\}^N,  \, k \neq Z \ldots Z}$ that originally contain at least one identity operator, and these $I$-containing Pauli strings are orthogonal to ${\cal M}_L^{\dag} [O]$ in Eq.~\eqref{non-Clifford-modified-observable-last-line} in the sense of the Hilbert-Schmidt inner product. Therefore, the only component $[\varrho_{\rm noisy}]_{{\cal M}_L^{\dag} [O]}$ of $\varrho_{\rm noisy}$ that overlaps with ${\cal M}_L^{\dag} [O]$ is 
\begin{eqnarray}
    && [\varrho_{\rm noisy}]_{{\cal M}_L^{\dag} [O]} = \frac{1}{2^{N}} \overleftarrow{\bigcirc}_l {\cal U}_l {\cal N}_l [Z^{\otimes N}] \nonumber\\
    && = \underbrace{ \frac{1}{2^{N_T}} \left( \sum_{n \in \{X,Y\}^{N_T}} \prod_{l} f_{l\boldsymbol{\beta}_n(l)} \right) }_{\sim \exp(- \varepsilon N L /2)}  \frac{Z^{\otimes N}}{2^{N}} + \sum_{k \in \{Z,Y\}^{N_T}} \underbrace{ \frac{1}{2^{N_T}} \left( \sum_{n \in \{X,Y\}^{N_T}} (-1)^{\langle k, n \rangle_Y} \prod_{l} f_{l\boldsymbol{\beta}_n(l)} \right) }_{\sim \mp \frac{1}{2^{N_T/2}} \varepsilon \sqrt{NL} \exp(- \varepsilon N L /2)} \frac{\sigma_k \otimes Z^{\otimes (N-N_T)}}{2^{N}}, \qquad \label{non-Clifford-modified-state-last-line}
\end{eqnarray}

\noindent with the first underbrace term being $\frac{1}{2^{N_T}} \sum_{n \in \{X,Y\}^{N_T}} \prod_{l} f_{l\boldsymbol{\beta}_n(l)} \approx \left( 1 \mp \frac{\varepsilon}{6} \sqrt{\frac{NL}{2^{N_T+1}}} \right) \exp(-\varepsilon N L / 2)$, the second underbrace term being $\frac{1}{2^{N_T}} \sum_{n \in \{X,Y\}^{N_T}} (-1)^{\langle k, n \rangle_Y} \prod_{l} f_{l\boldsymbol{\beta}_n(l)} \approx \mp \frac{\varepsilon}{6} \sqrt{\frac{NL}{2^{N_T+1}}} \exp(-\varepsilon N L / 2)$, and the signs in Eqs.~\eqref{non-Clifford-modified-observable-last-line} and \eqref{non-Clifford-modified-state-last-line} being correlated.

In the compression-free TEM, measurement of the modified observable ${\cal M}_L^{\dag} [O]$ in the noisy state $\varrho_{\rm noisy}$ gives the ideal noiseless expectation value, i.e., $\braket{{\cal M}_L^{\dag}[O]}_{\rm noisy} = {\rm tr}\big[ \varrho_{\rm noisy} {\cal M}_L^{\dag} [O] \big] = 1$ in our example. Instead, measurement of the TEM-rescaled observable~\eqref{tem-rescaled-original-O} at the output of the of noisy quantum processor gives
\begin{eqnarray}
    \braket{({\cal M}_L^{\dag})_{Z^{\otimes N},Z^{\otimes N}} Z^{\otimes N}}_{\rm noisy} = {\rm tr}[\varrho_{\rm noisy} Z^{\otimes N}] ({\cal M}_L^{\dag})_{Z^{\otimes N,Z^{\otimes N}}} &=& \frac{1}{2^{2N_T}} \left( \sum_{n \in \{X,Y\}^{N_T}} \prod_{l} f_{l\boldsymbol{\beta}_n(l)} \right) \left( \sum_{n \in \{X,Y\}^{N_T}} c_n \right) \nonumber\\
    & \approx & 1 + \frac{\varepsilon^2 NL}{72} \left(1-\frac{1}{2^{N_T}} \right) \approx 1 + \frac{\varepsilon^2 NL}{72} \text{~if~} N_T \gg 1. \label{TEM-rescaled-O-average}
\end{eqnarray}
Eq.~\eqref{TEM-rescaled-O-average} provides the evaluation of systematic error, $\varepsilon^2 NL /72$, emerging due to replacement of the true TEM-modified observable $\sum_{\boldsymbol{\beta}} w_{\boldsymbol{\beta}} {\cal M}_L^{\dag}[\sigma_{\boldsymbol{\beta}}]$ by the simpler TEM-rescaled version $\sum_{\boldsymbol{\beta}} w_{\boldsymbol{\beta}} \big( {\cal M}_L^{\dag} \big)_{\boldsymbol{\beta\beta}} \sigma_{\boldsymbol{\beta}}$ of the original observable $O = \sum_{\boldsymbol{\beta}} w_{\boldsymbol{\beta}} \sigma_{\boldsymbol{\beta}}$.

The derived systematic error is proportional to $\varepsilon^2 NL$ and vanishes in the limit $\varepsilon \rightarrow 0, NL \rightarrow \infty$, $\varepsilon N L = {\rm const}$. This means that while the improvements in technology lead to smaller values of the error rate $\varepsilon$ and the circuits of larger scale $NL \sim {\rm const}/ \varepsilon $ become accessible for noise mitigation, the systematic error due to the simplified measurement strategy gets smaller and smaller.

\textbf{Effect of bond dimension compression}.
If the tensor-network map ${\cal M}_L^{\dag (\chi)}$ is built with the use of bond dimension compression (bond dimension not exceeding $\chi$), then ${\cal M}_L^{\dag (\chi)}$ is an approximation to the exact map ${\cal M}_L^{\dag}$ and cannot exactly reproduce all the rescaling scaling factors $c_n \equiv \prod_{l} f_{l\boldsymbol{\beta}_n(l)}^{-1}$ but provides an approximation for them, $c_n^{(\chi)} \approx \prod_{l=1}^L f_{l\boldsymbol{\beta}_n(l)}^{-1}$. The mismatch between $c_n$ and $c_n^{(\chi)}$ results in a systematic error, which we evaluate in what follows. 

Consider the coefficient $c_n^{(\chi)}$ as a function of $\sim 12 N L$ noise variables $\lambda_{\boldsymbol{\alpha}}^{(l)} \in \cup_{l=1}^L {\rm SPL}({\cal N}_l)$. The bond dimension $\chi$ regulates the precision with which $c_n^{(\chi)}$ approximates $\prod_{l=1}^L f_{l\boldsymbol{\beta}_n(l)}^{-1}$, which is also a function of the same variables $\lambda_{\boldsymbol{\alpha}} \in \cup_{l=1}^L {\rm SPL}({\cal N}_l)$, see Eq.~\eqref{Pauli-fidelity}. Let ${\cal T}_1[\bullet]$ be a functional that returns the first Taylor polynomial of a function $\bullet$ with respect to variables $\lambda_{\boldsymbol{\alpha}} \in \cup_{l=1}^L {\rm SPL}({\cal N}_l)$. There exists a bond dimension $\chi^{\ast}$ such that the mathematical identity 
${\cal T}_1[c_n^{(\chi^{\ast})}] \equiv {\cal T}_1 \Big[ \prod_{l=1}^L f_{l\boldsymbol{\beta}_n(l)}^{-1} \Big]$ takes place. In Appendix~\ref{appendix-TEM-threshold-bond-dimension}, we show that $\chi^{\ast} \approx \frac{1}{2} L^2$ and does not depend on the number of qubits $N$. This implies that, in the approximation $c_n^{(\chi^{\ast})} \approx \prod_{l=1}^L f_{l\boldsymbol{\beta}_n(l)}^{-1}$, every exponent $\exp(-2 \lambda_{\alpha_q}^{[q]})$ and every exponent $\exp(-2 \lambda_{\alpha_{q_1},\alpha_{q_2}}^{\langle q_1,q_2\rangle})$ in the definition of $f_{l\boldsymbol{\beta}}$ in Eq.~\eqref{Pauli-fidelity} are reproduced exactly to first order; otherwise, the functional identity ${\cal T}_1[c_n^{(\chi^{\ast})}] \equiv {\cal T}_1 \Big[ \prod_{l=1}^L f_{l\boldsymbol{\beta}_n(l)}^{-1} \Big]$ does not hold. This means $c_n^{(\chi^{\ast})} = \prod_{l=1}^L \widetilde{f}_{l\boldsymbol{\beta}_n(l)}^{-1}$, with
\begin{equation*} \label{Pauli-fidelity-approximation}
\widetilde{f}_{l{\boldsymbol{\beta}}} = \prod_{q} \prod_{\alpha_q} \left\{ \begin{array}{ll}
    1 - 2 \lambda_{\alpha_q}^{[q]} + \mu_l^{[q]} & \text{if~} \{\sigma_{\alpha_q},\sigma_{\beta_q}\} \! = \! 0, \\
    1 + \mu_l^{[q]} & \text{otherwise}
\end{array} \right. \!\! \prod_{\langle q_1,q_2\rangle} \prod_{\alpha_{q_1},\alpha_{q_2}} \!\! \left\{ \begin{array}{ll}
    1 - 2 \lambda_{\alpha_{q_1},\alpha_{q_2}}^{\langle q_1,q_2\rangle} + \mu_l^{\langle q_1,q_2\rangle} & \text{if~} \{\sigma_{\alpha_{q_1}} \!\! \otimes \! \sigma_{\alpha_{q_2}}, \sigma_{\beta_{q_1}} \!\! \otimes \! \sigma_{\beta_{q_2}} \} \! = \! 0, \\
    1 + \mu_l^{\langle q_1,q_2\rangle} & \text{otherwise};
\end{array} \right.
\end{equation*}
where $\mu_l^{[q]}$ and $\mu_l^{\langle q_1,q_2\rangle}$ are the second-order terms proportional to $(\lambda_{\alpha_q}^{[q]})^2$ and $(\lambda_{\alpha_{q_1},\alpha_{q_2}}^{\langle q_1,q_2\rangle})^2$, respectively. On the other hand, application of Taylor's theorem with the Lagrange remainder to the exact fidelity~\eqref{Pauli-fidelity} yields
\begin{equation*} \label{Pauli-fidelity-Taylor-Lagrange}
f_{l{\boldsymbol{\beta}}} = \prod_{q} \prod_{\alpha_q: \ \{\sigma_{\alpha_q},\sigma_{\beta_q}\} = 0} \Big[ 1 - 2 \lambda_{\alpha_q}^{[q]} + \nu_l^{[q]} \Big] \prod_{\langle q_1,q_2\rangle} \prod_{\alpha_{q_1},\alpha_{q_2}: \ \{\sigma_{\alpha_{q_1}} \otimes \sigma_{\alpha_{q_2}}, \sigma_{\beta_{q_1}} \otimes \sigma_{\beta_{q_2}} \} = 0} \Big[ 1 - 2 \lambda_{\alpha_{q_1},\alpha_{q_2}}^{\langle q_1,q_2\rangle} + \nu_l^{\langle q_1,q_2\rangle} \Big],
\end{equation*}
for some $\nu_l^{[q]}$ and $\nu_l^{\langle q_1,q_2\rangle}$ satisfying $2 \exp(-2 \lambda_{\alpha_q}^{[q]}) (\lambda_{\alpha_q}^{[q]})^2 < \nu_l^{[q]} < 2 (\lambda_{\alpha_q}^{[q]})^2$ and $2 \exp(- 2 \lambda_{\alpha_{q_1},\alpha_{q_2}}^{\langle q_1,q_2\rangle})(\lambda_{\alpha_{q_1},\alpha_{q_2}}^{\langle q_1,q_2\rangle})^2 < \nu_l^{\langle q_1,q_2\rangle} < 2 (\lambda_{\alpha_{q_1},\alpha_{q_2}}^{\langle q_1,q_2\rangle})^2$. We assume $\mu_l^{[q]} \sim \left\vert \mu_l^{[q]} - \nu_l^{[q]} \right\vert \sim (\lambda_{\alpha_q}^{[q]})^2$ and $\mu_l^{\langle q_1,q_2\rangle} \sim \left\vert \mu_l^{\langle q_1,q_2\rangle} - \nu_l^{\langle q_1,q_2\rangle} \right\vert \sim (\lambda_{\alpha_{q_1},\alpha_{q_2}}^{\langle q_1,q_2\rangle})^2$, which results in 
\begin{equation} \label{TEM-systematic-c-coefficient}
\widetilde{f}_{l{\boldsymbol{\beta}}}^{-1} f_{l{\boldsymbol{\beta}}} \approx  \prod_{q} \prod_{\alpha_q} \Big[ 1 \pm a (\lambda_{\alpha_q}^{[q]})^2 \Big] \prod_{\langle q_1,q_2\rangle} \prod_{\alpha_{q_1},\alpha_{q_2}} \Big[ 1 \pm a (\lambda_{\alpha_{q_1},\alpha_{q_2}}^{\langle q_1,q_2\rangle})^2 \Big] \approx 1 \pm a \sum_{\lambda_{\boldsymbol{\alpha}}^{(l)} \in {\rm SPL}({\cal N}_l)} (\lambda_{\boldsymbol{\alpha}}^{(l)})^2,
\end{equation}
where the coefficient $a$ is to be determined from a numerical experiment. In the last expression, we have taken into account that the signs $\pm$ are controlled by the bond dimension compression and may potentially be all the same. Since $c_n^{(\chi)} \prod_{l=1}^L f_{l\boldsymbol{\beta}_n(l)} \approx 1 \pm a \sum_{l=1}^L \sum_{\lambda_{\boldsymbol{\alpha}}^{(l)} \in {\rm SPL}({\cal N}_l)} (\lambda_{\boldsymbol{\alpha}}^{(l)})^2$, $c_n^{(\chi)}$ reproduces $c_n$ with relative error $\sum_{\lambda_{\boldsymbol{\alpha}} \in \cup_l {\rm SPL}({\cal N}_l)} \lambda_{\boldsymbol{\alpha}}^2$ for every $n$ and the resulting relative systematic observable estimation error is $\delta_{\rm systematic}^{(3)} = a \sum_{\lambda_{\boldsymbol{\alpha}} \in \cup_l {\rm SPL}({\cal N}_l)} \lambda_{\boldsymbol{\alpha}}^2$. In Appendix~\ref{appendix-TEM-threshold-bond-dimension}, we present a numerical experiment with random Clifford circuits that justifies the derived evaluation of the systematic error with $a = 0.6$. Using this numerically inferred value and assuming that the one-qubit (two-qubit) relaxation rates $\lambda_{\boldsymbol{\alpha}}$ in the SPL model are normally distributed with both mean and standard deviation equal to $\varepsilon/12$ ($\varepsilon/36$), similarly to the experiment~\cite{van_den_berg_probabilistic_2023}, we get $\delta_{\rm systematic}^{(3)} = \varepsilon^2 N L / 30$.

The physical meaning of the error $\delta_{\rm systematic}^{(3)} = 0.6 \sum_{\lambda_{\boldsymbol{\alpha}} \in \cup_l {\rm SPL}({\cal N}_l)} \lambda_{\boldsymbol{\alpha}}^2$ is the same as in the quantum error correction with a code of distance $3$, i.e., a code that is able to correct $1$ error at a time among the physical qubits encoding a single logical qubit. For such a code, there is still a nonzero (although much smaller) probability of $2$ or more errors happening simultaneously. Therefore, the code corrects the first leading term in every local noise contribution and the result of error correction is the ideal unitary circuit whose layers are intervened by the effective noise channels ${\cal N}_l^{\rm eff}$ with an error rate that is quadratic w.r.t. the original error rate (see, e.g., \cite{huang_performance_2019}). In the case of TEM with bond dimension $\chi^{\ast} = \frac{1}{2} L^2$, we also have residual effective noise channels ${\cal N}_l^{\rm eff}$ that are Pauli channels with relaxation rates $\sim (\lambda_{\boldsymbol{\alpha}}^{(l)})^2$ instead of $\lambda_{\boldsymbol{\alpha}}^{(l)}$.

If the bond dimension in TEM is below the threshold ($\chi < \chi^{\ast} \approx \frac{1}{2} L^2$), then some first-order noisy terms also contribute to the infidelity $\widetilde{f}_{l{\boldsymbol{\beta}}}^{-1} f_{l{\boldsymbol{\beta}}}$. We analyze the spectrum of singular values in bonds of the TEM map in Appendix~\ref{appendix-TEM-systematic} and conclude that they decay exponentially below the threshold bond dimension. Based on this analysis, we evaluate the associated the total compression-induced systematic error below the threshold bond dimension (see Appendix~\ref{appendix-TEM-systematic}):
\begin{equation*}
\delta_{\rm systematic}^{(3)} = \sqrt{\left( \frac{NL\varepsilon^2}{30} \right)^2 + \frac{ \varepsilon^2 L^2 }{ 72 \ln(4\sqrt{2N}) } \left[ N \Big( \frac{1}{32N} \Big)^{2\chi / L^2} - \frac{1}{32} \right]}.
\end{equation*}

\section{Prospects for practical quantum advantage} \label{section-prospects-advantage}

\subsection{Universality framework} \label{section-prospects-advantage-universality}

The use of computational resources in industry is rather utility-driven: computers automatize processes and provide algorithmic solutions to numerous problems from the same particular class. Automated algorithms are \textit{universal} in the sense that they are problem-independent within the class, and do not need customization to solve different instances with new parameters. This leads to a substantial cost reduction in the research and development efforts and speeds up the production cycle of new products such as materials and drugs. Conversely, some problem-specific customized solutions may outperform the problem-independent algorithmic ones, but any customization requires time and precious human resources that are generally unaffordable. Nor is customization scalable when the number of product users increases. For these reasons, we propose to quantify the efficiency of any computation paradigm (be it classical, quantum, or hybrid) with respect to how well it provides accurate automated solutions for a specific class of problems. The requirement of \textit{universality} significantly limits the range of algorithms and their flexibility; however, it is very much aligned with industrial needs and we therefore stress its importance. Under the constraint of universality for a given class of problems, the algorithm cannot be adapted to deal with specific instances, so quantum advantage in this framework refers to a superior performance of quantum-enabled universal algorithms as compared to purely classical universal algorithms.

In this paper, we consider the universality framework for a class of problems related to many-body quantum dynamics in discrete time. In this case, the computational task is as follows: ``Develop a universal computational algorithm (classical, quantum, or hybrid) adequately simulating discrete-time quantum dynamics of observables for a system with given interactions and product initial state''. The task is posed to develop a \textit{universal} algorithm, which would be a fixed code (function) that takes elementary $m$-body transformations (unitary matrices), the local initial state, and the observable (or a set of observables) as the input and returns an estimation of the desired observable (or a set of observables) at a given discrete time as the output. (The initial state can be omitted in the task formulation if its preparation is encoded into the dynamics; in this case, the trivial factorized state $\ket{0}^{\otimes N}$ is the default initial state.) The algorithm has to work in a problem-agnostic way and provide reasonable estimations of dynamical quantities for whatever dynamics is specified (either trivial or complex). The broad scope of this class covers many relevant real-world applications, for instance in chemistry and materials science.

To the best of our knowledge, there is no better alternative for a universal classical algorithm than resorting to conventional tensor network simulations. The affordable bond dimension is then limited by the compute power available and the wall time allocated to solve the problem. In the following, we consider two scenarios for the classical computational resources: modest ones (`c'), and those available with world-class supercomputers (`C'). As the problem complexity grows, classical simulations exhibits a transition from exact to approximate when the exponentially growing exact bond dimension gets out of reach. Here, we argue that for some specific instances of the aforementioned generic problem, which involve highly entangled states (quantum advantage candidates), approximate tensor network simulations result in very inaccurate estimations even with state-of-the-art supercomputers. As we will see in the next section, in this worst-case scenario for classical simulations, the relative estimation error exhibits a phase transition from 0 to 1 when the exact bond dimension starts exceeding the affordable bond dimension by a factor from 1 to 2. 

Quantum computation without any noise mitigation (`q') provides a quick but incorrect estimation of dynamical quantities. However, being equipped with a noise mitigation technique requiring modest classical computational resources, the hybrid quantum-classical simulator (`q+c') can generally provide decent estimations (with an acceptable accuracy of several percent), provided the sampling overhead is not prohibitively large. This means that the `q+c' simulation is potentially able to provide a preciser estimation of an observable than the `C' simulation within the same wall time, see Fig.~\ref{figure-advantage}. The purpose of the next section is to present a problem example that would be hard to simulate classically via universal (uncustomized) tensor network simulations (namely, conventional matrix product states) and to evaluate the problem size and the quantum computing parameters (error rate, noise stability) with which quantum advantage can be potentially demonstrated.

\begin{figure}
\includegraphics[width=0.45\textwidth]{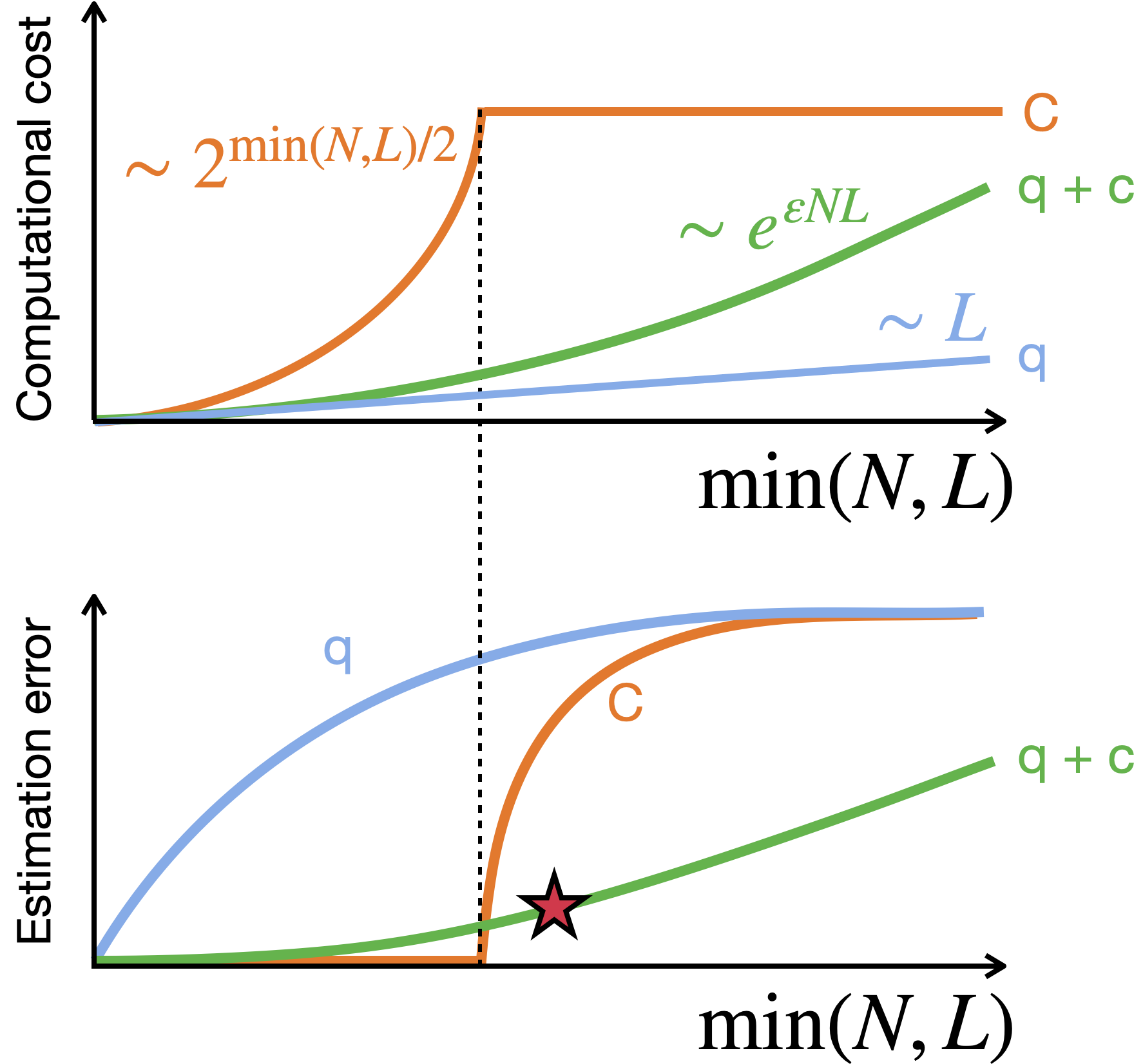}
\caption{\label{figure-advantage} Potential for quantum advantage in problems involving $N$ qubits and $L$ two-local unitary layers. The noisy quantum computation (q) is fast but inaccurate. The cost of exact classical simulation grows exponentially in the problem size and reaches the limit imposed by state-of-the-art compute power available (C), after which classical simulations become approximate for larger problem sizes. For specific problems involving highly entangled states (in the Schr\"{o}dinger picture) and operators (in the Heisenberg picture), the approximate classical solution rapidly becomes inadequate and results in a high estimation error. A tandem of noisy quantum computation and noise mitigation with modest classical resources (q+c) also involves an exponentially growing computation cost, with the exponent being proportional to the total number of errors in the causal area of the circuit. This exponential growth could be affordable for $N,L \sim 100$ if the average error rate $\varepsilon$ is sufficiently low. Within the same wall time, the resulting estimation error of the `q+c' computation can therefore be smaller than that for the `C' computation manifesting the useful quantum advantage (star).}
\end{figure}

\subsection{Verifiable model} \label{section-advantage-verifiable-model}

\begin{figure}
\begin{center}
\begin{quantikz}[row sep=0.1cm, column sep=0.1cm]
\gategroup[10,steps=5,style={dashed,rounded
corners,fill=green!20, inner
xsep=9pt,xshift=-7pt},background,label style={label
position=below,anchor=north,yshift=-0.2cm}]{$\ket{\psi_0}$} \lstick{$\ket{0}$} & \gate{R_Y(\frac{\pi}{2})} & \ctrl{1} & \gate{H} & \gate{R_Z(\varphi)}  &  & & \hphantomgate{\hspace{1pt}}  & \gate[label style={yshift=0.25cm}]{U} \gategroup[10,steps=2,style={dashed,rounded
corners,fill=blue!20, inner
xsep=2pt},background,label style={label
position=below,anchor=north,yshift=-0.2cm}]{\parbox{4em}{\centering Floquet operator}} & \gate[2]{U} & \hphantomgate{\hspace{1pt}} & \ \ldots \ & \hphantomgate{\hspace{1pt}} & & \gate[label style={yshift=0.25cm}]{U} \gategroup[10,steps=2,style={dashed,rounded
corners,fill=blue!20, inner
xsep=2pt},background,label style={label
position=below,anchor=north,yshift=-0.2cm}]{} & \gate[2]{U} & & \hphantomgate{\hspace{1pt}} & \meter{} \\
\lstick{$\ket{0}$} & & \targ{} & \gate{H} &  & & & & \gate[2]{U} & & & \ \ldots \ & & &  \gate[2]{U} & & & & \meter{} \\
\lstick[label style={yshift = 3pt, xshift = -5pt}]{$\vdots$} & \gate{R_Y(\frac{\pi}{2})} & \ctrl{1} & \gate{H} & \gate{R_Z(\varphi)}  &  & & & & \gate[2]{U} & & \ \ldots \ & & & & \gate[2]{U} & & & \meter{} \\
\lstick[label style={yshift = 3pt, xshift = -5pt}]{$\vdots$} & & \targ{} & \gate{H} &  & & & & \gate[2]{U} & & & \ \ldots \ & & &  \gate[2]{U} & & & & \meter{} \\
\lstick{$\ket{0}$} & \gate[style = {fill=yellow!20}]{R_Y(\theta)} & \ctrl{1} & \gate{H} & \gate{R_Z(\varphi)}  &  & & & & \gate[2]{U} & & \ \ldots \ & & & & \gate[2]{U} & & & \meter{} \\
\lstick{$\ket{0}$} & & \targ{} & \gate{H} &  & & & & \gate[2]{U} & & & \ \ldots \ & & &  \gate[2]{U} & & & & \meter{} \\
\lstick[label style={yshift = 3pt, xshift = -5pt}]{$\vdots$} & \gate{R_Y(\frac{\pi}{2})} & \ctrl{1} & \gate{H} & \gate{R_Z(\varphi)}  &  & & & & \gate[2]{U} & & \ \ldots \ & & & & \gate[2]{U} & & & \meter{} \\
\lstick[label style={yshift = 3pt, xshift = -5pt}]{$\vdots$} & & \targ{} & \gate{H} &  & & & & \gate[2]{U} & & & \ \ldots \ &  & &  \gate[2]{U} & & & & \meter{} \\
\lstick{$\ket{0}$} & \gate{R_Y(\frac{\pi}{2})} & \ctrl{1} & \gate{H} & \gate{R_Z(\varphi)}  &  & & & & \gate[2]{U} & & \ \ldots \ &  & & & \gate[2]{U} & & & \meter{} \\
\lstick{$\ket{0}$} & & \targ{} & \gate{H} & & & & & \gate[label style={yshift=-0.25cm}]{U} & & & \ \ldots \ & & & \gate[label style={yshift=-0.25cm}]{U} & & & & \meter{}
\end{quantikz}
\end{center}
\caption{\label{figure-circuit-2} Quantum circuit for a discrete-time dynamics of initially correlated qubit pairs under repeated local interactions, e.g., a kicked isotropic Heisenberg model with external $X$- and $Z$-fields, Eq.~\eqref{kH-model-2}. Translational symmetry is initially broken in the central qubit pair, which controls those degrees of freedom that are truncated in large-scale classical simulations with the help of uncustomized matrix product states.}
\end{figure}
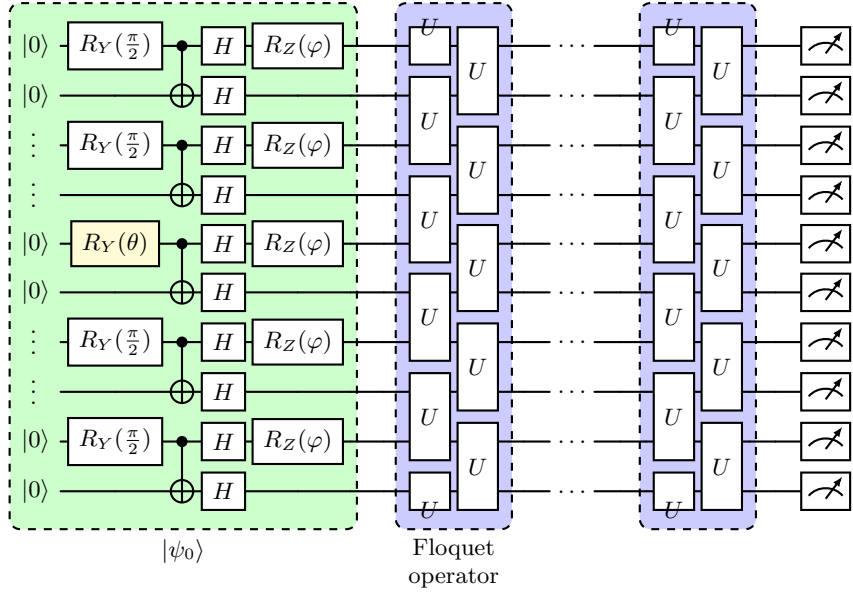

As a specific sub-class of problems, we study the simulation of discrete-time many-body dynamics~\cite{kim_evidence_2023,van_den_berg_probabilistic_2023,fauseweh_quantum_2024,salathe_digital_2015,keenan_evidence_2023,zaletel-2023,bertini_localized_2024}. Consider a one-dimensional Floquet circuit with brickwork structure of the elementary unitary transformation $U_{\langle q_1, q_2 \rangle}$ affecting neighboring qubits $\langle q_1, q_2 \rangle$~\cite{chan_solution_2018}. The total unitary transformation with periodic boundary conditions for $N = 4k+2$ qubits, $k \in \mathbb{N}$, is $W(t) \equiv \overleftarrow{\bigcirc}_{\tau=1}^{t} W_\tau$, where $W_\tau = U_{\langle 0,1\rangle} \otimes U_{\langle 2,3\rangle} \otimes \cdots \otimes U_{\langle N-2,N-1\rangle}$ if $\tau$ is even and $W_\tau = U_{\langle 1,2\rangle} \otimes U_{\langle 3,4 \rangle} \otimes \cdots \otimes U_{\langle N-1,0\rangle}$ if $\tau$ is odd. We focus on a particular non-Clifford transformation 
\begin{equation} \label{kH-model-2}
\begin{quantikz}[row sep=0.1cm]
& \gate[2]{U} & \\
& &
\end{quantikz} = \begin{quantikz}[row sep=0.1cm]
& \gate{\sqrt{X}} & \swap[partial swap={~J~}]{1} & \gate{\sqrt{X}} & \\
& \gate{T} &  \targX{} & \gate{T} &
\end{quantikz} = (\sqrt{X} \otimes T) \ \exp[-iJ(X \otimes X + Y \otimes Y + Z \otimes Z)] \ (\sqrt{X} \otimes T) 
\end{equation}
constituting the kicked isotropic Heisenberg model with external $X$- and $Z$-fields~\cite{garratt_local_2021}. Here $\sqrt{X} = e^{i\pi/4} R_X(\frac{\pi}{2})$ is a Clifford gate and $T = e^{i\pi/8} R_Z(\frac{\pi}{4})$ is a non-Clifford gate. $T$-gates do not cancel in $W(t)$ if $J \in (0,\frac{\pi}{2})$ meaning the whole transformation is not Clifford either. 

Suppose the initial state is a $2$-local correlated state with broken translational symmetry,
\begin{equation}
\ket{\psi_0} = \left[ R_Z(\varphi) \otimes I \big(\cos\tfrac{\theta}{2} \ket{+ +} + \sin\tfrac{\theta}{2} \ket{- -} \big) \right]_{\frac{N}{2} - 1, \frac{N}{2}} \otimes \bigotimes_{n \neq (N-2)/4} \tfrac{1}{\sqrt{2}} \left( \ket{0_{2n} 0_{2n+1}} + e^{i \varphi} \ket{1_{2n} 1_{2n+1}} \right),
\end{equation}
where $\ket{\pm} = \frac{1}{\sqrt{2}} (\ket{0} \pm \ket{1})$. Such a state is created with the help of a single layer of {\sc cnot} gates, see Fig.~\ref{figure-circuit-2}. The goal of the computation is to estimate the dynamical observable $\bra{\psi_t} O \ket{\psi_t}$ with $\ket{\psi_t} = W(t) \ket{\psi_0}$ for some values of the interaction-strength parameter $J$, the local phase $\varphi$, and the symmetry-breaking parameter $\theta$. 

Classical MPS simulation remains exact up to time step $t$ if it uses a bond dimension $\chi_{\rm exact} = 2^{\min(t+1,N/2)}$, because each Floquet layer $W_\tau W_{\tau+1}$ corresponds to a multiplicative factor $4$ in the exact bond dimension, $\ket{\psi_0}$ is an MPS with bond dimension $2$, and the MPS bond dimension for $N$ qubits can always be exactly compressed down to $2^{N/2}$. By time step $t = \lfloor N/4 \rfloor$, the causal cone for every initial qubit pair expands to half of the spin chain, so diametrically opposite qubits get potentially entangled. The computational cost of constructing the exact MPS for $\ket{\psi_{ \lfloor N/4 \rfloor}}$ is $32 N \sum_{\tau=1}^{\lfloor N/4 \rfloor} 2^{2\tau}$ floating point operations and grows exponentially in $N$, but the main difficulty is the final contraction with the observable, which requires at least $2N \chi_{\rm exact}^3 = 2 N 2^{3\lfloor N/4 \rfloor+3}$ floating point operations~\cite{schollwock-2011}. (The same arguments about the final contraction complexity are valid if one propagates the state and the observable until the middle layer of the circuit in the Schr\"{o}dinger picture and in the Heisenberg picture, respectively~\cite{begusic_fast_2023}.) For $N=122$ ($t = 30$) such a simulation would require at least $63\,854$ years of distributed computing wall time with the most powerful supercomputer currently available, \texttt{Frontier}, with $\sim 1.2 \times 10^{18}$ FLOPS. 

\begin{figure}
\includegraphics[width=0.7\textwidth]{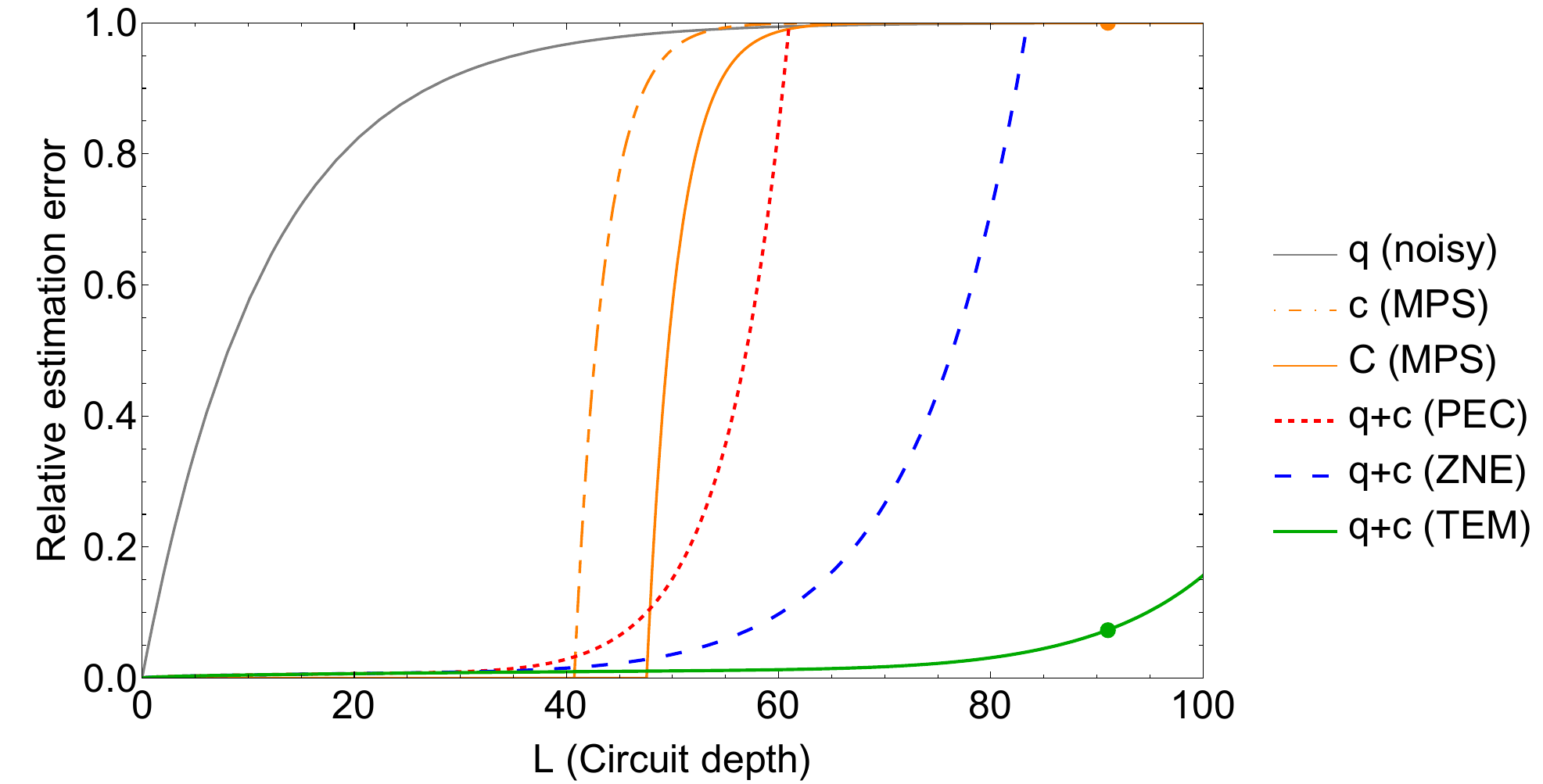}
\caption{\label{figure-c+q-2} Accuracy of observable estimation in quantum computation without noise mitigation (q), classical simulation with modest supercomputer power (c) and world-class supercomputer power (C), and hybrid approaches for quantum computation with noise mitigation (q+c). In (q), the error rate per qubit per layer $\varepsilon = 0.0014$. In (c), wall time $24$~hours and $10^{15}$~FLOPS. In (C), wall time $24$~hours and $1.2 \times 10^{18}$~FLOPS. In (q+c), noise stability is quantified by the standard deviation in the relative error per layer per qubit $\Theta = 1.8 \%$. Points correspond to the results for the parity observable estimation in the kicked isotropic Heisenberg dynamics with external inhomogeneous $X$- and $Z$-fields and a 2-local correlated initial state; circuit parameters: Heisenberg interaction strength $J = \frac{\pi}{4}$, symmetry-breaking and phase-shift angles $\theta = 1.5$ and $\varphi = 2.63$, number of qubits $N=122$, number of layers $L=91$ (evolution step $t=30$).}
\end{figure}

To get an approximate solution to a problem within a shorter time period (say, ${\rm T}=24$~hours), MPS compression is needed. While the compression affects the degrees of freedom that contribute the least to the state, it may nevertheless lead to a completely wrong result even if the bond dimension is as large as, e.g., the half of the exact one, $\chi = \frac{1}{2} \chi_{\rm exact}$. Let us illustrate this in the verifiable regime with $J=\frac{\pi}{4}$, in which the interaction term $\exp[-iJ(X \otimes X + Y \otimes Y + Z \otimes Z)]$ reduces to the {\sc swap} operation and the circuit becomes dual unitary~\cite{piroli_2020,goold_2024}. In this case,
\begin{eqnarray} \label{Bell-pair-rainbow-2}
\ket{\psi_{t}} & = & \left[ (\sqrt{X} T)^t R_Z(\varphi) \otimes (T \sqrt{X})^t \big(\cos\tfrac{\theta}{2} \ket{+ +} + \sin\tfrac{\theta}{2} \ket{- -} \big) \right]_{\frac{N}{2}-1-t, \frac{N}{2}+t} \nonumber\\
&& \otimes \bigotimes_{n \neq (N-2)/4} \Big[ (\sqrt{X} T)^t \otimes (T \sqrt{X})^t \tfrac{1}{\sqrt{2}} \left( \ket{0 0} + e^{i \varphi} \ket{1 1} \right) \Big]_{2n-t,2n+1+t}.
\end{eqnarray}
The state~\eqref{Bell-pair-rainbow-2} is composed of $\frac{N}{2}-1$ maximally entangled qubit pairs ($n=0,1,\ldots,(N-4)/2$), and one ``skewed'' Bell-like pair for qubits $\frac{N}{2}-1-t$ and $ \frac{N}{2}+t$ that is equivalent to the state $\cos\tfrac{\theta}{2} \ket{0 0} + \sin\tfrac{\theta}{2} \ket{1 1}$ up to local unitary transformations. If $t \leq \lfloor N/4 \rfloor$, the $\chi_{\rm exact}$ Schmidt coefficients of the entanglement across the link between the middle qubits are divided into two equal halves: those proportional to $|\cos\tfrac{\theta}{2}|$ and those proportional to $|\sin\tfrac{\theta}{2}|$. If $|\cos\tfrac{\theta}{2}| > |\sin\tfrac{\theta}{2}|$, then the bond dimension compression ${\cal C}_{\chi}$ with $\chi = \frac{1}{2} \chi_{\rm exact}$ truncates the singular values $\propto |\sin\tfrac{\theta}{2}|$ ($\rightarrow 0$) and results in the renormalized state
\begin{eqnarray} \label{Bell-pair-rainbow-broken-2}
\ket{\psi_{t}^{(\chi)}} \equiv {\cal C}_{\chi} \big( \ket{\psi_{t}} \big) & = & (\sqrt{X} T)^t R_Z(\varphi) \ket{+}_{\frac{N}{2}-1-t} \otimes (T \sqrt{X})^t \ket{+}_{\frac{N}{2}+t} \nonumber\\
&& \otimes \bigotimes_{n \neq (N-2)/4} \Big[ (\sqrt{X} T)^t \otimes (T \sqrt{X})^t \tfrac{1}{\sqrt{2}} \left( \ket{0 0} + e^{i \varphi} \ket{1 1} \right) \Big]_{2n-t,2n+1+t}.
\end{eqnarray}
The difference between the states \eqref{Bell-pair-rainbow-2} and \eqref{Bell-pair-rainbow-broken-2} manifests itself in the average value of any observable involving edge qubits in the spin chain. For example, the case of $N=122$ qubits, $t=30$ steps, parameters $\theta = 1.5$ and $\varphi = 2.63$ leads to a drastic discrepancy between the exact and approximate average values for the parity observable $O=Z^{\otimes 122}$:
\begin{eqnarray*}
\bra{\psi_{30}} Z^{\otimes 122} \ket{\psi_{30}} & \approx & 0.997, \\
\bra{\psi_{30}^{(2^{30})}} Z^{\otimes 122} \ket{\psi_{30}^{(2^{30})}} & \approx & 0.016.
\end{eqnarray*}

A similar discrepancy between the exact value and its classical approximation is observed for the 2-local observable $O = Z_{\frac{N}{2}-1-t} Z_{\frac{N}{2}+t}$ and the 2-body correlator $\braket{Z_{\frac{N}{2}-1-t} Z_{\frac{N}{2}+t}} - \braket{Z_{\frac{N}{2}-1-t}} \braket{Z_{\frac{N}{2}+t}}$. We therefore argue that the relative estimation error in universal classical simulations may exhibit a phase transition from $0$ to $1$ (similarly to what is observed in Ref.~\cite{jaschke_is_2023} for systems with the Marchenko-Pastur distribution for the singular values) even if one had an enormous compute power to deal with the exponentially large bond dimension $\chi = \frac{1}{2} \chi_{\rm exact}$ ($7\,981$ years with the supercomputer \texttt{Frontier}, with $\sim 1.2 \times 10^{18}$ FLOPS, to estimate the observable for an MPS with $N=122$ and $\chi = 2^{30}$). In our example, if $\frac{1}{2} \chi_{\rm exact} < \chi < \chi_{\rm exact}$, then $\chi / \chi_{\rm exact}$ is the fraction of the relevant Schmidt coefficients contributing to the correct observable estimation, so the relative estimation error of the conventional MPS simulation is $\sim 1 - \chi / \chi_{\rm exact}$. We illustrate the phase transition in the classical-simulation error in Fig.~\ref{figure-c+q-2} and compare it with noisy quantum computation without error mitigation (estimation error $\sim 1-e^{-\varepsilon N L /2}$), which performs worse. We take into account that in a quantum computation a general two-qubit unitary transformation involves $3$ {\sc cnot} gates~\cite{vatan-2004}, for instance,
\begin{equation*}
\begin{quantikz}[row sep=0.1cm]
& \swap[partial swap={~J~}]{1} & \midstick[2,brackets=none]{=} & \gate{S} & \ctrl{1} & \gate{R_Y(\frac{\pi}{2}+2J)} & \gate{H} & \ctrl{1} & \gate{H} & \gate{R_Y(-\frac{\pi}{2}-2J)} & \ctrl{1} & & \\
&  \targX{} &  & & \targ{} & \gate{R_Z(-\frac{\pi}{2}-2J)} & \gate{H} & \targ{} & \gate{H} & & \targ{} & \gate{S^{\dag}} &
\end{quantikz}
\end{equation*}
and that the circuit depth $L = 3t+1$ in the kicked-Heisenberg dynamics with a 2-locally correlated initial state. 

We expect a significant difference between $\bra{\psi_t} O \ket{\psi_t}$ and $\bra{\psi_t^{(\chi)}} O \ket{\psi_t^{(\chi)}}$ to take place also \textit{beyond} the exactly verifiable regime (i.e., $J \neq \frac{\pi}{4}$) provided the average value $\bra{\psi_t} O \ket{\psi_t}$ is essentially nonzero. The latter condition is not impossible to satisfy for the considered kicked Heisenberg model with external $X$- and $Z$- fields and the parity observable $O=Z^{\otimes N}$ because a single unitary operator $\exp[-iJ(X \otimes X + Y \otimes Y + Z \otimes Z)]$ with $J \neq \frac{\pi}{4}$ is a linear combination of the identity transformation and a {\sc swap}, resulting in a random braiding of qubit trajectories, to which the qubit-uniform observable $O=Z^{\otimes N}$ is insensitive in the absence of external fields or in the case of homogeneous external fields. These are the inhomogeneous local fields that render the quantity $\bra{\psi_t} O \ket{\psi_t}$ unknown beyond the verifiable regime. Since the phases from $\sqrt{X}$- and $T$- gates accumulate and can cancel (as in $J=0$ and a number of steps divisible by $4$), the signal may remain noticeable even in the case of random qubit braiding. Most likely, some customized classical methods can provide a solution to this specific problem; however, another quantum dynamical problem (e.g., with anisotropic Heisenberg interaction or an interaction of another kind) would require yet another \textit{ad hoc} classical approach, whereas quantum computation with noise mitigation provides a universal toolbox for all problems within the class. 

\section{Conclusions}

By analyzing the major contributions to the final estimation error in noise-mitigated quantum computations with three noise-aware techniques in use (PEC, ZNE with probabilistic error amplification, and TEM), we have filled a substantial gap in forecasting noise mitigation performance at large scale. The minimal random error and the optimal sampling overhead in ZNE have remained essentially unknown prior to this work. Similarly, we have revealed the systematic errors due to the noise instability or imprecision in the learned noise parameters that are inherent to all noise-aware mitigation techniques. We have also evaluated technique-specific systematic errors such as the extrapolation error in ZNE and the compression error in TEM. We have proved the optimality of TEM's sampling cost by comparing it against the universal lower bound~\cite{tsubouchi_universal_2023} and argued that TEM with the threshold bond dimension has a similar impact as quantum error correction with code distance $3$, which effectively shifts the average error rate to a higher order, $\varepsilon \rightarrow \varepsilon^2$. The latter fact establishes a connection between error mitigation, relying on additional measurements as a resource, and error correction, relying on additional qubits, thus opening a possibility for a trade-off between these two resources as quantum computing keeps improving. The partial interchangeability of quantum error mitigation and quantum error correction has the potential to accelerate the deployment of large-scale quantum algorithms with practical use. For example, since early error-corrected quantum computations will unavoidably have residual errors due to the inaccessibility of recursive encodings~\cite{bluvstein_logical_2024}, the combination with optimal error mitigation of said residual errors could significantly extend the simulation scale. Also, cross-fertilized noise mitigation could potentially combine advantages of different noise-mitigation techniques in quantum processors with sophisticated topology and complex noise model. 

The derived analytical expressions for the total error $\delta$ in noise-mitigated estimations should be considered as a prediction for the \textit{typical error} one could expect in dense quantum circuits of high complexity; however, the specific errors for different observables will generally be scattered around $\delta$. Nonetheless, our analysis sets clear requirements on quantum computational resources and their quality in order to demonstrate a practical quantum advantage. The advantage can be seen as a faster and preciser estimation of observables in large scale brickwork circuits with the help of a \textit{universal} algorithm based on noisy quantum computing complemented by an error mitigation technique, rather than with the help of a \textit{universal uncustomized} classical simulation algorithm. The stress on universality is dictated by industrial needs, where problem-specific customized solutions are costly. With foreseeable large-scale quantum processors, the overall error in noise-mitigated estimations can be reasonably small ($\delta \lesssim 10\%$) in $100 \times 100$ circuits, where all-purpose classical simulations struggle. Elaborating on a paradigmatic example of the kicked Heisenberg model with external inhomogeneous fields, where the entanglement grows exponentially in the circuit depth, we clarify the failure mechanism of matrix product states, the relative estimation error of which approaches 100\% even if the compressed bond dimension is as large as half of the exact (exponentially scaling) one. This example serves as a verifiable point in the class of the discrete-time Floquet dynamical problems, within which quantum advantage can be potentially demonstrated.

\begin{acknowledgements}
The authors greatly appreciate useful discussions on quantum error mitigation and its applications with Alberto Baiardi, Ram\'{o}n Panad\'{e}s Barrueta, Bruno Bertini, Elsi-Mari Borrelli, Gonzalo Camacho, Marco Cattaneo, Daniel Cavalcanti, Roberto Di Remigio Eik\r{a}s, Shane Dooley, Andrew Eddins, Daniel Egger, Jens Eisert, Satoshi Ejima, Benedikt Fauseweh, Laurin Fischer, Adam Glos, John Goold, Erik Gustafson, Roni Harnik, Susana Huelga, Abhinav Kandala, Nathan Keenan, Stefan Knecht, Vijay Krishna, Doga Kurkcuoglu, Elica Kyoseva, Hank Lamm, Matea Leahy, Kevin Lively, Joonas Malmi, Stefano Mangini, Antonio Mezzacapo, Thomas O'Brien, Fabijan Pavosevic, Gabe Perdue, Matteo Rossi, Max Rossmannek, Boris Sokolov, Francesco Tacchino, Walter Talarico, Ivano Tavernelli, and Zolt\'{a}n Zimbor\'{a}s. The authors thank Daniel Cavalcanti and John Goold for bringing references \cite{takagi_universal_2023,tsubouchi_universal_2023} to their attention.
\end{acknowledgements}

\section*{Competing interests}
Elements of this work are included in patents filed by Algorithmiq Ltd.

\section*{Author constributions}
SNF derived the majority of the results, with some contributions from GGP. SNF and GGP designed and directed the research. SNF ran numerical simulations, prepared the figures, and wrote the first version of the manuscript. All authors contributed to scientific discussions and to the writing of the manuscript.
\appendix

\section{Resource estimates for PEC} \label{appendix-PEC-summary}

Given a finite allocated time ${\rm T}$ for quantum computation, the available number of circuit executions is $M = {\rm T}/(L \tau_{\rm l} + \tau_{\rm m} + \tau_{\rm delay})$, where $L$ is the circuit depth, $\tau_{\rm l}$ is the duration of single-qubit and two-qubit gates in one layer, $\tau_{\rm m}$ is the measurement duration, and $\tau_{\rm delay}$ is the circuit delay between sequential circuit executions ($\tau_{\rm l} \sim 0.6$ $\mu$s, $\tau_{\rm m} \sim 0.8$ $\mu$s, and foreseeable $\tau_{\rm delay} = 0.5$ ms in experiments with superconducting qubits~\cite{kim_evidence_2023}). The classical cost of calculating the mitigated value $\bar{\bar{O}}$ and its standard deviation $\Delta\bar{\bar{O}}$ is negligible if the acquired data are processed on the fly. The total estimation error $\delta$ for a non-local observable $O$ in a dense $N \times L$ circuit is then evaluated through
\begin{equation}
\delta^2 = \left( \frac{\exp (\varepsilon N L)}{\sqrt{{\rm T}/(L \tau_{\rm l} + \tau_{\rm m} + \tau_{\rm delay})} } \right)^2 + \left( \dfrac{1}{2} \varepsilon N \sqrt{L} \Theta \right)^2,
\end{equation}
where $\varepsilon$ is the average error rate in the circuit and $\Theta$ is the standard deviation in the relative error density fluctuations. The time needed for the noise characterization is excluded from this resource analysis, as the number of unique layers can vary significantly from problem to problem (for example, it equals 2 for the Floquet dynamics considered in Sec.~\ref{section-prospects-advantage}). Should the noise be recharacterized during the experiment, this would essentially incur no extra cost in designing new circuits and processing new measurement outcomes because the sampling probabilities and the rescaling factor $\gamma$ are trivially expressed through renewed noise parameters.

\section{Lower bound for the random error in ZNE} \label{appendix-ZNE-random}

Given a collection of positive parameters $\{a_j\}_{j}$, $a_j > 0$, the sum $\sum_j \frac{a_j}{|{\sf S}_j|}$ attains the minimum value under the constraint $\sum_i |{\sf S}_i| = M$ if $|{\sf S}_j|^2 = c a_j$ for all $j$ and some coefficient $c$. This follows from the partial derivatives $\frac{\partial {\cal L}}{\partial |{\sf S}_j|} = 0$ of the Lagrangian function ${\cal L} = \sum_i \frac{a_i}{|{\sf S}_i|} - \lambda \big( \sum_i |{\sf S}_i| - M \big)$. In the case of expression \eqref{ZNE-random-error}, the optimal distribution of shots is $|{\sf S}_j^{\ast}| \propto \left\vert \sum_i G_i (G_i - G_j) \right\vert K^{-G_j}$. 

Numerical investigations suggest that the optimization with several extrapolation points ($R=3,4$) gives the same minimum in~\eqref{ZNE-random-error-min} as that for $R=2$. If $R=2$, then we can readily optimize~\eqref{ZNE-random-error-min} over parameter $G_1$ by minimizing the numerator and maximizing the denominator:
\begin{equation*}
\min_{G_2 > G_1 \geq 1} \frac{G_2 (\frac{1}{K})^{G_1} + G_1 (\frac{1}{K})^{G_2}}{G_2 - G_1} = \min_{G_2 > 1} \frac{G_2 \frac{1}{K} + (\frac{1}{K})^{G_2}}{G_2 - 1},
\end{equation*}

\noindent that is $G_1^{\ast} = 1$. Since the derivative
\begin{equation*}
\frac{d}{d G_2} \left( \frac{G_2 \frac{1}{K} + (\frac{1}{K})^{G_2}}{G_2 - 1} \right) = \frac{(G_2-1)(\frac{1}{K})^{G_2} \ln(\frac{1}{K}) - (\frac{1}{K})^{G_2} - \frac{1}{K} }{(G_2-1)^2}
\end{equation*}
vanishes if $K^{G_2-1} \exp(K^{G_2-1}) = \frac{1}{e}$, the optimal parameter $G_2^{\ast}$ is such that $w \equiv K^{G_2^{\ast}-1}$ is a solution of equation $we^w = \frac{1}{e}$, i.e., $w=W(\frac{1}{e}) \approx 0.278$, where $W(z)$ is the principal branch of the Lambert W function. Finally, $G_2^{\ast} = [1+W(1/e)+\ln (1/K)] / {\ln (1/K)} \approx 1 + \frac{2.557}{\varepsilon N L}$. Substituting $G_1^{\ast}$ and $G_2^{\ast}$ in~\eqref{ZNE-random-error-min}, we get the lower bound~\eqref{ZNE-minimal-random-error} for ZNE's random error.

\section{Lower bound the systematic error in ZNE} \label{appendix-ZNE-systematic}

To derive the lower bound~\eqref{ZNE-systematic-error} for the ZNE systematic error, we simplify the fraction by negating the numerator and exploiting the following sequence of transformations:
\begin{eqnarray*}
    \Big( \sum_i G_i \Big) \Big(\sum_i G_i^3 \Big) - \Big(\sum_i G_i^2\Big)^2 &=& \sum_{ij} (G_i G_j^3 - G_i^2 G_j^2) \\
    &=& \sum_{ij} G_i G_j^2 (G_j - G_i) \\
    &=& \sum_{i<j} (G_i G_j^2 - G_i^2 G_j) (G_j - G_i) \\
    &=& \sum_{i<j} G_i G_j (G_j - G_i)^2 \\
    &\geq & \Big( \min_{G_i \neq G_j} G_i G_j \Big) \sum_{i<j} (G_j - G_i)^2 \\
    &=& G_1 G_2 \sum_{i<j} (G_j - G_i)^2 \\
    &=& G_1 G_2 \, \frac{1}{2} \sum_{ij} (G_j - G_i)^2 \\
    &=& G_1 G_2 \Big[ R \sum_{i} G_i^2 - \Big(\sum_i G_i\Big)^2 \Big].
\end{eqnarray*}

\noindent This means the fraction in the first line of Eq.~\eqref{ZNE-systematic-error} is always negative and its absolute value always exceeds $G_1 G_2$.

\section{Resource estimates for ZNE} \label{appendix-ZNE-summary}

Given a finite allocated time ${\rm T}$ for quantum computation, the available number of circuit executions is $M = {\rm T}/(L \tau_{\rm l} + \tau_{\rm m} + \tau_{\rm delay})$, where $L$ is the circuit depth, $\tau_{\rm l}$ is the duration of single-qubit and two-qubit gates in one layer, $\tau_{\rm m}$ is the measurement duration, and $\tau_{\rm delay}$ is the circuit delay between sequential circuit executions ($\tau_{\rm l} \sim 0.6$ $\mu$s, $\tau_{\rm m} \sim 0.8$ $\mu$s, and foreseeable $\tau_{\rm delay} = 0.5$ ms in experiments with superconducting qubits~\cite{kim_evidence_2023}). To minimize the random error, the total budget of $M$ circuit executions should be distributed in an optimal way among experiments with specific noise gains $\{G_i^{\ast}\}_i$. The classical cost of extrapolation and evaluation of errors via bootstrapping is negligible. The total estimation error $\delta$ for a non-local observable $O$ in a dense $N \times L$ circuit is then evaluated through
\begin{equation}
\delta^2 = \left( \frac{(1+1.795 \varepsilon N L) \exp (\varepsilon N L / 2)}{\sqrt{{\rm T}/(L \tau_{\rm l} + \tau_{\rm m} + \tau_{\rm delay})} } \right)^2 + \left( \dfrac{1}{2} \varepsilon N \sqrt{L} \Theta \right)^2 + \left[ \left( 1 + \dfrac{2.557}{\varepsilon N L} \right) \dfrac{NL\varepsilon^2}{36} \right]^2,
\end{equation}
where $\varepsilon$ is the average error rate in the circuit and $\Theta$ is the standard deviation in the relative error density fluctuations. The time needed for the noise characterization is excluded from this resource analysis, as the number of unique layers can vary significantly from problem to problem (for example, it equals 2 for the Floquet dynamics considered in Sec.~\ref{section-prospects-advantage}). Noise recharacterization during a single ZNE session with several experiments for different noise gains $\{G_i\}_i$ allows at least two ways of data acquisition and processing: (i) collecting smaller batches of measurement shots for different noise gains $\{G_i\}_i$ in between the noise characterizations, extrapolating these batches independently, then averaging the extrapolated values; (ii) averaging noisy observable estimations for each noise gain over different noise characterizations, and then extrapolating them to get the final observable estimation. Both approaches are imperfect: the former one implies much bigger errors in separate exponential extrapolations due to a smaller number of shots per extrapolation; the latter one effectively performs extrapolation for a sum of exponents, not a single exponent, which was shown to lead to a systematic error. To simplify the analysis, we assume the noise instability is minor so that characterizing noise only once suffices.

\section{Compression-induced systematic error in TEM} \label{appendix-TEM-compression}

Consider a quantum computation that starts with the preparation of the initial state $\ket{0}^{\otimes N}$ for $N$ qubits. This state is stabilized by $2^N$ Pauli-string operators $\{G_i\}_{i=0}^{2^N-1}$ acting nontrivially as $Z$ Pauli operator on some, all, or none of the qubits (symbols `1' in the binary representation of $i$ define the qubit positions where $G_i$ acts nontrivially, e.g., $G_{1} = I^{\otimes (N-1)} \otimes Z$ and $G_{2^N-1} = Z^{\otimes N}$). A noiseless quantum circuit implements a unitary transformation $U$ that maps $\ket{0}^{\otimes N}$ to $\ket{\psi} = U\ket{0}^{\otimes N}$. The output state is stabilized by operators $O_i = U G_i U^{\dag}$ meaning $O_i \ket{\psi} = \ket{\psi}$. Circuits composed of Clifford gates (such as {\sc h} and {\sc cnot}) are classically simulable because the stabilizers $O_i$ are just Pauli strings in this case and the state $\ket{\psi}$ is uniquely defined by $N$ generators of the stabilizer group $\{O_i\}_{i=0}^{2^N-1}$. At the end of the noiseless computation, $\braket{O_i} = 1$ for any stabilizer observable. One can also consider $L$ layers of a long Clifford circuit and track the evolution of the depth-dependent stabilizer operator $O_i(L) = U(L) G_i U^{\dag}(L)$ satisfying $\braket{O_i(L)} = 1$. We use this exact expectation value as a benchmark to estimate systematic errors in TEM.

Suppose each unitary-layer transformation ${\cal U}_l[\bullet] = U_l \bullet U_l^{\dag}$ is preceded by an SPL noise layer ${\cal N}_l$, where ${\cal N}_l = \bigcirc_{q} {\cal N}_{lq}$ is a concatenation of one- and two-qubit Pauli channels ${\cal N}_{lq}$ acting nontrivially either at a qubit $q$ or a pair of neighboring qubits $q=\langle q_1,q_2\rangle$. In the expanded form, ${\cal N}_l [\bullet] = \sum_{\boldsymbol{\alpha}} p_{\boldsymbol{\alpha}}^{(l)} \sigma_{\boldsymbol{\alpha}} \bullet \sigma_{\boldsymbol{\alpha}}$, where $\sigma_{\boldsymbol{\alpha}} = \bigotimes_{m=0}^{N-1} \sigma_{\alpha_m}$ and $\boldsymbol{\alpha} = (\alpha_0,\ldots,\alpha_{N-1})$. Importantly, any Pauli string $\sigma_{\boldsymbol{\beta}}$ is an eigenoperator of the SPL noise, ${\cal N}_l[\sigma_{\boldsymbol{\beta}}] = f_{l{\boldsymbol{\beta}}} \sigma_{\boldsymbol{\beta}}$, and the corresponding eigenvalue (fidelity) $f_{l{\boldsymbol{\beta}}}$ is readily obtained from the parameters of the constituent noisy maps $\{{\cal N}_{lq}\}_q$, Eq.~\eqref{Pauli-fidelity} in the main text. In a Clifford circuit with the SPL noise, $L$ layers of noisy evolution transform the initial stabilizer $G_i$ into $c_i(L) O_i(L)$, where $c_i(L) = \prod_{l=1}^L f_{l\boldsymbol{\beta}(i,l)}$ and $\boldsymbol{\beta}(i,l)$ is the index of Pauli string $O_i(l) \equiv \sigma_{\boldsymbol{\beta}(i,l)}$. The noise results in the average value $\braket{O_i(L)}_{\rm noisy} = c_i(L)$ that exponentially decays in the circuit depth and makes it harder for the noise mitigation algorithm to recover the noiseless value $\braket{O_i(L)} = 1$.

TEM provides a collection of noise mitigation maps $\{{\cal M}_l^{\dag (\chi)}\}_{l=1}^L$ for different depths via the iterative procedure ${\cal M}_l^{\dag (\chi)} = {\cal C}_{\chi}\left( {\cal U}_l \circ ({\cal N}_l^{-1})^{\dag} \circ {\cal M}_{l-1}^{\dag (\chi)} \circ {\cal U}_l^{\dag} \right)$ involving the compression ${\cal C}_{\chi}$ down back to $\chi$. For a Clifford circuit of depth $L$, the noise-mitigated estimation of the stabilizer observable is $\braket{O_i(L)}_{\rm n.m.}^{(\chi)} = c_i(L) d_i(\chi,L)$, where $d_i(\chi,L)$ is the diagonal element of the MPO representation for ${\cal M}_L^{(\chi)}$ obtained by setting the input and output indices to $\boldsymbol{\beta}(i,L)$ [such that $O_i(l) \equiv \sigma_{\boldsymbol{\beta}(i,l)}$]. If no compression is utilized, then the calculations are exact and $\braket{O_i(L)}_{\rm n.m.}^{\infty} = 1$. The compression leads to the systematic error $\Delta O_i(\chi,L) = |c_i(L) d_i(\chi,L) - 1|$, which we now analyze. 

In a numerical experiment, to mimic the shape of the experimentally observed distribution~\cite{van_den_berg_probabilistic_2023,kim_evidence_2023} of noise parameters $\lambda_{\boldsymbol{\alpha}}^{(l)} \in {\rm SPL}({\cal N}_l)$, we choose a target error rate $\varepsilon$ and sample one-qubit (two-qubit) relaxation rates from the normal distribution where both the mean and the standard deviation are equal to $\varepsilon/12$ ($\varepsilon/36$). Negative rates are nullified for the maps $\{{\cal N}_l\}$ to be completely positive (in the same fashion as it is done in SPL noise learning~\cite{van_den_berg_probabilistic_2023}). As a result of such sampling, $\gamma_l \approx (1 + \varepsilon)^N$. 

In Fig.~\ref{figure-bd-threshold}, we depict the median $Q_{0.5} \left[ {\Delta O}(\chi,L) \right] \equiv {\rm median}\left[ \{ \Delta O_i(\chi,L) \}_{i=0}^{2^N-1} \right]$ of systematic errors for $10^3$ randomly chosen stabilizer observables in a Clifford circuit with brickwork {\sc cnot} topology and random single-qubit gates, $N=30$ qubits, $L=30$ layers, and error rates per qubit per layer $\varepsilon = 0.022$ and $\varepsilon = 0.002$ in the simulated SPL noise. We also depict $16\%$ and $84\%$ quantiles $Q_{0.16} \left[ {\Delta O}(\chi,L) \right]$ and $Q_{0.84} \left[ {\Delta O}(\chi,L) \right]$ that define the error bars. The first observation is that TEM compensates for errors even with the smallest nontrivial bond dimension $\chi = 2$: for $\varepsilon = 0.022$ and $\varepsilon = 0.002$ the systematic errors are centered around $0.2$ and $0.017$, respectively, in contrast to the unmitigated errors $1 - \exp\big( -\sum_{l=1}^L \sum_{\lambda_{\boldsymbol{\alpha}}^{(l)} \in {\rm SPL}({\cal N}_l)} \lambda_{\boldsymbol{\alpha}}^{(l)} \big) \approx 1.0$ and $0.6$. The reason is that TEM provides a reasonable collection of rescaling coefficients for Pauli strings even with a modest bond dimension. TEM with bond dimension $\chi = 2$ provides the exact correction of the global depolarizing noise~\cite{filippov_scalable_2023}. The second observation is that the systematic error quickly (exponentially) decreases with the increase of TEM bond dimension and exhibits a transition to a much slower decrease (barely visible within the error bars) at about the threshold bond dimension $\chi^{\ast} \approx 0.5 L^2 = 450$. The values of the systematic error after passing the threshold value are in agreement with the coefficient $a=0.6$ in Eq.~\eqref{TEM-systematic-c-coefficient}.

\begin{figure}
\includegraphics[width=0.5\textwidth]{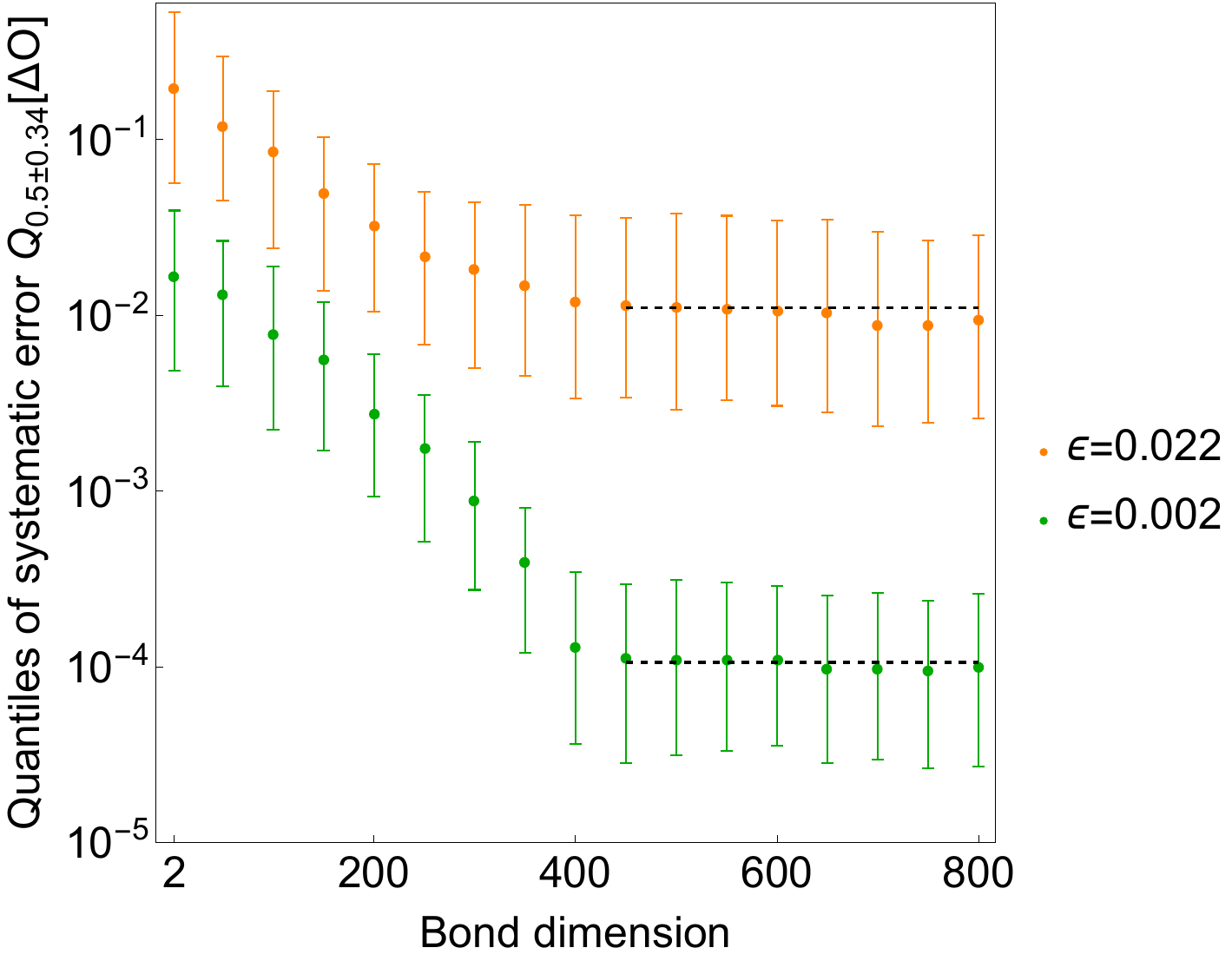}
\caption{\label{figure-bd-threshold} Compression-induced TEM systematic errors in a Clifford circuit with brickwork {\sc cnot} topology and random single-qubit gates, $N=30$ qubits, $L=30$ layers, and different error rates per qubit per layer in the simulated SPL noise. The median error $Q_{0.5}[\Delta O(\chi,L)]$ and the quantiles $Q_{0.16} \left[ {\Delta O}(\chi,L) \right]$ and $Q_{0.84} \left[ {\Delta O}(\chi,L) \right]$ are estimated with $10^3$ random stabilizer observables. Dotted lines depict the quantity $0.6 \sum_{l=1}^L \sum_{\lambda_{\boldsymbol{\alpha}}^{(l)} \in {\rm SPL}({\cal N}_l)} (\lambda_{\boldsymbol{\alpha}}^{(l)})^2$. The leftmost point in dotted lines depicts the approximate threshold bond dimension $\chi^{\ast} \approx 0.5 L^2 = 450$.}
\end{figure}

\section{Threshold bond dimension in TEM} \label{appendix-TEM-threshold-bond-dimension}

In Sec.~\ref{section-tem-systematic-error}, we introduced the functional ${\cal T}_1[\bullet]$ that returns the first-order Taylor polynomial of $\bullet$ with respect to variables $\lambda_{\boldsymbol{\alpha}}^{(l)} \in \cup_{l=1}^L {\rm SPL}({\cal N}_l)$. Upon application to the noise mitigation map ${\cal M}_L^{\dag (\chi)}$ in TEM, this functional provides a simpler tensor network that represents a sum of local propagated noises, ${\cal T}_1[{\cal M}_L^{\dag (\chi)}] = {\rm Id} + \sum_l \sum_q \big[ (\widetilde{\cal N}_l^{[q] \dag})^{-1} - {\rm Id} \big] + \sum_l \sum_{\langle q_1, q_2 \rangle} \big[ (\widetilde{\cal N}_l^{\langle q_1, q_2 \rangle \dag})^{-1} - {\rm Id} \big]$, instead of the product ${\cal M}_L^{\dag (\chi)} = \overleftarrow{\bigcirc}_l \Big( \bigcirc_q (\widetilde{\cal N}_l^{[q] \dag})^{-1} \bigcirc_{\langle q_1, q_2 \rangle} (\widetilde{\cal N}_l^{\langle q_1, q_2 \rangle \dag})^{-1} \Big)$ in the exact Eq.~\eqref{Schrodinger-to-Heisenberg-TME}. 

\begin{figure}
\includegraphics[width=0.5\textwidth]{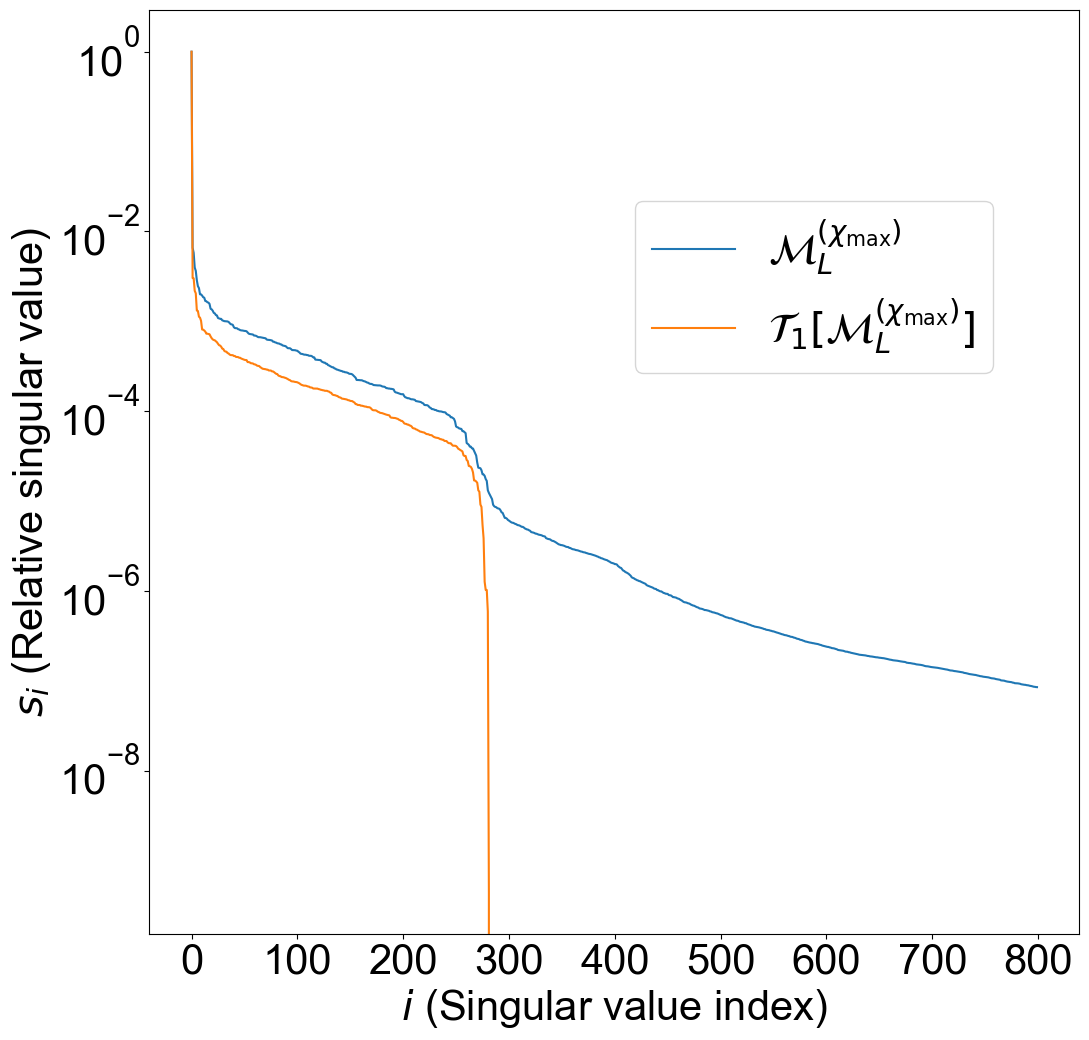}
\caption{\label{figure-bd-threshold-explained} Relative singular values in the central link ($\ell = \lfloor N/2 \rfloor$) of the tensor-network maps ${\cal M}_L^{(\chi_{\max})}$ and its first-order expansion ${\cal T}_1 [{\cal M}_L^{(\chi_{\max})}]$ w.r.t. the relaxation rates in the SPL noise with error rate per qubit per layer $\varepsilon = 0.002$ in a Clifford circuit with brickwork {\sc cnot} topology and random single-qubit gates, $N=30$ qubits, $L=30$ layers. The maximum bond dimension $\chi_{\max} = 800$. Relative singular values in the central link of ${\cal T}_1 [{\cal M}_L^{(\chi_{\max})}]$ vanish for $i > \chi^{\ast}_{\ell = 15} = 282$. The threshold bond dimension $\chi^{\ast} = \max_{\ell} \chi_{\ell}^{\ast} = 415$ is attained at link $\ell^{\ast} = 17$.
}
\end{figure}

Consider the singular values $\{ \mathit{\Sigma}_i^{(\ell)} \}_{i=1}^{\chi}$ in the link $\ell$, for the MPO ${\cal M}_L^{(\chi)}$ and the same link for the MPO ${\cal T}_1 [{\cal M}_L^{(\chi)}]$. The larger the number of nonzero singular values, the more operator entanglement between the two parts of the MPO the link can accommodate. In Fig.~\ref{figure-bd-threshold-explained}, we illustrate the relative singular values $s_i^{(\ell)} = \mathit{\Sigma}_i^{(\ell)} / \mathit{\Sigma}_1^{(\ell)}$ (arranged in the descending order) in the central link ($\ell = \lfloor N/2 \rfloor$) for the example of a Clifford circuit with brickwork {\sc cnot} topology and random single-qubit gates, $N=30$ qubits, $L=30$ layers, and SPL noise with error rate density $\varepsilon = 0.002$. The relative singular values $s_i^{(\ell)}$ of ${\cal T}_1 [{\cal M}_l^{(\chi_{\max})}]$ in the link $\ell$ vanish after the index overcomes some critical index $\chi_{\ell}^{\ast}$. The highest of these critical indices among all the links is the threshold bond dimension $\chi^{\ast} = \max_{\ell} \chi_{\ell}^{\ast}$. With $\chi = \chi^{\ast}$, the map ${\cal T}_1 [{\cal M}_L^{(\chi^{\ast})}] = {\cal T}_1[{\cal M}_L^{(\chi_{\max})}]$ is reproduced exactly, whereas the map ${\cal M}_L^{(\chi^{\ast})} \approx {\cal M}_L^{(\chi_{\max})}$ is just a rather good approximation that fully compensates for all errors linear in the SPL relaxation rates ($\propto \lambda_{\boldsymbol{\alpha}}^{(L)}$) and the residual second-order errors are given by Eq.~\eqref{TEM-systematic-c-coefficient}. Note that ${\cal M}_L^{(\chi^{\ast})} \neq {\cal T}_1 [{\cal M}_L^{(\chi^{\ast})}]$ as it is clearly visible from the difference in the largest $\chi^{\ast}$ singular values in Fig.~\ref{figure-bd-threshold-explained}. Rather, the range $[1,\chi^{\ast}]$ in singular-value indices serves as a proper bandwidth, within which the approximation error decreases much quicker than beyond it. This bandwidth specifies the resources that are worth investing to get a reasonably small systematic error (the best quality/resources ratio).

\begin{figure}
\includegraphics[width=0.5\textwidth]{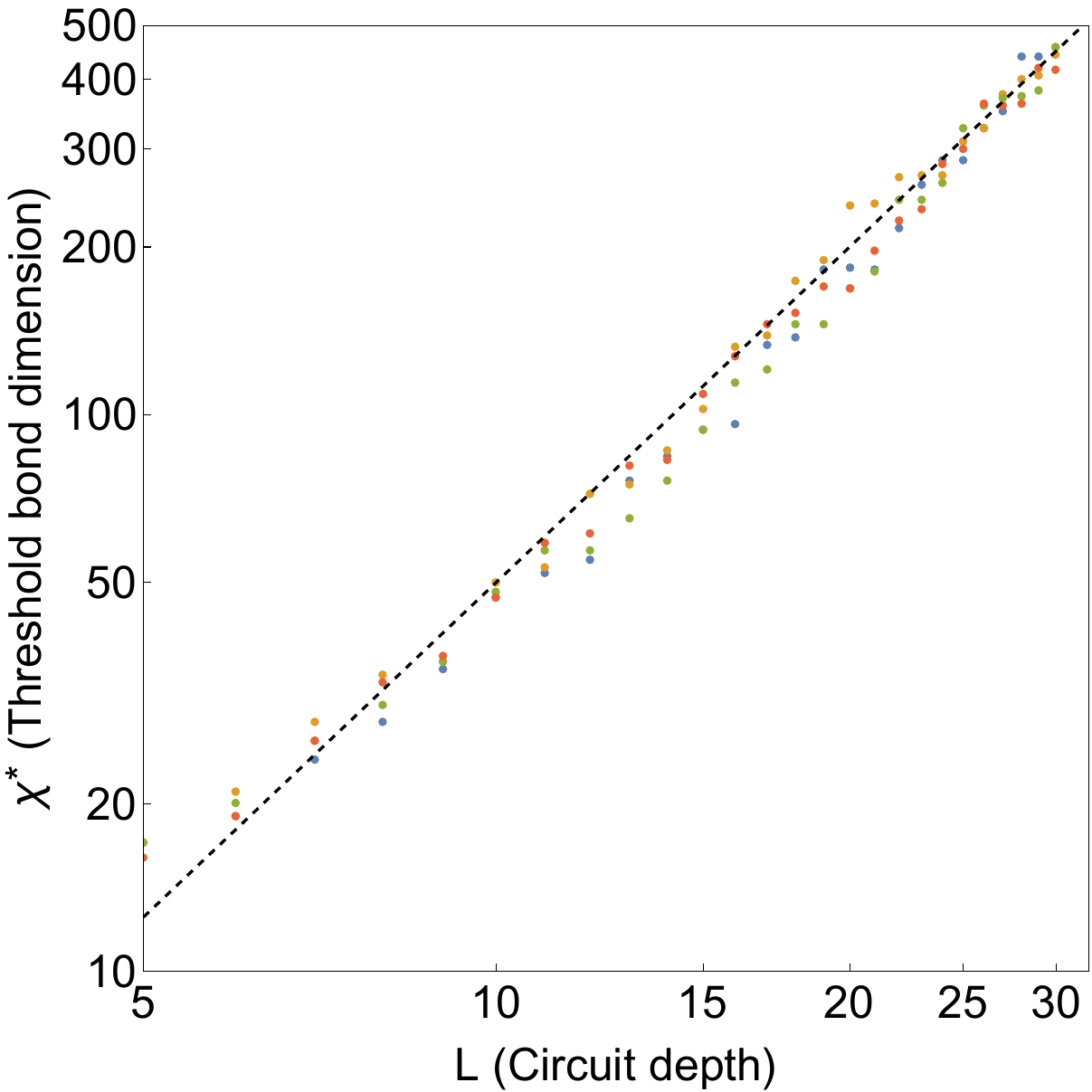}
\caption{\label{figure-bd-threshold-scaling} Threshold bond dimension vs circuit depth for various $30$-qubit noisy Clifford circuits with brickwork {\sc cnot} topology and random single-qubit gates (points); the fitting function $\chi^{\ast}(L) = 0.5 L^2$ (dotted line).
}
\end{figure}

We estimate the scaling of the threshold bond dimension $\chi^{\ast}$ via error propagation. At a depth $L$, the causal cone of any link in the exact MPO ${\cal M}_l$ has $\propto L^2$ local noisy maps ${\cal N}_l^{[q]}$ and ${\cal N}_l^{\langle q_1, q_2 \rangle}$ in its area. $\chi^{\ast}$ is the bond dimension of the first-order polynomial ${\rm Id} + \sum_l \sum_q {\cal U}_{\geq l} \circ ({\cal N}_l^{[q] \dag})^{-1} \circ {\cal U}_{\geq l}^{\dag}  + \sum_l \sum_{\langle q_1, q_2 \rangle} {\cal U}_{\geq l} \circ ({\cal N}_l^{\langle q_1, q_2 \rangle \dag})^{-1} \circ {\cal U}_{\geq l}^{\dag}$ for ${\cal N}_l^{[q]}$ and ${\cal N}_l^{\langle q_1, q_2 \rangle}$ within the area. Since some of the terms remain local after propagation ${\cal U}_{\geq l} \bullet {\cal U}_{\geq l}^{\dag}$ and do not contribute to the bond dimension, we can only conclude that $\chi^{\ast} \propto L^2$ and the proportionality coefficient is to be determined numerically. In Fig.~\ref{figure-bd-threshold-scaling}, we depict the threshold bond dimensions for various noisy Clifford circuits with brickwork {\sc cnot} topology and random single-qubit gates ($N=30$, $L \leq 30$) and conclude that $\chi^{\ast} \approx 0.5 L^2$. If $1 \ll N \ll L$, then one can expect a slower scaling $\chi^{\ast} \propto NL$ because of the causal cone transformation into a stripe shape. The scalings are valid until the bond dimension starts approaching the exact bond dimension $\chi_{\rm exact}$: for a general $N$-qubit MPO $\chi_{\rm exact} = 16^{\lfloor N/2 \rfloor}$, whereas for an $N$-qubit Pauli-diagonal MPO (including the noise-mitigation map in Clifford circuits with the SPL noise) $\chi_{\rm exact} = 4^{\lfloor N/2 \rfloor}$. For the sake of completeness, we present here the proof of the latter statement.

\begin{proposition} \label{proposition-tem-mpo-clifford-bd}
Let ${\cal M}$ be the noise-mitigation map for an $N$-qubit Clifford circuit with an arbitrary Pauli noise. Then ${\cal M}$ has an exact matrix-product-operator representation ${\cal M} = \sum_{a_1,\ldots,a_{N-1}} {\cal A}^{[0]}_{a_1} \otimes {\cal A}^{[1]}_{a_1 a_2} \otimes \cdots \otimes {\cal A}^{[N-2]}_{a_{N-2} a_{N-1}} \otimes {\cal A}^{[N-1]}_{a_{N-1}}$ with dimensions $|\{a_k\}| \leq \min(4^k,4^{N-k})$. The maximum bond dimension does not exceed $4^{\lfloor N/2 \rfloor}$.
\end{proposition}
\begin{proof}
A layer of Pauli noise has the form ${\cal N}_l [\bullet] = \sum_{\boldsymbol{\alpha}} p_{\boldsymbol{\alpha}}^{(l)} \sigma_{\boldsymbol{\alpha}} \bullet \sigma_{\boldsymbol{\alpha}}$, where $\{p_{\boldsymbol{\alpha}}^{(l)}\}_{\boldsymbol{\alpha}}$ is a probability distribution, $\sigma_{\boldsymbol{\alpha}} = \bigotimes_{m=0}^{N-1} \sigma_{\alpha_m}$, and $\boldsymbol{\alpha} = (\alpha_0,\ldots,\alpha_{N-1})$. The inverse map is a Pauli map ${\cal N}_l^{-1} [\bullet] = \sum_{\boldsymbol{\alpha}} q_{\boldsymbol{\alpha}}^{(l)} \sigma_{\boldsymbol{\alpha}} \bullet \sigma_{\boldsymbol{\alpha}}$, where $\{q_{\boldsymbol{\alpha}}^{(l)}\}_{\boldsymbol{\alpha}}$ is a real quasi-probability distribution containing negative elements. Since the initial map ${\cal M}_{0} = {\rm Id}$ is the identity transformation, TEM for the first layer yields ${\cal M}_1 = {\cal N}_1^{-1}$, i.e., ${\cal M}_1$ is a Pauli map. 

Let us prove by induction that ${\cal M}_l$ is a Pauli map for all $l$. Assume that ${\cal M}_{l-1}$ is a Pauli map for some $l$, i.e., ${\cal M}_{l-1} = \sum_{\boldsymbol{\beta}} r_{\boldsymbol{\beta}}^{(l-1)} \sigma_{\boldsymbol{\beta}} \bullet \sigma_{\boldsymbol{\beta}}$, then
\begin{eqnarray}
{\cal M}_{l}[\bullet] &=& ({\cal N}_l^{-1} \circ {\cal U}_l^{\dag} \circ {\cal M}_{l-1} \circ {\cal U}_l) [\bullet] \nonumber\\
&=&  \sum_{\boldsymbol{\alpha},\boldsymbol{\beta}} q_{\boldsymbol{\alpha}}^{(l)} r_{\boldsymbol{\beta}}^{(l-1)} \sigma_{\boldsymbol{\alpha}} U_l^{\dag} \sigma_{\boldsymbol{\beta}} U_l \bullet U_l^{\dag} \sigma_{\boldsymbol{\beta}} U_l \sigma_{\boldsymbol{\alpha}} \nonumber\\
&=&  \sum_{\boldsymbol{\alpha},\boldsymbol{\beta}} q_{\boldsymbol{\alpha}}^{(l)} r_{\boldsymbol{\beta}}^{(l-1)} \sigma_{\boldsymbol{\varphi}(\boldsymbol{\alpha},\boldsymbol{\beta},l)} \bullet \sigma_{\boldsymbol{\varphi}(\boldsymbol{\alpha},\boldsymbol{\beta},l)},
\end{eqnarray}
where the Pauli string $\sigma_{\boldsymbol{\varphi}(\boldsymbol{\alpha},\boldsymbol{\beta},l)}$ is a product of Pauli strings $\sigma_{\boldsymbol{\alpha}}$ and $U_l^{\dag} \sigma_{\boldsymbol{\beta}} U_l$. The latter is a Pauli string because $U_l$ is a Clifford unitary layer. Therefore, ${\cal M}_l = \sum_{\boldsymbol{\varphi}} r_{\boldsymbol{\varphi}}^{(l)} \sigma_{\boldsymbol{\varphi}} \bullet \sigma_{\boldsymbol{\varphi}}$ is a Pauli map for all $l$. 

The matrix product representation of the unitary transformation $\sigma_{\boldsymbol{\varphi}} \bullet \sigma_{\boldsymbol{\varphi}}$ has trivial bond dimension $1$ because $\sigma_{\boldsymbol{\varphi}} = \bigotimes_{m=0}^{N-1} \sigma_{\varphi_m}$. For this reason, the MPO for ${\cal M}_l$ follows from the matrix-product-state (MPS) representation for the tensor $r_{\boldsymbol{\varphi}}^{(l)} \equiv r_{\varphi_0,\ldots,\varphi_{N-1}}^{(l)}$ with the formal physical dimension $|\{\varphi_m\}| = 4$ (the number of single-qubit Pauli operators including the identity operator). Any $N$-order tensor $r_{\varphi_0,\ldots,\varphi_{N-1}}^{(l)}$ with physical dimension $d$ in each leg is known to have an exact MPS representation $r_{\varphi_0,\ldots,\varphi_{N-1}}^{(l)} = \sum_{a_1,\ldots,a_{N-1}} (r^{[0]})_{a_1}^{\varphi_0} \otimes (r^{[1]})_{a_1 a_2}^{\varphi_1} \otimes \cdots \otimes (r^{[N-2]})_{a_{N-2} a_{N-1}}^{\varphi_{N-2}} \otimes (r^{[N-1]})_{a_{N-1}}^{\varphi_{N-1}}$ with dimensions $|\{a_k\}| \leq \min(d^k,d^{N-k})$ \cite{schollwock-2011}. Substituting $d=4$ concludes the proof.
\end{proof}

\section{Systematic error in TEM below the threshold bond dimension} \label{appendix-TEM-systematic}

For bond dimensions $\chi < \chi^{\ast}$ the TEM systematic error can be evaluated based on the truncated singular values. Figure \ref{figure-sv-12q} illustrates typical spectra of relative singular values in the central MPO link of ${\cal M}_L^{(\chi_{\rm exact})}$, where $\chi_{\rm exact} = 4096$ for a Clifford circuit with the SPL noise according to Proposition~\ref{proposition-tem-mpo-clifford-bd}. Extensive numerical experiments suggest that the relative singular values decay exponentially in between the two horizontal lines, the upper one being the regularized Euclidean norm of the relaxation rates $\Lambda_2 \equiv \Big[ \frac{1}{L} \sum_{l=1}^L \sum_{\lambda_{\boldsymbol{\alpha}}^{(l)} \in {\rm SPL}({\cal N}_l)} (\lambda_{\boldsymbol{\alpha}}^{(l)})^2 \Big]^{1/2} \approx \sqrt{N/2} \, \varepsilon/3$ and the lower one being the median of all the relaxation rates $\Lambda_1 \equiv Q_{0.5}\Big( \cup_{l=1}^L \cup_{\lambda_{\boldsymbol{\alpha}}^{(l)} \in {\rm SPL}({\cal N}_l)} \lambda_{\boldsymbol{\alpha}}^{(l)} \Big) \approx \varepsilon / 24$ if one-qubit (two-qubit) relaxation rates are normally distributed with mean and standard deviation $\varepsilon/12$ ($\varepsilon/36$). The singular values experience a steep decrease as soon as they cross the value $\Lambda_1$, which is a manifestation of the critical bond dimension $\chi_{\ell = \lfloor N/2 \rfloor }^{\ast}$ for the central link. Therefore, for bond dimensions $\chi < \chi^{\ast}$ we can approximate the relative singular values as
\begin{equation} \label{sv-approximate}
    s_i \approx \left( \frac{\Lambda_1}{\Lambda_2} \right)^{i / \chi^{\ast}} \Lambda_2 \approx \frac{\varepsilon}{3} \sqrt{\frac{N}{2}} \left( \frac{1}{4\sqrt{2N}} \right)^{2i/L^2}.
\end{equation}

\begin{figure}
\includegraphics[width=0.6\textwidth]{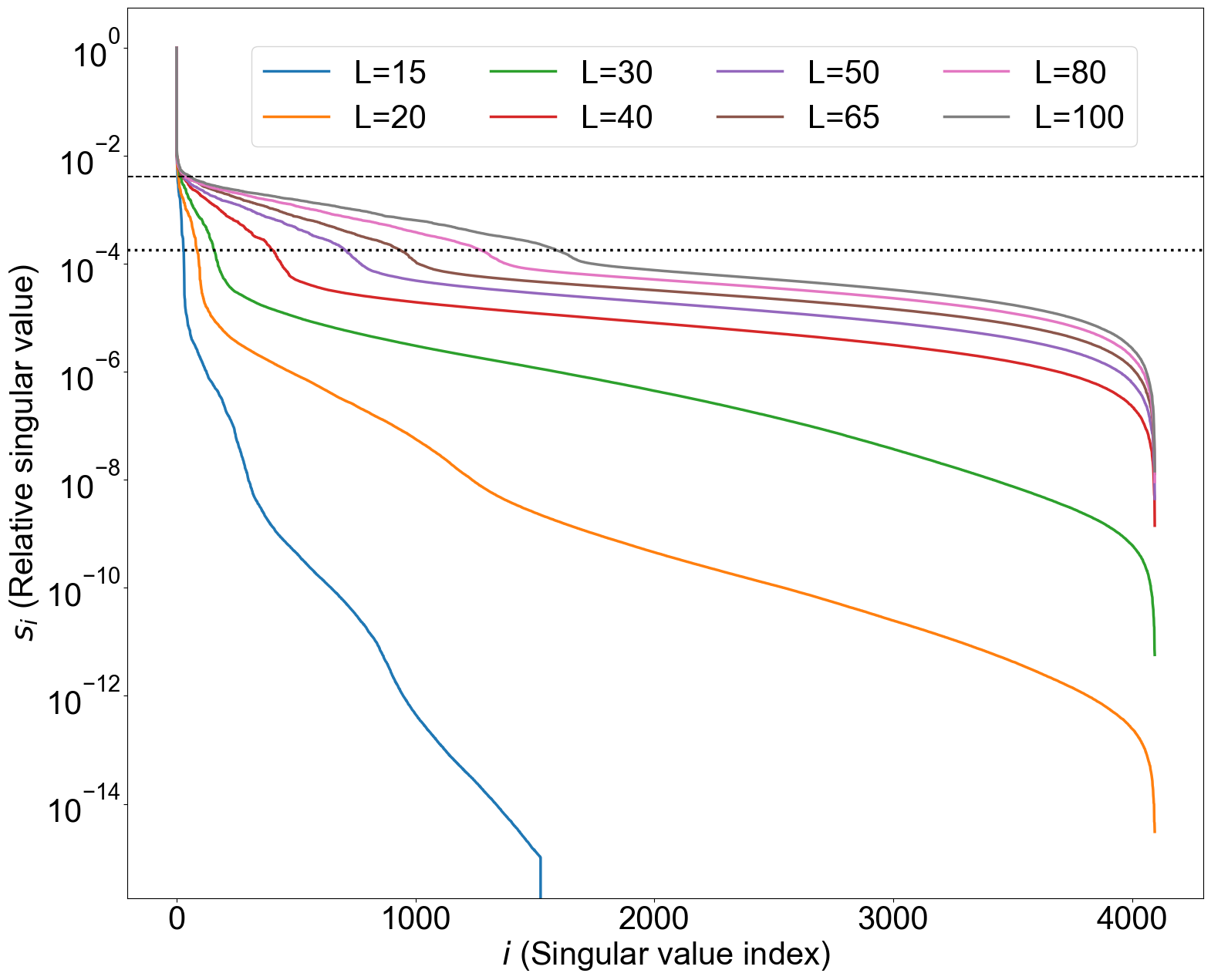}
\caption{\label{figure-sv-12q} Relative singular values in the central link of the exact noise-mitigation map ${\cal M}_L^{(4096)}$ for a $12$-qubit Clifford circuit with brickwork {\sc cnot} topology and random single-qubit gates and the SPL noise with the error rate per qubit per layer $\varepsilon = 0.005$. The decay of singular values is exponential in between the upper horizontal line $\Lambda_2 \equiv \Big[ \frac{1}{L} \sum_{l=1}^L \sum_{\lambda_{\boldsymbol{\alpha}}^{(l)} \in {\rm SPL}({\cal N}_l)} (\lambda_{\boldsymbol{\alpha}}^{(l)})^2 \Big]^{1/2}$ and the lower horizontal line $\Lambda_1 \equiv Q_{0.5}\Big( \cup_{l=1}^L \cup_{\lambda_{\boldsymbol{\alpha}}^{(l)} \in {\rm SPL}({\cal N}_l)} \lambda_{\boldsymbol{\alpha}}^{(l)} \Big)$. 
}
\end{figure}

Suppose we truncate singular values $i > \chi$ for some fixed bond dimension $\chi \in (1,\chi^{\ast})$. The truncated singular values are translated into the compression error, but the rigorous upper bound $\Delta O \leq \|O\|_2 \sqrt{2 \sum_{\rm bonds} \sum_{i = \chi + 1}^{\chi_{\rm exact}} \mathit{\Sigma}_i^2}$ turns out to be too loose to be used in practice \cite{filippov_scalable_2023}. This is because of the exponential growth of both the 2-norm $\|O\|_2$ and the unnormalized singular values $\mathit{\Sigma}_i$. We aim at finding a reasonable estimate of the systematic error $\Delta O$ in terms of the disregarded singular values via numerical analysis. Figure \ref{figure-bd-300} depicts the median systematic error $\Delta O$ originating from the repeated truncation of the smallest singular value while constructing the noise mitigation map ${\cal M}^{(\chi)}_L$ as well as one-, two-, and infinite norms of the vector ${\bf s}_{i > \chi}^{(l)} \equiv \{ s_{i}^{(l)} \}_{i=\chi+1}^{\chi_{\rm exact}}$ for each layer $l \in [1,L]$. The 2-norm $\|{\bf s}_{i > \chi}^{(l)}\|_2$ is in very good agreement with the observed median error, whereas $\|{\bf s}_{i > \chi}^{(l)}\|_1$ significantly overestimates and $\|{\bf s}_{i > \chi}^{(l)}\|_\infty$ significantly underestimates the error. The 2-norm $\| {\bf s}_{i > \chi}^{(1:L)} \|_2$ of the depth-extension vector ${\bf s}_{i > \chi}^{(1:L)} \equiv \Big\{ \{ s_{i}^{(l)} \}_{i=\chi+1}^{\chi_{\rm exact}} \Big\}_{l=1}^{L}$ has a tendency to overestimate $Q_{0.5}[\Delta O(\chi,L)]$. Using the 2-norm $\|{\bf s}_{i > \chi}^{(l)}\|_2$ as an additional error on top of the threshold-bond-dimension error $\Delta O(\chi^{\ast},L) \approx 0.6 \sum_{l=1}^L \sum_{\lambda_{\boldsymbol{\alpha}}^{(l)} \in {\rm SPL}({\cal N}_l)} (\lambda_{\boldsymbol{\alpha}}^{(l)})^2 \approx NL\varepsilon^2/30$ and the approximate expression~\eqref{sv-approximate} for the relative singular values, we obtain the evaluation of TEM's systematic error below the threshold:
\begin{eqnarray}
\Delta O(\chi,L) & \approx & \sqrt{\Delta O(\chi^{\ast},L)^2 + \frac{ \Lambda_1^{2\chi/\chi^{\ast}} \Lambda_2^{2(\chi^{\ast} - \chi)/\chi^{\ast}} - \Lambda_1^2 }{\left( \Lambda_2 / \Lambda_1 \right)^{2/\chi^{\ast}} - 1} } \nonumber\\
& \approx & \sqrt{\Delta O(\chi^{\ast},L)^2 + \frac{ \Lambda_1^{2\chi/\chi^{\ast}} \Lambda_2^{2(\chi^{\ast} - \chi)/\chi^{\ast}} - \Lambda_1^2 }{2 \ln (\Lambda_2 / \Lambda_1) /\chi^{\ast}} } \nonumber\\
& \approx & \sqrt{\left( \frac{NL\varepsilon^2}{30} \right)^2 + \frac{ \varepsilon^2 L^2 }{ 72 \ln(4\sqrt{2N}) } \left[ N \Big( \frac{1}{32N} \Big)^{2\chi / L^2} - \frac{1}{32} \right]}. \label{error-analytic-evaluation}
\end{eqnarray}

\begin{figure}
\includegraphics[width=0.75\textwidth]{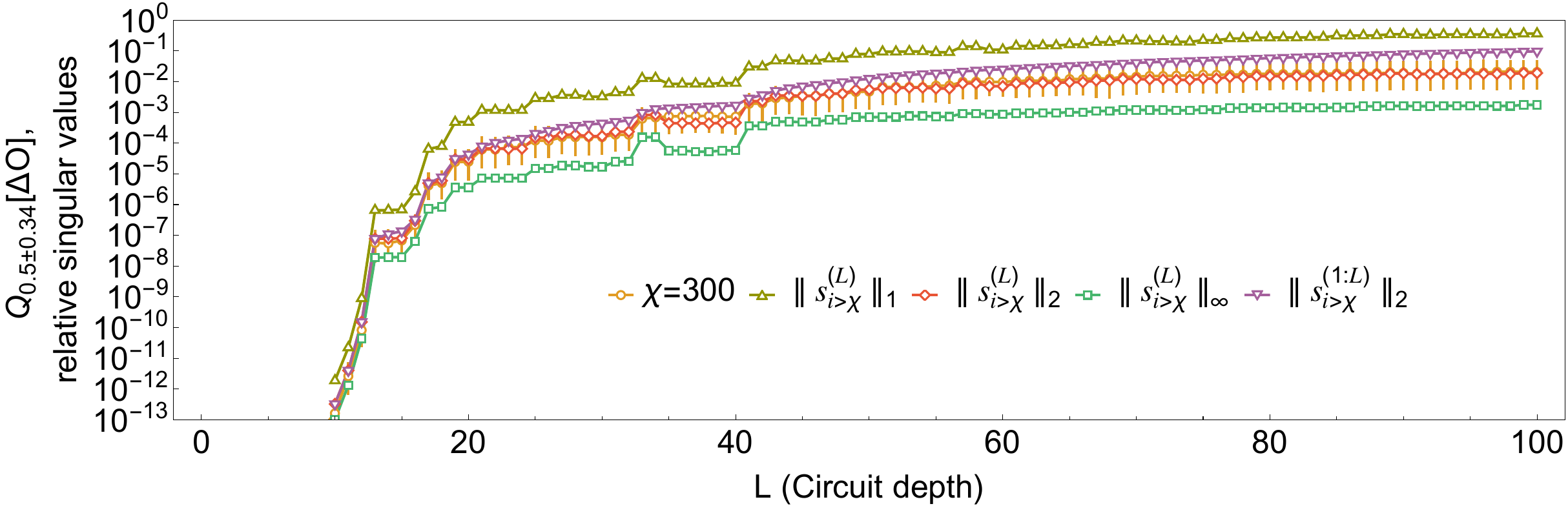}
\caption{\label{figure-bd-300} Median systematic error with error bars $Q_{0.5 \pm 0.34}[\Delta O(\chi,L)]$ in estimating noise-mitigated stabilizer observables via TEM with bond dimension $\chi = 300$ for the 10-qubit example. The systematic error is well reproduced by $\|{\bf s}_{i > \chi}^{(L)}\|_2$ and approximated from above by $\|{s}_{i > \chi}^{(1:L)}\|_2$.}
\end{figure}

\begin{figure}
\includegraphics[width=0.5\textwidth]{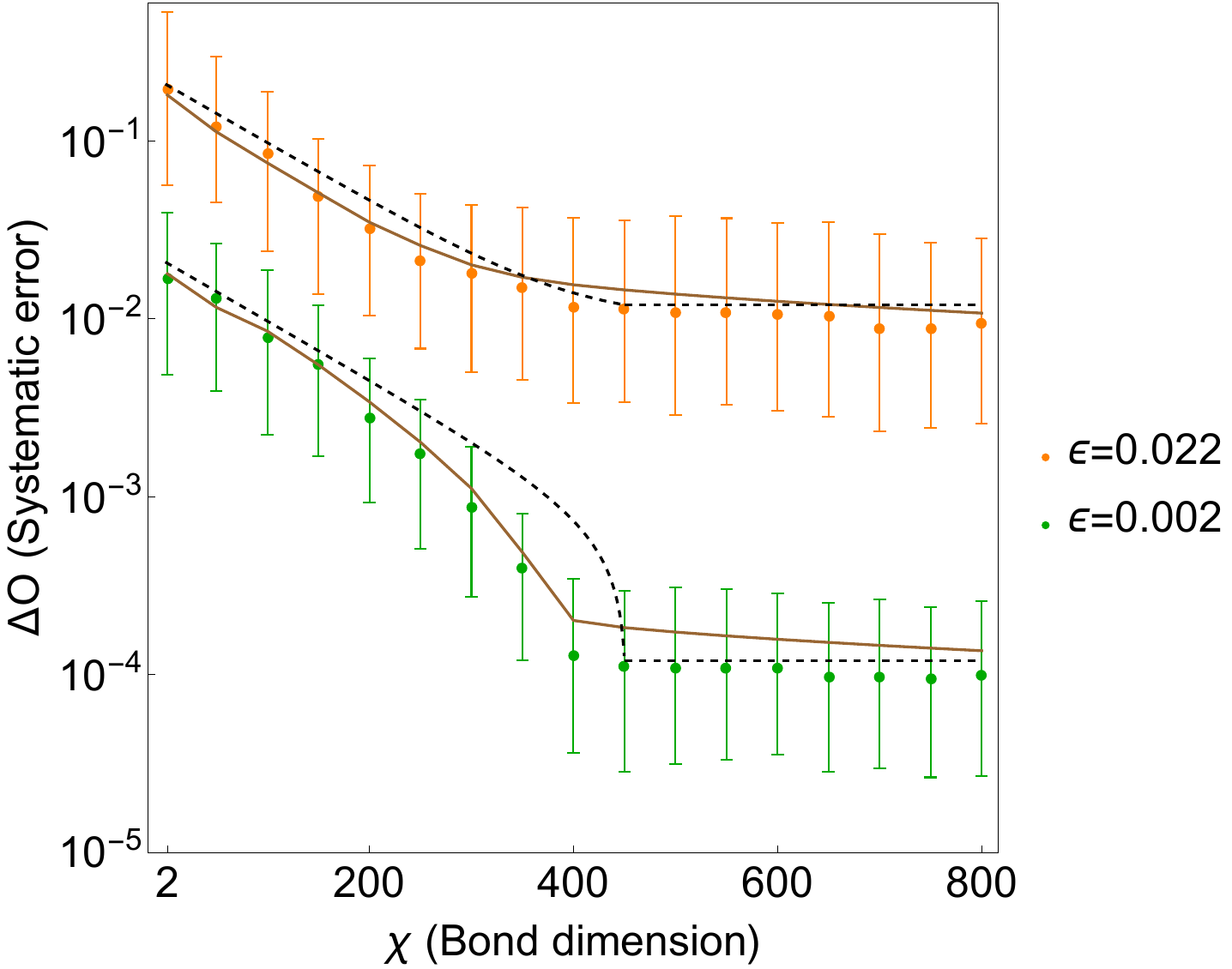}
\caption{\label{figure-error-30x30} Compression-induced TEM systematic error $Q_{0.5 \pm 0.34}[\Delta O(\chi,L)]$ as in Fig.~\ref{figure-bd-threshold} and its approximate evaluations: the data-driven quantifier ${\cal E}(\chi, L)$ of compression quality in Eq.~\eqref{error-estimate} (solid line); the analytical formula~\eqref{error-analytic-evaluation} based on the pattern of singular values in the noise-mitigation MPO (dotted line).}
\end{figure}

Targeting $100 \times 100$ circuits, we depict in Fig.~\ref{figure-scaling-100x100} how the compression-induced systematic error in TEM scales in the bond dimension. The analytical expression~\eqref{error-analytic-evaluation} overestimates the median error for $10^3$ random stabilizer observables but the slope is in agreement with the numerical experiment. To get the error of several percent one would need to use bond dimension $\chi \sim 10^3$ for the technologically achievable density of errors $\varepsilon = 0.0016$~\cite{mckay_benchmarking_2023}.

\begin{figure}
\includegraphics[width=0.55\textwidth]{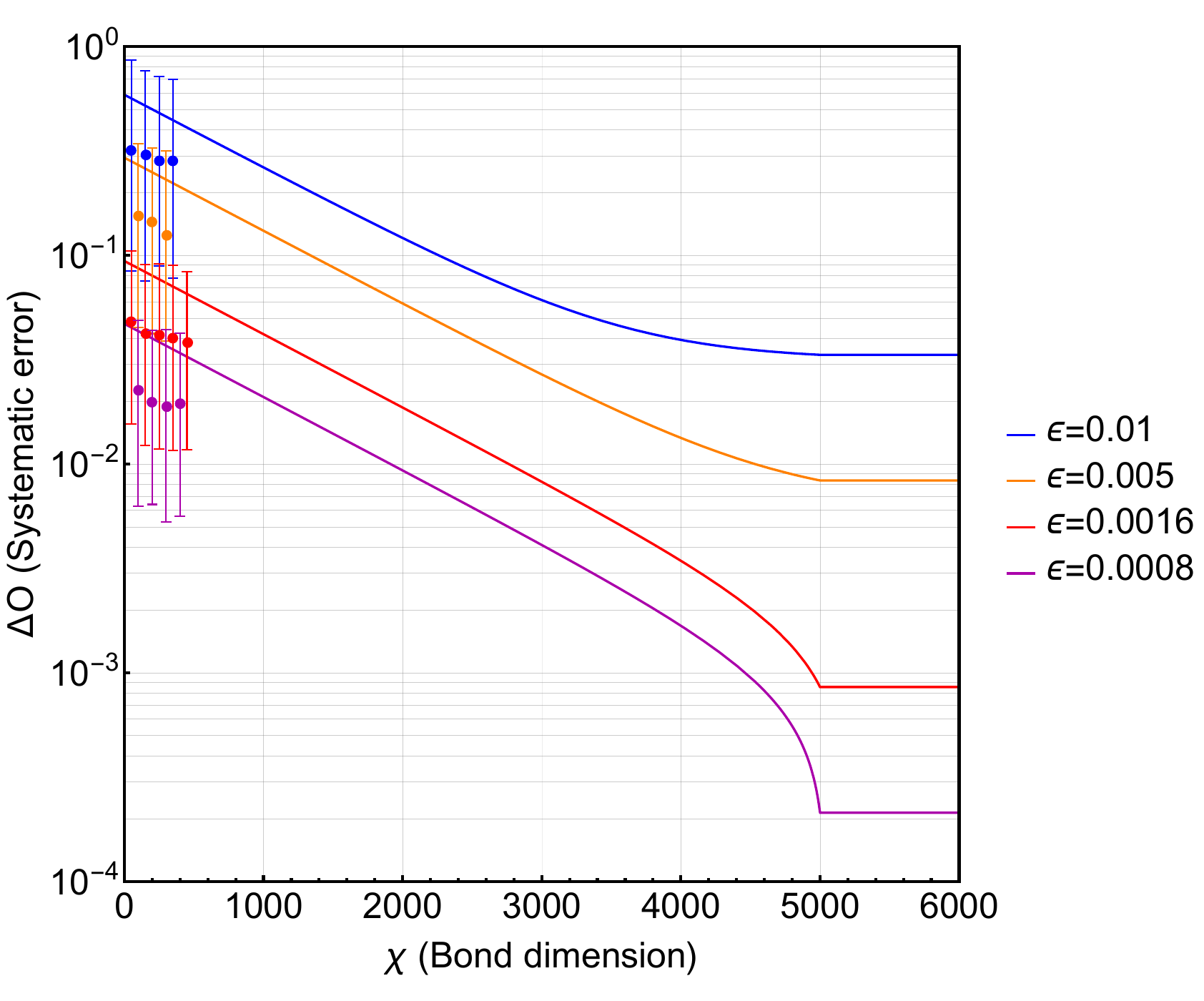}
\caption{\label{figure-scaling-100x100} Compression-induced TEM systematic error $Q_{0.5 \pm 0.34}[\Delta O(\chi,L)]$ (points) and its analytical evaluation~\eqref{error-analytic-evaluation} (lines) for high Pauli-weight observables in a dense stabilizer circuit of depth $L = 100$ and width $N = 100$ qubits (4950 {\sc cnot} gates) for different values of the error per layer per qubit. }
\end{figure}

\section{Data-driven estimation of TEM systematic error}

Since $\chi_{\rm exact} = 4^{\lfloor N/2 \rfloor}$ grows exponentially in the number of qubits $N$, there is no feasible way to calculate the whole spectrum of singular values if $N \gtrsim 15$. Therefore, the error in approximating ${\cal M}_L^{\dag}$ by ${\cal M}_L^{\dag (\chi)}$ should be inferred by another method. Luckily, the solution comes from the very construction process of the approximation ${\cal M}_L^{\dag (\chi)}$ via iterations ${\cal M}_l^{\dag (\chi)} = {\cal C}_{\chi} \big( {\cal U}_l \circ ({\cal N}_l^{-1})^{\dag} \circ {\cal M}_{l-1}^{\dag (\chi)} \circ {\cal U}_l^{\dag} \big)$, $l=1,\ldots,L$ with ${\cal M}_{0} = {\rm Id}$. The difference between the maps ${\cal U}_l \circ ({\cal N}_l^{-1})^{\dag} \circ {\cal M}_{l-1}^{\dag (\chi)} \circ {\cal U}_l^{\dag}$ and ${\cal M}_l^{\dag (\chi)}$ is due to the compression ${\cal C}_{\chi}$. As both maps have MPO representations, their Frobenius norms (i.e., 2-norms) and the overlap $\braket{{\cal M}_l^{\dag (\chi)} | {\cal U}_l \circ ({\cal N}_l^{-1})^{\dag} \circ {\cal M}_{l-1}^{\dag (\chi)} \circ {\cal U}_l^{\dag} }$ are relatively easy to compute at each iteration [the complexity grows as $O(N\chi^2)$], so that the distance $\|{\cal M}_l^{\dag (\chi)} - {\cal U}_l \circ ({\cal N}_l^{-1})^{\dag} \circ {\cal M}_{l-1}^{\dag (\chi)} \circ {\cal U}_l^{\dag} \|_2$ is readily available. The relative distance $\|{\cal M}_l^{\dag (\chi)} - {\cal U}_l \circ ({\cal N}_l^{-1})^{\dag} \circ {\cal M}_{l-1}^{\dag (\chi)} \circ {\cal U}_l^{\dag} \|_2 / \|{\cal M}_l^{\dag (\chi)} \|_2$ approximates the 2-norm of the truncated relative singular values of the map ${\cal U}_l \circ ({\cal N}_l^{-1})^{\dag} \circ {\cal M}_{l-1}^{\dag (\chi)} \circ {\cal U}_l^{\dag}$ at iteration $l$, and in this sense it approximates $\|{\bf s}_{i > \chi}^{(l)}\|_2$ (as if ${\cal U}_l \circ ({\cal N}_l^{-1})^{\dag} \circ {\cal M}_{l-1}^{\dag (\chi)} \circ {\cal U}_l^{\dag}$ were an exact map ${\cal M}_l^{\dag (\chi_{\rm exact})}$). Therefore, we can approximate the selected error quantifier $\| {s}_{i > \chi}^{(1:L)} \|_2 \approx {\cal E}(\chi,L)$ while constructing the noise-mitigation tensor network via
\begin{equation} \label{error-estimate}
    {\cal E}(\chi,L) = \sqrt{ \sum_{l=1}^L \frac{\|{\cal M}_l^{\dag (\chi)} - {\cal U}_l \circ ({\cal N}_l^{-1})^{\dag} \circ {\cal M}_{l-1}^{\dag (\chi)} \circ {\cal U}_l^{\dag} \|_2^2}{\|{\cal M}_l^{\dag (\chi)} \|_2^2} }.
\end{equation}
The benefit of Eq.~\eqref{error-estimate} is that ${\cal E}(\chi,L)$ can be calculated for any compression possible (not only for the SVD compression but also for the variational one, for example, where the truncated singular values are not available in principle). Figure \ref{figure-error-30x30} illustrates a good agreement between the median systematic error $Q_{0.5}[\Delta O(\chi,L)]$ in estimating stabilizer observables and the error estimate ${\cal E}(\chi,L)$ inferred during the TEM map construction for different choices of the bond dimension $\chi$. 

\section{Resource estimates for TEM} \label{appendix-TEM-summary}

Given a finite allocated time ${\rm T}$ for a quantum computation, the available number of circuit executions is $M = {\rm T}/(L \tau_{\rm l} + \tau_{\rm m} + \tau_{\rm delay})$, where $L$ is the circuit depth, $\tau_{\rm l}$ is the duration of single-qubit and two-qubit gates in one layer, $\tau_{\rm m}$ is the measurement duration, and $\tau_{\rm delay}$ is the circuit delay between sequential circuit executions ($\tau_{\rm l} \sim 0.6$ $\mu$s, $\tau_{\rm m} \sim 0.8$ $\mu$s, and foreseeable $\tau_{\rm delay} = 0.5$ ms in experiments with superconducting qubits~\cite{kim_evidence_2023}). As in the resource analysis for ZNE, we can assume that one noise characterization suffices if the noise instability is minor. If this is the case, then roughly the same time ${\rm T}$ is allowed for classical contraction of the TEM map. Should the noise be recharacterized $n_{\rm rec}$ times during the allocated period ${\rm T}$ for quantum computation, the time bins for the classical contraction of the noise-specific TEM maps would become shorter and last roughly ${\rm T} / n_{\rm rec}$. Since the tensor-network compression cost scales cubically in the bond dimension, the affordable bond dimension $\chi$ in TEM is evaluated through the equation $\chi^3 = c_{\rm b} {\rm P}{\rm T} / ( n_{\rm rec} N L)$, where ${\rm P}$ quantifies the available classical computational power in floating point operations per second (FLOPS), $c_{\rm b} = \chi_{\rm test}^3 N_{\rm test} L_{\rm test} / ({\rm P}_{\rm test} {\rm T}_{\rm test})$ is the proportionality coefficient inferred from the benchmark contractions on classical computer of similar architecture with power ${\rm P}_{\rm test}$ that take time ${\rm T}_{\rm test}$ for a test bond dimension $\chi_{\rm test}$. Using a conventional SVD-compression algorithm without any optimization, we get the estimate $c_{\rm b}^{-1} = 3 \times 10^5$ floating point operations. Once the TEM map is built and contracted with the observable, the classical cost of summing over measurement shots is negligible. The total estimation error $\delta$ for a non-local observable $O$ in a dense $N \times L$ circuit is then evaluated through
\begin{eqnarray*}
\delta^2 &=& \left( \frac{\exp (\varepsilon N L / 2)}{\sqrt{{\rm T}/(L \tau_{\rm l} + \tau_{\rm m} + \tau_{\rm delay})} } \right)^2 + \left( \dfrac{1}{2} \varepsilon N \sqrt{L} \Theta \right)^2 \nonumber\\
&& + \left\{ \begin{array}{ll}
   \left( \frac{NL\varepsilon^2}{30} \right)^2 & \text{if~} \sqrt[3]{c_{\rm b} {\rm P} {\rm T} / ( n_{\rm rec} N L) } \geq \frac{L^2}{2}, \\
   \left( \frac{NL\varepsilon^2}{30} \right)^2 + \frac{ \varepsilon^2 L^2 }{ 72 \ln(4\sqrt{2N}) } \left[ N \big( \frac{1}{32N} \big)^{2 \sqrt[3]{c_{\rm b} {\rm P} {\rm T} / ( n_{\rm rec} N L)} / L^2} - \frac{1}{32} \right]  & \text{otherwise},
\end{array} \right.
\end{eqnarray*}
where $\varepsilon$ is the average error rate in the circuit and $\Theta$ is the standard deviation in the relative error density fluctuations. The time needed for the noise characterization is excluded from this resource analysis as the number of unique layers can vary significantly from problem to problem (for example, it equals 2 for the Floquet dynamics considered in Sec.~\ref{section-prospects-advantage}).

\section{Cardinality of the anticommutant set of a Pauli string} \label{appendix-anticommutativity}

Every nontrivial $N$-qubit Pauli string $\sigma_{\boldsymbol{\alpha}}$ [with $\boldsymbol{\alpha} = (\alpha_0, \ldots, \alpha_{N-1}) \neq \boldsymbol{0}$] anticommutes with exactly $4^N/2$ of all $4^N$ $N$-qubit Pauli strings.
To show this, let us construct a bijection between the set ${\cal C}_{\boldsymbol{\alpha}}$ of strings that commute with $\sigma_{\boldsymbol{\alpha}}$ and the set ${\cal A}_{\boldsymbol{\alpha}}$ of strings that anticommute with it, implying an equal amount of both ($|{\cal C}_{\boldsymbol{\alpha}}| = |{\cal A}_{\boldsymbol{\alpha}}| = 4^N/2$).

Let $q$ be a qubit number for which $\alpha_q \neq 0$, and define the $4$-component vector $\boldsymbol{\Omega} = \Big( 0, \alpha_q, \alpha_q \mod 3 + 1, (\alpha_q + 1) \mod 3 + 1 \Big)$ that is a specific ordered permutation of $(0,1,2,3)$.
$\sigma_{\Omega_{i}}$ commutes (anticommutes) with $\sigma_{\alpha_q}$ for $i=0, 1$ ($i=2, 3$). The permutation $\pi: \{0, \ldots, 3 \} \rightarrow \{0, \ldots, 3 \}$ is defined implicitly by $\Omega_{\pi(i)} = i$. A bijection ${\cal C}_{\boldsymbol{\alpha}} \ni \sigma_{\boldsymbol{\beta}}  \leftrightarrow \sigma_{\mathcal{M}(\boldsymbol{\beta})} \in {\cal A}_{\boldsymbol{\alpha}}$ is established through the bijective map 
\begin{equation*}
    \mathcal{M}(\boldsymbol{\beta}) = (\beta_0, \ldots, \beta_{q-1}, \Omega_{(\pi(\beta_q) + 2) \mod 4}, \beta_{q + 1}, \ldots, \beta_{N - 1}).
\end{equation*}
In words, given $\beta_q$, the mapping associates with it a new value by shifting its position in vector $\boldsymbol{\Omega}$, $\pi(\beta_q)$, two positions.
Consequently, if $\sigma_{\beta_{q}}$ commutes (anticommutes) with $\sigma_{\alpha_q}$, $\sigma_{\Omega_{(\pi(\beta_q) + 2) \mod 4}}$ anticommutes (commutes) with it.
Since the (anti)commutativity of two Pauli strings depends on the parity of the number of qubit-wise (anti)commutation relations between them, if $\{ \sigma_{\boldsymbol{\beta}}, \sigma_{\boldsymbol{\alpha}} \} = 0 \Rightarrow \{ \sigma_{\mathcal{M}(\boldsymbol{\beta})}, \sigma_{\boldsymbol{\alpha}} \} \neq 0$, and vice-versa.
Moreover, since applying the mapping a second time returns the original string ($\sigma_{\mathcal{M}(\mathcal{M}(\boldsymbol{\beta}))} = \sigma_{\boldsymbol{\beta}}$), the procedure defines a bijection.

\bibliography{bibliography}

\end{document}